%% file: Truth.tex
\documentclass[runningheads]{llncs}
\usepackage{amsmath}
\usepackage{amscd}
\usepackage{latexsym}
\usepackage{amssymb}
\usepackage{float}
\usepackage{array}
\usepackage{wrapfig}
\usepackage{ulem}
\usepackage{epic}

\newtheorem{fact}{Fact}
\newtheorem{meta-axiom}{Meta-axiom}
\floatstyle{ruled}
\newfloat{Note}{thp}{lop}[section]

\DeclareMathAlphabet{\mathrm}{OT1}{cmr}{m}{n}
\DeclareMathAlphabet{\mathrmbf}{OT1}{cmr}{bx}{n}
\DeclareMathAlphabet{\mathrmit}{OT1}{cmr}{m}{it}
\DeclareMathAlphabet{\mathrmbfit}{OT1}{cmr}{bx}{it}
\DeclareMathAlphabet{\mathsf}    {OT1}{cmss}{m}{n}
\DeclareMathAlphabet{\mathsfbf}  {OT1}{cmss}{bx}{n}
\DeclareMathAlphabet{\mathsfit}  {OT1}{cmss}{m}{sl}
\DeclareMathAlphabet{\mathttbf}{OT1}{cmtt}{bx}{n}

\newenvironment{eg}
{\begin{center}\begin{minipage}{346pt}\textbf{Example:}}
{\end{minipage}\end{center}}

\begin{document}

\title{The Architecture of Truth}
\author{Robert E. Kent}
\institute{Ontologos}
\maketitle

\begin{abstract}
The theory of institutions is framed as an indexed/fibered duality,
where the indexed aspect specifies the fibered aspect.
Tarski represented truth in terms of a satisfaction relation.
The theory of institutions encodes satisfaction as its core architecture in the indexed aspect.
Logical environments enrich this truth architecture by axiomatizing the truth adjunction in the fibered aspect.
The truth architecture is preserved by morphisms of logical environments.
~\footnote{Although not every institution is a logical environment,
each institution has an associated logical environment
defined via the intent of the structures of the institution,
and each institution
is represented by an indexed functor into the structure category of the classification logical environment $\mathtt{Cls}$.}
\end{abstract}

\tableofcontents

\input{structure-FOL}

\input{framework}

\input{parta}

\input{intent}\input{fiber-meets}
\input{partb}

\input{extent}\input{fiber-sums}
\end{document}

%% file: structure-FOL.tex
\newpage
\section{FOL Structures}\label{structures-FOL}

\subsection{Objects}

A ``possible world'' or \emph{structure} for a type language provides an interpretation for all types. 
It traditionally associates sets (domains) with entity types and relations with relation types. 
These associations must respect the various typings. 
The notion of structure introduced here 
(see the IFF Model Theory Ontology)
is framed in terms of the notion of a classification along with a suitable notion of a hypergraph. 
This is a new formulation of structure,
and introduced here for the first time. 
However, these structures can be placed in the traditional perspective. 
These structures and their morphisms form a category called $\mathrmbf{Struc}$.


\begin{figure}
\begin{center}
\begin{tabular}{c}

\setlength{\unitlength}{0.5pt}
\begin{picture}(160,80)(0,0)
\put(0,80){\makebox(0,0){\footnotesize{$\mathrmbfit{var}(\mathrmbf{A})$}}}
\put(160,80){\makebox(0,0){\footnotesize{$\mathrmbfit{typ}(\mathrmbfit{ent}(\mathrmbf{A}))$}}}
\put(0,0){\makebox(0,0){\footnotesize{$\mathrmbfit{var}(\mathrmbf{A})$}}}
\put(160,-18){\makebox(0,0){\footnotesize{$\underset{\textstyle\mathrmbfit{univ}(\mathrmbf{A})}
{\underbrace{\mathrmbfit{inst}(\mathrmbfit{ent}(\mathrmbf{A}))}}$}}}
\put(-10,40){\makebox(0,0)[r]{\small{$=$}}}
\put(170,40){\makebox(0,0)[l]{\scriptsize{$\models_{\mathrmbfit{typ}(\mathrmbfit{ent}(\mathrmbf{A}))}$}}}
\put(80,95){\makebox(0,0){\scriptsize{$\ast$}}}
\put(160,15){\line(0,1){50}}
\put(40,80){\vector(1,0){60}}
\put(80,50){\makebox(0,0){\small{$\tau_{\mathrmbf{A}}=\mathrmbfit{refer}(\mathrmbf{A})$}}}
\put(80,30){\makebox(0,0){\scriptsize{$semidesignation$}}}
\put(260,40){\makebox(0,0)[l]{\footnotesize{$\left.\rule{0pt}{25pt}\right\}\mathrmbfit{ent}(\mathrmbf{A})$}}}
\end{picture}
\\ \\ \\ \\
\setlength{\unitlength}{0.5pt}
\begin{picture}(160,80)(0,0)
\put(0,80){\makebox(0,0){\footnotesize{$\mathrmbfit{typ}(\mathrmbfit{rel}(\mathrmbf{A}))$}}}
\put(160,80){\makebox(0,0){\footnotesize{$\mathrmbfit{sign}(\ast)$}}}
\put(0,-18){\makebox(0,0){\footnotesize{$\underset{\textstyle\mathrmbfit{tuple}(\mathrmbf{A})}
{\underbrace{\mathrmbfit{inst}(\mathrmbfit{rel}(\mathrmbf{A}))}}$}}}
\put(160,0){\makebox(0,0){\footnotesize{$\mathrmbfit{tuple}(\tau_{\mathrmbf{A}})$}}}
\put(-10,40){\makebox(0,0)[r]{\scriptsize{$\models_{\mathrmbfit{rel}(\mathrmbf{A})}$}}}
\put(170,40){\makebox(0,0)[l]{\scriptsize{$\models_{\mathrmbfit{sign}(\tau_{\mathrmbf{A}})}$}}}
\put(85,95){\makebox(0,0){\scriptsize{$\delta_{1}$}}}
\put(85,-15){\makebox(0,0){\scriptsize{$\delta_{0}$}}}
\put(0,15){\line(0,1){50}}
\put(160,15){\line(0,1){50}}
\put(54,80){\vector(1,0){60}}
\put(54,0){\vector(1,0){60}}
\put(80,50){\makebox(0,0){\small{$\delta_{\mathrmbf{A}}=\mathrmbfit{sign}(\mathrmbf{A})$}}}
\put(80,30){\makebox(0,0){\scriptsize{$designation$}}}
\put(-80,40){\makebox(0,0)[r]{\footnotesize{$\mathrmbfit{rel}(\mathrmbf{A})\left\{\rule{0pt}{25pt}\right.$}}}
\put(260,40){\makebox(0,0)[l]{\footnotesize{$\left.\rule{0pt}{25pt}\right\}\mathrmbfit{sign}(\tau_{\mathrmbf{A}})$}}}
\end{picture}
\\ \\
\end{tabular}
\end{center}
\caption{FOL Structure}
\label{fol:structure}
\end{figure}


A structure 
$\mathrmbf{A} = 
{\langle{\mathrmbfit{refer}(\mathrmbf{A}),\mathrmbfit{sign}(\mathrmbf{A})}\rangle}$ 
is a hypergraph of classifications
-- a two dimensional construction (Figure~\ref{fol:structure}) consisting of
a reference semidesignation $\tau_{\mathrmbf{A}} = \mathrmbfit{refer}(\mathrmbf{A})$ and
a signature designation $\delta_{\mathrmbf{A}} = \mathrmbfit{sign}(\mathrmbf{A})$,
where the signature classification of the reference semidesignation is the target classification of the signature designation
$\mathrmbfit{sign}(\tau_{\mathrmbf{A}}) = \mathrmbfit{tgt}(\delta_{\mathrmbf{A}})$.
For convenience of practical reference, 
we introduce intuitive terminology for various subcomponents. 
The source of the signature designation is called the relation classification of $\mathrmbf{A}$ 
and denoted by $\mathrmbfit{rel}(\mathrmbf{A})$. 
The classification of the reference semidesignation is called the entity classification of $\mathrmbf{A}$ 
and denoted by $\mathrmbfit{ent}(\mathrmbf{A})$. 
The set of the reference semidesignation is called the set of logical variables of $\mathrmbf{A}$ 
and denoted by $\mathrmbfit{var}(\mathrmbf{A})$. 
The reference semidesignation and signature and arity~\footnote{The carinality of the arity is the valence.
It is important to make this distinction.}
designations are
\begin{center}
$\footnotesize{\begin{array}{r@{\hspace{5pt}=\hspace{5pt}}l@{\hspace{5pt}:\hspace{5pt}}l}
\tau_{\mathrmbf{A}} 
& \mathrmbfit{refer}(\mathrmbf{A}) 
& \mathrmbfit{var}(\mathrmbf{A}) \looparrowright \mathrmbfit{ent}(\mathrmbf{A})
\\
\delta_{\mathrmbf{A}} 
& \mathrmbfit{sign}(\mathrmbf{A}) \;\;\, = {\langle{\delta_{0},\delta_{1}}\rangle} 
& \mathrmbfit{rel}(\mathrmbf{A}) \rightrightarrows \mathrmbfit{sign}(\tau_{\mathrmbf{A}} )
\\
\alpha_{\mathrmbf{A}} 
& \mathrmbfit{arity}(\tau_{\mathrmbf{A}}) = {\langle{\alpha_{0},\alpha_{1}}\rangle}
& \mathrmbfit{rel}(\mathrmbf{A}) \rightrightarrows \mathrmbfit{sup}(\mathrmbfit{var}(\mathrmbf{A})),
\end{array}}$
\end{center}


Let us define some further common terminology for models. 
The set 
$\mathrmbfit{univ}(\mathrmbf{A}) = \mathrmbfit{inst}(\mathrmbfit{ent}(\mathrmbf{A}))$ 
of entity instances is called the universe (or universal domain) for $\mathrmbf{A}$. 
In the entity classification $\mathrmbfit{ent}(\mathrmbf{A})$, 
the elements of the universe are classified by entity types. 
Elements of the set of relational instances 
$\mathrmbfit{tuple}(\mathrmbf{A}) = \mathrmbfit{inst}(\mathrmbfit{rel}(\mathrmbf{A}))$ 
are called tuples or (relational) arguments. 
In the relation classification $\mathrmbfit{rel}(\mathrmbf{A})$,
tuples are classified by relation types. 

Although we do not specify it, 
usually there is a subset of individuals $\mathrmbfit{indiv}(\mathrmbf{A}) \subseteq \mathrmbfit{univ}(\mathrmbf{A})$, 
whose elements are called individual designators. 
We use some terminology from conceptual graphs here. 
Individual designators, thought of as identifiers for individuals, 
represent either objects or data values. 
A locator represents an object: 
a marker `{\ttfamily ISBN-0-521-58386-1}', an indexical `{\ttfamily you}' or a name `{\ttfamily K. Jon Barwise}'. 
A data value is represented as a literal: 
a string `{\ttfamily xyz}', a number `{\ttfamily 3.14159}', a date `{\ttfamily 1776/07/04}', etc.

Tuples in $r \in \mathrmbfit{tuple}(\mathrmbf{A})$ are abstract, and by themselves are typeless. 
They are associated with concrete reference tuples, 
and thus elements of the universe, via the tuple function
$\delta_{0}
: \mathrmbfit{tuple}(\mathrmbf{A}) \rightarrow \mathrmbfit{tuple}(\mathrmbfit{refer}(\mathrmbf{A}))$.
A concrete tuple 
$\delta_{0}(r) \in \mathrmbfit{tuple}(\mathrmbfit{refer}(\mathrmbf{A}))$ 
is a tuple of entity instances $\delta_{0}(r) \in \mathrmbfit{univ}(\mathrmbf{A})^{\mathrmbfit{arity}(r)}$. 
Thus, 
it assigns elements of the universe (entity instances) to the variables in the subset $\mathrmbfit{arity}(r)$. 
The arity is associated with each abstract tuple $r \in \mathrmbfit{tuple}(\mathrmbf{A})$ 
via the arity function
$\alpha_{0}
: \mathrmbfit{tuple}(\mathrmbf{A}) \rightarrow {\wp}\mathrmbfit{var}(\mathrmbf{A})$.

Amongst the five basic sets
---
variables $\mathrmbfit{var}(\mathrmbf{A})$,
entity types $\mathrmbfit{typ}(\mathrmbfit{ent}(\mathrmbf{A}))$,
universe $\mathrmbfit{univ}(\mathrmbf{A})$,
relation types $\mathrmbfit{typ}(\mathrmbfit{rel}(\mathrmbf{A}))$ and
abstract tuples $\mathrmbfit{tuple}(\mathrmbf{A})$
---
we make no a priori assumption about their connections
other than the linkages provided by the 
reference, relation signature, tuple, relation arity, and tuple arity functions.
\begin{center}
$\footnotesize{\begin{array}{r@{\hspace{5pt}\doteq\hspace{5pt}}r@{\hspace{5pt}:\hspace{5pt}}l}
\ast & \mathrmbfit{typ}(\tau_{\mathrmbf{A}}) &
\mathrmbfit{var}(\mathrmbf{A}) \rightarrow \mathrmbfit{typ}(\mathrmbfit{ent}(\mathrmbf{A}))
\\
\delta_{1} & \mathrmbfit{typ}(\delta_{\mathrmbf{A}}) &
\mathrmbfit{typ}(\mathrmbfit{rel}(\mathrmbf{A})) \rightarrow \mathrmbfit{sign}(\ast)
\\
\delta_{0} & \mathrmbfit{inst}(\delta_{\mathrmbf{A}}) &
\mathrmbfit{tuple}(\mathrmbf{A}) \rightarrow \mathrmbfit{tuple}(\tau_{\mathrmbf{A}})
\\
\alpha_{1} & \delta_{1} \cdot \mathrmbfit{arity}(\ast) &
\mathrmbfit{typ}(\mathrmbfit{rel}(\mathrmbf{A})) \rightarrow {\wp}\mathrmbfit{var}(\mathrmbf{A})
\\
\multicolumn{1}{r@{\hspace{5pt}\rule{7pt}{0pt}\hspace{5pt}}}{\alpha_{0}} & \multicolumn{1}{r@{\hspace{5pt}:\hspace{5pt}}}{} & 
\mathrmbfit{tuple}(\mathrmbf{A}) \rightarrow {\wp}\mathrmbfit{var}(\mathrmbf{A}).
\end{array}}$
\end{center}
They may or may not overlap or be disjoint.
Special assumptions can be used for special cases.
See the database-as-category example given below.

Tuples are called records in the context of database relations: 
if tuple $r$ has relational type $\rho$, $r \models_{\mathrmbfit{rel}(\mathrmbf{A})} \rho$, 
then the concrete tuple $\delta_{0}(r)$ has the signature type 
$\delta_{1}(\rho), \delta_{0}(r) \models_{\mathrmbfit{sign}(\tau_{\mathrmbf{A}})} \delta_{1}(\rho)$; 
this means that the arities are related as 
$\mathrmbfit{arity}(r) \supseteq \mathrmbfit{arity}(\rho)$, 
and 
$\delta_{0}(r)_{x} \models_{\mathrmbfit{ent}(\mathrmbf{A})} \delta_{1}(\rho)_{x}$ 
for each variable $x \in \mathrmbfit{arity}(\rho)$. 
Of course, 
although a tuple $r$ has a fixed arity $\mathrmbfit{arity}(r) = \{x_{1},x_{2},\cdots x_{m}\}$ 
and a fixed signature 
$\delta_{0}(r) 
= \{ \delta_{0}(r)_{x} \mid x \in \mathrmbfit{arity}(r) \} 
= (\delta_{0}(r)_{x_{1}},\delta_{0}(r)_{x_{2}},\cdots \delta_{0}(r)_{x_{m}})$, 
a different relational classification incidence 
$r \models_{\mathrmbfit{rel}(\mathrmbf{A})} \rho'$ 
may assert $\mathrmbfit{arity}(r) \supseteq \mathrmbfit{arity}(\rho')$, 
and $\delta_{0}(r)_{x'} \models_{\mathrmbfit{ent}(\mathrmbf{A})} \delta_{1}(\rho')_{x'}$ 
for each variable $x' \in \mathrmbfit{arity}(\rho')$, 
where even the cardinalities (valences) differ $|\mathrmbfit{arity}(\rho)| \not= |\mathrmbfit{arity}(\rho')|$. 
For an example with the same valences but different arities, 
the concrete tuple $(\mathtt{Sam},70)$ could represent an age or a height, 
corresponding to two different type signatures $(\mathtt{Person},\mathtt{Years})$ and $(\mathtt{Person},\mathtt{Inches})$.

\begin{table}
\begin{center}
{\footnotesize \begin{tabular}[t]{c}
\begin{minipage}{300pt}
\begin{wrapfigure}{r}{80pt}
\setlength{\extrarowheight}{1.5pt}
\fbox{\scriptsize \begin{tabular}{|@{\hspace{2pt}}>{\itshape}l@{\hspace{2pt}}|@{\hspace{2pt}}>{\rmfamily}l@{\hspace{2pt}}|} \hline
\multicolumn{2}{|l|}{{\bfseries{\sffamily frame:} {\rmfamily send}}}	\\ \hline
{\sffamily\slshape role}	&	{\sffamily filler}	\\ \hline\hline
agent					&	Adam		\\ \hline
patient				&	Eve			\\ \hline
object				&	flowers	\\ \hline
instrument		&	email		\\ \hline
\end{tabular}}
\end{wrapfigure}
The frame name `{\bfseries\rmfamily send}' is a relation type.
Each role 
`{\itshape agent}', `{\itshape patient}', `{\itshape object}' and `{\itshape instrument}' 
is a variable (entity type name).
The set of roles $\{\text{`{\itshape agent}'}, \text{`{\itshape patient}'}, \text{`{\itshape object}'}, \text{`{\itshape instrument}'} \}$ is an arity.
Since this kind of frame is untyped, 
by using a universal entity type `Entity', 
the signature of the relation type `{\bfseries\rmfamily send}' is
({\itshape agent}: Entity, {\itshape patient}: Entity, {\itshape object}: Entity, {\itshape instrument}: Entity).
Each role filler `Adam', `Eve', `flowers' and `email' is an element of the universe (an entity instance).
There is an abstract tuple of relation type `{\bfseries\rmfamily send}', 
which is assigned this tuple of instances,
(agent: Adam, patient: Eve, object: flowers, instrument: email).
Each role filling, 
such as `{\itshape agent}: Adam' 
is a pointwise classification incidence 
`Adam' $\models$ `{\itshape agent}'.
\end{minipage}
\end{tabular}}
\end{center}
\caption{Role Frame}
\label{role:frame}
\end{table}

\paragraph{Frames.}
Frames use linguistic roles, 
such as `{\itshape agent}', `{\itshape patient}' and `{\itshape instrument}', 
to fill verb-specific frame slots.
For example, 
the sentence ``Adam sent the flowers to Eve by email'' 
would fill the `{\bfseries\rmfamily send}' frame shown in Table~\ref{role:frame}. 
This table illustrates one way to represent role frames in FOL structures.

Frames are used to represent objects, events, abstract relationships, etc. 
A frame has (at least) the following parts:
(1) a name,
(2) zero or more abstractions or superclasses (classes to which the concept belongs)
(3) zero or more slots or attributes that describe particular properties of the concept.
By using the role frame example in Table~\ref{role:frame},
we see that 
frame names are represented by relation type symbols, 
roles are represented by variables, and 
role fillers are represented by entity instances. 
The frame itself represents 
the classification incidence between the tuple and the relation type `{\bfseries\rmfamily send}'. 
Frame abstractions (sub-classification) 
can be specified by 
using sequents in the IF theories that correspond to classifications.

\newpage
\subsection{Datasets as $\mathrmbf{Set}$-Functors}

Assume our models are unified (entity classification = relation classification)
$\mathrmbfit{ent}(\mathrmbf{A}) = \mathrmbfit{rel}(\mathrmbf{A})$
and trim (abstract tuple arity = concrete tuple arity)
$\alpha_{0}(r) = \alpha_{1}(\rho)$ if $r \models_{\mathrmbf{A}} \rho$.
Note in particular that abstract tuples are identical to objects in the universe
$\mathrmbfit{tuple}(\mathrmbf{A}) = \mathrmbfit{univ}(\mathrmbf{A})$.
\begin{center}
$\footnotesize{\begin{array}{r@{\hspace{5pt}\doteq\hspace{5pt}}r@{\hspace{5pt}:\hspace{5pt}}l}
\ast & \mathrmbfit{typ}(\tau_{\mathrmbf{A}}) &
\mathrmbfit{var}(\mathrmbf{A}) \rightarrow \mathrmbfit{typ}(\mathrmbfit{rel}(\mathrmbf{A}))
\\
\delta_{1} & \mathrmbfit{typ}(\delta_{\mathrmbf{A}}) &
\mathrmbfit{typ}(\mathrmbfit{rel}(\mathrmbf{A})) \rightarrow \mathrmbfit{sign}(\ast)
\\
\delta_{0} & \mathrmbfit{inst}(\delta_{\mathrmbf{A}}) &
\mathrmbfit{univ}(\mathrmbf{A}) \rightarrow \mathrmbfit{tuple}(\tau_{\mathrmbf{A}})
\\
\alpha_{1} & \delta_{1} \cdot \mathrmbfit{arity}(\ast) &
\mathrmbfit{typ}(\mathrmbfit{rel}(\mathrmbf{A})) \rightarrow {\wp}\mathrmbfit{var}(\mathrmbf{A})
\\
\multicolumn{1}{r@{\hspace{5pt}\rule{7pt}{0pt}\hspace{5pt}}}{\alpha_{0}} & \multicolumn{1}{r@{\hspace{5pt}:\hspace{5pt}}}{} & 
\mathrmbfit{univ}(\mathrmbf{A}) \rightarrow {\wp}\mathrmbfit{var}(\mathrmbf{A}).
\end{array}}$
\end{center}


\begin{center}
\begin{tabular}{c}
\\ \\
\setlength{\unitlength}{0.5pt}
\begin{picture}(160,80)(-50,0)
\put(-160,0){\begin{picture}(160,80)(0,0)
\put(0,80){\makebox(0,0){\footnotesize{$\mathrmbfit{var}(\mathrmbf{A})$}}}
\put(160,80){\makebox(0,0){\footnotesize{$\mathrmbfit{typ}(\mathrmbfit{rel}(\mathrmbf{A}))$}}}
\put(0,0){\makebox(0,0){\footnotesize{$\mathrmbfit{var}(\mathrmbf{A})$}}}
\put(160,0){\makebox(0,0){\footnotesize{$\mathrmbfit{univ}(\mathrmbf{A})$}}}
\put(-10,40){\makebox(0,0)[r]{\small{$=$}}}
\put(170,40){\makebox(0,0)[l]{\scriptsize{$\models_{\mathrmbfit{rel}(\mathrmbf{A})}$}}}
\put(70,95){\makebox(0,0){\scriptsize{$\ast$}}}
\put(160,15){\line(0,1){50}}
\put(40,80){\vector(1,0){64}}
\put(80,50){\makebox(0,0){\small{$\mathrmbfit{refer}(\mathrmbf{A})$}}}
\put(80,30){\makebox(0,0){\scriptsize{$semidesignation$}}}
\end{picture}}
\put(0,80){\makebox(0,0){\footnotesize{$\mathrmbfit{typ}(\mathrmbfit{rel}(\mathrmbf{A}))$}}}
\put(160,80){\makebox(0,0){\footnotesize{$\mathrmbfit{sign}(\ast)$}}}
\put(0,0){\makebox(0,0){\footnotesize{$\mathrmbfit{univ}(\mathrmbf{A})$}}}
\put(160,0){\makebox(0,0){\footnotesize{$\mathrmbfit{tuple}(\mathrmbfit{refer}(\mathrmbf{A}))$}}}
\put(170,40){\makebox(0,0)[l]{\scriptsize{$\models_{\mathrmbfit{sign}(\mathrmbfit{refer}(\mathrmbf{A}))}$}}}
\put(85,95){\makebox(0,0){\scriptsize{$\delta_{1}$}}}
\put(76,-15){\makebox(0,0){\scriptsize{$\delta_{0}$}}}
\put(0,15){\line(0,1){50}}
\put(160,15){\line(0,1){50}}
\put(54,80){\vector(1,0){66}}
\put(46,0){\vector(1,0){40}}
\put(100,50){\makebox(0,0){\small{$\mathrmbfit{sign}(\mathrmbf{A})$}}}
\put(100,30){\makebox(0,0){\scriptsize{$designation$}}}
\end{picture}
\\ \\ \\
\end{tabular}
\end{center}


\begin{center}
{\scriptsize $
\mathrmbf{A}={\langle{\mathrmbfit{refer}(\mathrmbf{A}),\mathrmbfit{sign}(\mathrmbf{A})}\rangle} 
\stackrel{\mathrmbfit{h}}{\rightleftarrows}
{\langle{\mathrmbfit{refer}(\mathrmbf{B}),\mathrmbfit{sign}(\mathrmbf{B})}\rangle}=\mathrmbf{B}$}
\newline
{\scriptsize $\mathrmbfit{univ}(\mathrmbfit{h})(s)\in\mathrmbfit{ext}_{\mathrmbfit{F}}(\rho)$ iff
$s\in\mathrmbfit{ext}_{\mathrmbfit{G}}(\mathrmbfit{rel}(\mathrmbfit{h})(\rho))$}
\newline
{\scriptsize $\mathrmbfit{univ}(\mathrmbfit{h})^{-1}(\mathrmbfit{ext}_{\mathrmbfit{F}}(\rho))=
\mathrmbfit{ext}_{\mathrmbfit{G}}(\mathrmbfit{rel}(\mathrmbfit{h})(\rho))$}
\end{center}

Assume structure morphism
$\mathrmbf{A}\stackrel{\mathrmbfit{h}}{\rightleftarrows}\mathrmbf{B}$
has identity universe component.
Then
$\mathrmbfit{rel}(\mathrmbfit{h})$ corresponds to $\mathrmbfit{obj}(\mathrmbfit{H})$
and
$\mathrmbfit{var}(\mathrmbfit{h})$ corresponds to $\mathrmbfit{mor}(\mathrmbfit{H})$.
\newline
Is this a limitation, since $\mathrmbfit{var}(\mathrmbfit{h})$ must be an isomorphism.


\begin{center}
\begin{tabular}{c@{\hspace{20pt}}c}
\\ \\
\setlength{\unitlength}{0.5pt}
\begin{picture}(160,0)(0,0)
\put(-40,80){\makebox(0,0){\footnotesize{$\mathrmbf{A}$}}}
\put(40,80){\makebox(0,0){\footnotesize{$\mathrmbf{B}$}}}
\put(0,0){\makebox(0,0){\footnotesize{$\mathrmbf{Set}$}}}
\put(0,90){\makebox(0,0){\scriptsize{$\mathrmbfit{H}$}}}
\put(-30,40){\makebox(0,0)[r]{\scriptsize{$\mathrmbfit{F}$}}}
\put(30,40){\makebox(0,0)[l]{\scriptsize{$\mathrmbfit{G}$}}}
\put(-25,80){\vector(1,0){50}}
\put(-30,60){\vector(1,-2){20}}
\put(30,60){\vector(-1,-2){20}}
\put(0,120){\makebox(0,0){\scriptsize{$\mathrmbfit{ext}_{\mathrmbfit{F}}(\rho)=\mathrmbfit{ext}_{\mathrmbfit{G}}(\mathrmbfit{H}(\rho))$}}}
\end{picture}
&
\setlength{\unitlength}{0.45pt}
\begin{picture}(120,80)(0,0)
\put(0,80){\makebox(0,0){\footnotesize{$\int\mathrmbfit{F}$}}}
\put(120,80){\makebox(0,0){\footnotesize{$\int\mathrmbfit{G}$}}}
\put(0,0){\makebox(0,0){\footnotesize{$\mathrmbf{A}$}}}
\put(120,0){\makebox(0,0){\footnotesize{$\mathrmbf{B}$}}}
\put(-10,40){\makebox(0,0)[r]{\scriptsize{$\mathrmbfit{pr}$}}}
\put(130,40){\makebox(0,0)[l]{\scriptsize{$\mathrmbfit{pr}$}}}
\put(60,90){\makebox(0,0){\scriptsize{$\int\mathrmbfit{H}$}}}
\put(60,-10){\makebox(0,0){\scriptsize{$\mathrmbfit{H}$}}}
\put(0,60){\vector(0,-1){40}}
\put(120,60){\vector(0,-1){40}}
\put(20,80){\vector(1,0){80}}
\put(20,0){\vector(1,0){80}}
\put(-40,110){\begin{picture}(0,30)(0,0)
\put(0,30){\makebox(0,0){\scriptsize{$x:{\langle{\rho,r}\rangle}\rightarrow{\langle{\rho',r'}\rangle}$}}}
\put(0,15){\makebox(0,0){\scriptsize{$x:\rho\rightarrow\rho'$}}}
\put(0,0){\makebox(0,0){\scriptsize{$\mathrmbfit{H}(x)(r)=r'$}}}
\end{picture}}
\put(60,120){\makebox(0,0){\large{$\mapsto$}}}
\put(200,110){\begin{picture}(0,30)(0,0)
\put(0,30){\makebox(0,0){\scriptsize{$\mathrmbfit{H}(x):{\langle{\mathrmbfit{H}(\rho),r}\rangle}\rightarrow{\langle{\mathrmbfit{H}(\rho'),r'}\rangle}$}}}
\put(0,15){\makebox(0,0){\scriptsize{$\mathrmbfit{H}(x):\mathrmbfit{H}(\rho)\rightarrow\mathrmbfit{H}(\rho')$}}}
\put(0,0){\makebox(0,0){\scriptsize{$\mathrmbfit{F}(x)=\mathrmbfit{G}(\mathrmbfit{H}((x))(r)=r'$}}}
\end{picture}}
\end{picture}
\\ & \\
\end{tabular}
\end{center}


\begin{center}
\setlength{\extrarowheight}{3pt}
{\footnotesize $\begin{array}{@{\hspace{-100pt}}r@{\hspace{5pt}}|c@{\hspace{5pt}}|@{\hspace{5pt}}c|}
\multicolumn{1}{r}{} & \multicolumn{1}{c}{\;\;\;\;\mathsfbf{IFF}\;\;\;\;\;} & \multicolumn{1}{c}{\mathsfbf{Db=Cat}}
\\ \cline{2-3}
& \text{ structure}\; \mathrmbf{A} 
= {\langle{\mathrmbfit{refer}(\mathrmbf{A}),\mathrmbfit{sign}(\mathrmbf{A})}\rangle}
& \text{database state}\; \mathrmbfit{F} : \mathrmbf{A} \rightarrow \mathrmbf{Set}
\\ \cline{2-3}
\text{database schema} 
& {\langle{\mathrmbfit{var}(\mathrmbf{A}),\mathrmbfit{typ}(\mathrmbfit{rel}(\mathrmbf{A})),\alpha_{1},\ast}\rangle}
& \mathrmbf{A}
= {\langle{\mathrmbfit{mor}(\mathrmbf{A}),\mathrmbfit{obj}(\mathrmbf{A}),\mathrmbfit{src}_{A},\mathrmbfit{tgt}_{A}}\rangle}
\\
\text{relational type or table} 
& \rho \in \mathrmbfit{typ}(\mathrmbfit{rel}(\mathrmbf{A})) 
& \rho \in \mathrmbfit{obj}(\mathrmbf{A})
\\
\cline{2-3}
\text{relational columns}
& x \in \mathrmbfit{var}(\mathrmbf{A})
& x \in \mathrmbfit{mor}(\mathrmbf{A})
\\
x \text{ is col of table } \rho 
& x \in {\alpha}(\rho)
&
\\
\text{column type}
& \ast(x)
& \mathrmbfit{tgt}_{A}(x)
\\
\text{ col } x \text{ has type } \rho' 
& x \in {\ast}^{-1}(\rho')
&
\\ 
\text{col of } \rho \text{ valued in } \rho'
& \delta_{1}(\rho)_{x} = \rho'
& x : \rho \rightarrow \rho'
\\ \cline{2-3}
\text{relational rows}
& \mathrmbfit{univ}(\mathrmbf{A})
& \mathrmbfit{obj}(\int\mathrmbfit{F}) = \{ (\rho,r) \mid \rho \in \mathrmbfit{obj}(A), r \in \mathrmbfit{F}(\rho) \}
\\
\text{relational rows of } \rho
& \mathrmbfit{ext}_{\mathrmbfit{F}}(\rho)
& \mathrmbfit{F}(\rho)
\\ \cline{2-3}
\text{relational arity (columns)}
& \alpha_{1}(\rho)
& \{ x \mid \mathrmbfit{src}_{\mathrmbf{A}}(x) = \rho \}
\\
\text{relational signature}
& \delta_{1}(\rho) : \alpha_{1}(\rho) \rightarrow \mathrmbfit{typ}(\mathrmbfit{rel}(\mathrmbf{A})) 
: x \mapsto \delta_{1}(\rho)_{x} = \rho'
& 
(\mathrmbfit{tgt}_{\mathrmbf{A}}(x) \mid \mathrmbfit{src}_{\mathrmbf{A}}(x) = \rho)
\\ 
\text{tupling arity}
& \alpha_{0}(r)
& 
\mathrmbfit{arity}(r) = \mathrmbfit{arity}(\rho) \text{ if } r \in \mathrmbfit{F}(\rho)
\\
\text{row record}
& \delta_{0}(r) : \alpha_{0}(r) \rightarrow \mathrmbfit{univ}(\mathrmbf{A}) 
& 
(\mathrmbfit{F}(x)(r) \mid \mathrmbfit{src}_{\mathrmbf{A}}(x) = \rho)
\\ 
\cline{2-3}
\text{cell indices of } \rho
& (r,x) \text{ where } r \models_{\mathrmbf{A}} \rho, x \in \alpha_{0}(r)
& (r,x) \text{ where } r \in \mathrmbfit{F}(\rho), \rho = \mathrmbfit{src}_{A}(x)
\\
\text{set of cell values (universe)}
& \mathrmbfit{univ}(\mathrmbf{A})
& \coprod_{r,x} \mathrmbfit{F}(x)(r)
\\
\cline{2-3}
\end{array}$}
\end{center}
%

%% file: framework.tex
\newpage
\section{Framework}\label{framework}

\begin{meta-axiom}\label{specification:preorder}
{\normalsize$\blacksquare$}
There is an adjointly indexed preorder of specifications
\footnotesize\[
\mathrmbfit{spec} : \mathrmbf{Lang}^{\mathrm{op}} \rightarrow \mathrmbf{Adj}.
\]\normalsize
\end{meta-axiom}
For any language $\Sigma$,
there is a fiber preorder $\mathrmbfit{spec}(\Sigma) 
= {\langle{\mathrmbfit{spec}(\Sigma),\geq_{\Sigma}}\rangle}$,
where an element $T \in \mathrmbfit{spec}(\Sigma)$ is called a $\Sigma$-specification
and an ordering $T \geq_{\Sigma} T'$ in $\mathrmbfit{spec}(\Sigma)$ is called 
a reverse entailment (concept lattice, or generalization-specialization) $\Sigma$-specification ordering. 
For any language morphism $\sigma : \Sigma_{1} \rightarrow \Sigma_{2}$,
there is a inverse image fiber monotonic function
$\mathrmbfit{spec}(\sigma) = \mathrmbfit{inv}(\sigma) 
: \mathrmbfit{spec}(\Sigma_{2}) \rightarrow \mathrmbfit{spec}(\sigma_{1})$,
which is right adjoint 
to a direct image fiber monotonic function
$\mathrmbfit{dir}(\sigma)
: \mathrmbfit{spec}(\Sigma_{1}) \rightarrow \mathrmbfit{spec}(\sigma_{2})$:
\footnotesize\[
\mathrmbfit{dir}(\sigma)(T_{1}) \geq_{\Sigma_{2}} T_{2}
\;\;\text{\underline{iff}}\;\;
T_{1} \geq_{\Sigma_{1}} \mathrmbfit{inv}(\sigma)(T_{2})
\]\normalsize
for each source specification $T_{1} \in \mathrmbfit{spec}(\sigma_{1})$
and target specification $T_{2} \in \mathrmbfit{spec}(\sigma_{2})$.

The homogenization (Grothendieck construction) is a fibered category $\mathrmbf{Spec}$ 
with index projection functor $\mathrmbfit{pr} : \mathrmbf{Spec} \rightarrow \mathrmbf{Lang}$.
An object in $\mathrmbf{Spec}$,
called a specification,
is a pair $\mathcal{T} = {\langle{\Sigma,T}\rangle}$,
where $\Sigma$ is a language and 
$T \in \mathrmbfit{spec}(\Sigma)$ is a $\Sigma$-specification.
A morphism $\sigma : {\langle{\Sigma_{1},T_{1}}\rangle} \rightarrow {\langle{\Sigma_{2},T_{2}}\rangle}$ in $\mathrmbf{Spec}$, 
called a specification morphism,
is a language morphism $\sigma : \Sigma_{1} \rightarrow \Sigma_{2}$,
where $T_{1} \geq_{\Sigma_{1}} \mathrmbfit{inv}(\sigma)(T_{2})=\mathrmbfit{spec}(\sigma)(T_{2})$
is a generalization-specialization ordering in the source fiber preorder $\mathrmbfit{spec}(\Sigma_{1})$;
equivalently $\mathrmbfit{dir}(\sigma)(T_{1}) \geq_{\Sigma_{2}} T_{2}$
is a generalization-specialization ordering in the target fiber preorder $\mathrmbfit{spec}(\Sigma_{2})$.

\begin{meta-axiom}\label{structure:category}
{\normalsize$\blacksquare$}
There is an indexed category of structures
\footnotesize\[
\mathrmbfit{struc} : \mathrmbf{Lang}^{\mathrm{op}} \rightarrow \mathrmbf{Cat}.
\]\normalsize
\end{meta-axiom}
For any language $\Sigma$,
there is a fiber category $\mathrmbfit{struc}(\Sigma_{1})$,
where an object $M \in \mathrmbfit{struc}(\Sigma)$ is called a $\Sigma$-structure 
and a morphism $f : M \rightarrow M'$ in $\mathrmbfit{struc}(\Sigma)$ is called a $\Sigma$-structure morphism~\footnote{For each $\Sigma$-structure morphism $f : M \rightarrow M'$,
the intent indexed functor (see subsection~\ref{intent})
associates a generalization-specialization $\Sigma$-specification ordering
$\mathrmbfit{int}_{\Sigma}(M) \geq_{\Sigma} \mathrmbfit{int}_{\Sigma}(M')$.
By means of the intent functor,
we regard the fiber category $\mathrmbfit{struc}(\Sigma_{1})$
to be a semantic analog of the concept lattice at $\Sigma$ generated by the satisfaction relation,
with the $\Sigma$-structure morphism $f : M \rightarrow M'$
as an abstract generalization-specialization $\Sigma$-structure ordering pointing downward in the concept lattice.}. 
For any language morphism $\sigma : \Sigma_{1} \rightarrow \Sigma_{2}$,
there is a fiber functor
$\mathrmbfit{struc}(\sigma) : \mathrmbfit{struc}(\Sigma_{2}) \rightarrow \mathrmbfit{struc}(\Sigma_{1})$,
which maps a $\Sigma_{2}$-structure $M_{2}$
to the $\Sigma_{1}$-structure $\mathrmbfit{struc}(\sigma)(M_{2})$
and maps a $\Sigma_{2}$-structure morphism $f_{2} : M_{2} \rightarrow M'_{2}$
to the $\Sigma_{1}$-structure morphism 
$\mathrmbfit{struc}(\sigma)(f_{2}) : \mathrmbfit{struc}(\sigma)(M_{2}) \rightarrow \mathrmbfit{struc}(\sigma)(M'_{2})$.

The homogenization (Grothendieck construction) is a fibered category $\mathrmbf{Struc}$ 
with index projection functor $\mathrmbfit{pr} : \mathrmbf{Struc} \rightarrow \mathrmbf{Lang}$.
An object in $\mathrmbf{Struc}$,
called a structure,
is a pair $\mathcal{M} = (\Sigma,M)$,
where $\Sigma$ is a language and 
$M \in \mathrmbfit{struc}(\Sigma)$ is a $\Sigma$-structure.
A morphism in $\mathrmbf{Struc}$, 
called a structure morphism,
is a pair $(\sigma,h) : (\Sigma_{1},M_{1}) \rightarrow (\Sigma_{2},M_{2})$, 
where $\sigma : \Sigma_{1} \rightarrow \Sigma_{2}$ is a language morphism 
and $h : M_{1} \rightarrow \mathrmbfit{struc}(\sigma)(M_{2})$ is a $\Sigma_{1}$-structure morphism,
a morphism in the fiber category $\mathrmbfit{struc}(\Sigma_{1})$.~\footnote{Giving 
the generalization-specialization $\Sigma_{1}$-specification ordering
\newline
$\mathrmbfit{int}_{\Sigma_{1}}(M_{1}) 
\geq_{\Sigma_{1}} \mathrmbfit{int}_{\Sigma_{1}}(\mathrmbfit{struc}(\sigma)(M_{2}))
= \mathrmbfit{spec}(\sigma)(\mathrmbfit{int}_{\Sigma_{2}}(M_{2}))$,
hence implying that
$\sigma : (\Sigma_{1},\mathrmbfit{int}_{\Sigma_{1}}(M_{1})) \rightarrow (\Sigma_{2},\mathrmbfit{int}_{\Sigma_{1}}(M_{2}))$
is a specification morphism.
This defines the intent fibered functor,
the homogenization of the intent indexed functor.}

\begin{definition}\label{logic:category}
{\normalsize$\blacksquare$}
The logic indexed category 
is the product of the structure and specification indexed categories
\footnotesize\[
\mathrmbfit{log} = \mathrmbfit{struc}{\times}\mathrmbfit{spec} : 
\mathrmbf{Lang}^{\mathrm{op}} \rightarrow \mathrmbf{Cat}.
\]\normalsize
There are indexed projection functors
\footnotesize\[
\mathrmbfit{struc} \stackrel{\mathrmbfit{pr}_{0}}{\Longleftarrow} \mathrmbfit{log} \stackrel{\mathrmbfit{pr}_{1}}{\Longrightarrow} \mathrmbfit{spec}.
\]\normalsize
\end{definition}
For any language $\Sigma$,
there is a fiber category $\mathrmbfit{log}(\Sigma_{1})=\mathrmbfit{struc}(\Sigma_{1}){\times}\mathrmbfit{spec}(\Sigma_{1})$,
where an object $L = {\langle{M,T}\rangle} \in \mathrmbfit{log}(\Sigma)$,
called a $\Sigma$-logic,
consists of a $\Sigma$-structure $M$ and a $\Sigma$-specification $T$, and
a morphism 
$f : L = {\langle{M,T}\rangle} \rightarrow {\langle{M',T'}\rangle} = L'$,
called a $\Sigma$-logic morphism,
consists of a $\Sigma$-structure morphism $f : M \rightarrow M'$
and a $\Sigma$-specification ordering $T \geq_{\Sigma} T'$~\footnote{In general,
this ordering is independent of the $\Sigma$-specification ordering
$\mathrmbfit{int}_{\Sigma}(M) \geq_{\Sigma} \mathrmbfit{int}_{\Sigma}(M')$
associated with the $\Sigma$-structure morphism $f : M \rightarrow M'$.}.
The projection functors 
map a logic $L = {\langle{M,T}\rangle}$
to its components $\mathrmbfit{pr}_{0}(L) = M$ and $\mathrmbfit{pr}_{1}(L) = T$, and 
map a logic morphism $f$
to its components $\mathrmbfit{pr}_{0}(f) = f : M \rightarrow M'$ and $\mathrmbfit{pr}_{1}(f) = (T \geq_{\Sigma} T')$.
For any language morphism $\sigma : \Sigma_{1} \rightarrow \Sigma_{2}$,
there is a fiber functor
$\mathrmbfit{log}(\sigma) : \mathrmbfit{log}(\Sigma_{2}) \rightarrow \mathrmbfit{log}(\Sigma_{1})$,
which maps a $\Sigma_{2}$-logic $L_{2} = {\langle{M_{2},T_{2}}\rangle}$
to the $\Sigma_{1}$-logic $\mathrmbfit{log}(\sigma)(L_{2}) = {\langle{\mathrmbfit{struc}(\sigma)(M_{2}),\mathrmbfit{spec}(\sigma)(T_{2})}\rangle}$
and maps a $\Sigma_{2}$-logic morphism $f_{2} : L_{2} \rightarrow L'_{2}$
consisting of a $\Sigma_{2}$-structure morphism $f_{2} : M_{2} \rightarrow M'_{2}$ 
and a $\Sigma_{2}$-specification ordering $T_{2} \geq_{\Sigma} T'_{2}$
to the $\Sigma_{1}$-logic morphism 
$\mathrmbfit{log}(\sigma)(f_{2}) : {\langle{\mathrmbfit{struc}(\sigma)(M_{2}),\mathrmbfit{spec}(\sigma)(T_{2})}\rangle} \rightarrow {\langle{\mathrmbfit{struc}(\sigma)(M'_{2}),\mathrmbfit{spec}(\sigma)(T'_{2})}\rangle}$
consisting of a $\Sigma_{1}$-structure morphism 
$\mathrmbfit{struc}(\sigma)(f_{2} : M_{2} \rightarrow M'_{2}) = 
\mathrmbfit{struc}(\sigma)(f_{2})  : \mathrmbfit{struc}(\sigma)(M_{2}) \rightarrow \mathrmbfit{struc}(\sigma)(M'_{2})$
and a $\Sigma_{1}$-specification ordering
$\mathrmbfit{spec}(\sigma)(T_{2} \geq_{\Sigma_{2}} T'_{2}) = 
(\mathrmbfit{spec}(\sigma)(T_{2}) \geq_{\Sigma_{1}} \mathrmbfit{spec}(\sigma)(T'_{2}))$.
For any language morphism $\sigma : \Sigma_{1} \rightarrow \Sigma_{2}$,
projection satisfies the naturality commutative diagrams
$\mathrmbfit{log}(\sigma) \circ \mathrmbfit{pr}_{0,\Sigma_{1}} = 
\mathrmbfit{pr}_{0,\Sigma_{2}} \circ \mathrmbfit{struc}(\sigma)$
and
$\mathrmbfit{log}(\sigma) \circ \mathrmbfit{pr}_{1,\Sigma_{1}} = 
\mathrmbfit{pr}_{1,\Sigma_{2}} \circ \mathrmbfit{spec}(\sigma)$~\footnote{since,
$\mathrmbfit{pr}_{0,\Sigma_{1}}(\mathrmbfit{log}(\sigma)(f_{2} : L_{2} \rightarrow L'_{2}))
= \mathrmbfit{struc}(\sigma)(f_{2}) : \mathrmbfit{struc}(\sigma)(M_{2}) \rightarrow \mathrmbfit{struc}(\sigma)(M'_{2})
= \mathrmbfit{struc}(\sigma)(f_{2} : M_{2} \rightarrow M'_{2})
= \mathrmbfit{struc}(\sigma)(\mathrmbfit{pr}_{\mathrmbfit{struc},\Sigma_{2}}(f_{2} : L_{2} \rightarrow L'_{2}))$
and
$\mathrmbfit{pr}_{1,\Sigma_{1}}(\mathrmbfit{log}(\sigma)(f_{2} : L_{2} \rightarrow L'_{2}))
= (\mathrmbfit{spec}(\sigma)(T_{2}) \geq_{\Sigma_{1}} \mathrmbfit{spec}(\sigma)(T'_{2}))
= \mathrmbfit{spec}(\sigma)(T_{2} \geq_{\Sigma_{2}} T'_{2})
= \mathrmbfit{spec}(\sigma)(\mathrmbfit{pr}_{1,\Sigma_{2}}(f_{2} : L_{2} \rightarrow L'_{2}))$.}.

The homogenization (Grothendieck construction) is the fibered product category $\mathrmbf{Log}=\mathrmbf{Struc}{\times}\mathrmbf{Spec}$ 
with an index projection $\mathrmbfit{pr} : \mathrmbf{Log} \rightarrow \mathrmbf{Lang}$
and
product projection fibered functors
\footnotesize\[
\mathrmbf{Struc} \stackrel{\mathrmbfit{pr}_{0}}{\longleftarrow} \mathrmbf{Log} \stackrel{\mathrmbfit{pr}_{1}}{\longrightarrow} \mathrmbfit{Spec},
\]\normalsize
where product projections commute with fibered projections 
$\mathrmbfit{pr}_{0} \circ \mathrmbfit{pr} 
= \mathrmbfit{pr}
= \mathrmbfit{pr}_{1} \circ \mathrmbfit{pr}$.
An object in $\mathrmbf{Log}$,
called a logic,
is a triple $\mathcal{L} = (\Sigma,M,T)$,
where $\Sigma$ is a language, 
$M \in \mathrmbfit{struc}(\Sigma)$ is a $\Sigma$-structure
and 
$T \in \mathrmbfit{spec}(\Sigma)$ is a $\Sigma$-specification.
A morphism in $\mathrmbf{Log}$, 
called a logic morphism,
is a pair $(\sigma,f) : (\Sigma_{1},M_{1},T_{1}) \rightarrow (\Sigma_{2},M_{2},T_{2})$, 
where $\sigma : \Sigma_{1} \rightarrow \Sigma_{2}$ is a language morphism, and
$f : (M_{1},T_{1}) \rightarrow 
\mathrmbfit{log}(\sigma)(M_{2},T_{2}) = (\mathrmbfit{struc}(\sigma)(M_{2}),\mathrmbfit{inv}(\sigma)(T_{2}))$
is a $\Sigma_{1}$-logic morphism; 
that is,
$f : M_{1} \rightarrow \mathrmbfit{struc}(\sigma)(M_{2})$ is a $\Sigma_{1}$-structure morphism, and 
$T_{1} \geq_{\Sigma_{1}} \mathrmbfit{inv}(\sigma)(T_{2})$ is a $\Sigma_{1}$-specification ordering.
Hence,
a logic morphism $(\sigma,f) : (\Sigma_{1},M_{1},T_{1}) \rightarrow (\Sigma_{2},M_{2},T_{2})$
consists of
a structure morphism $(\sigma,f) : (\Sigma_{1},M_{1}) \rightarrow (\Sigma_{2},M_{2})$ and
a specification morphism $\sigma : (\Sigma_{1},T_{1}) \rightarrow (\Sigma_{2},T_{2})$. 


\begin{definition}\label{sound:logic:category}
{\normalsize$\blacksquare$}
The indexed category of sound logics (Figure~\ref{truth:architecture})
is the full indexed subcategory of the logic indexed category
consisting of only sound logics and all logic morphisms between them
\footnotesize\[
\mathrmbfit{inc} 
: \mathrmbfit{snd} \Rightarrow \mathrmbfit{log} 
: \mathrmbf{Lang}^{\mathrm{op}} \rightarrow \mathrmbf{Cat}.
\]\normalsize
There are indexed projection functors
~\footnote{Institutions do not have an intent functor, 
but only an intent object function.
The fiber subcategory $\mathrmbfit{snd}(\Sigma)$ is definable, but 
the fiber functor
$\mathrmbfit{snd}(\sigma) : \mathrmbfit{snd}(\Sigma_{2}) \rightarrow \mathrmbfit{snd}(\Sigma_{1})$
is not.
And
the naturality commutative diagram
is true only objectwise
$|\mathrmbfit{inc}_{\Sigma_{2}}| \cdot |\mathrmbfit{log}(\sigma)| 
= |\mathrmbfit{snd}(\sigma)| \cdot |\mathrmbfit{inc}_{\Sigma_{1}}|$.
Thus,
although the fibered category $\mathrmbf{Snd}$ is definable,
it is not the homogenization of an indexed category
$\mathrmbfit{snd} : \mathrmbf{Lang}^{\mathrm{op}} \rightarrow \mathrmbf{Cat}$.

But,
also of interest are the sound fibers $\mathrmbfit{snd}(\Sigma,T)$ over a specification ${\langle{\Sigma,T}\rangle}$
and the sound fiber functors
$\mathrmbfit{snd}(\sigma) : \mathrmbfit{snd}(\Sigma_{2},T_{2}) \rightarrow \mathrmbfit{snd}(\Sigma_{1},T_{1})$
over a specification morphism $\sigma : (\Sigma_{1},T_{1}) \rightarrow (\Sigma_{2},T_{2})$ in $\mathrmbf{Spec}$.}
\footnotesize\[
\mathrmbfit{struc} \stackrel{\mathrmbfit{pr}_{0}}{\Longleftarrow} \mathrmbfit{snd} \stackrel{\mathrmbfit{pr}_{1}}{\Longrightarrow} \mathrmbfit{spec},
\]\normalsize
defined by composition with the logic indexed projection functors
\footnotesize\[
\mathrmbfit{pr}_{0} = \mathrmbfit{inc} \circ \mathrmbfit{pr}_{0} 
\;\;\text{and}\;\;
\mathrmbfit{pr}_{1} = \mathrmbfit{inc} \circ \mathrmbfit{pr}_{1}.
\]\normalsize
\end{definition}
For any language $\Sigma$,
there is a fiber subcategory $\mathrmbfit{snd}(\Sigma) \hookrightarrow \mathrmbfit{log}(\Sigma)$
of sound logics,
where an object ${\langle{M,T}\rangle} \in \mathrmbfit{snd}(\Sigma)$
is of a $\Sigma$-logic 
that is sound $M \models_{\Sigma} T$, and
a morphism 
$f : {\langle{M,T}\rangle} \rightarrow {\langle{M',T'}\rangle}$
is a $\Sigma$-logic morphism,
consisting of a $\Sigma$-structure morphism $f : M \rightarrow M'$, 
so that 
$\mathrmbfit{int}_{\Sigma}(M) \geq_{\Sigma} \mathrmbfit{int}_{\Sigma}(M')$,
and a $\Sigma$-specification ordering $T \geq_{\Sigma} T'$. 
The projections 
map a sound logic ${\langle{M,T}\rangle}$
to its components $\mathrmbfit{pr}_{0}(M,T) = M$ and 
$\mathrmbfit{pr}_{1}(M,T) = T$, and 
map a logic morphism $f$
to its components $\mathrmbfit{pr}_{0}(f) = f : M \rightarrow M'$ and $\mathrmbfit{pr}_{1}(f) = (T \geq_{\Sigma} T')$.
For any language morphism $\sigma : \Sigma_{1} \rightarrow \Sigma_{2}$,
there is a fiber functor
$\mathrmbfit{snd}(\sigma) : \mathrmbfit{snd}(\Sigma_{2}) \rightarrow \mathrmbfit{snd}(\Sigma_{1})$,
which maps a sound $\Sigma_{2}$-logic ${\langle{M_{2},T_{2}}\rangle}$
to the sound $\Sigma_{1}$-logic 
$\mathrmbfit{snd}(\sigma)(M_{2},T_{2}) = \mathrmbfit{log}(\sigma)(M_{2},T_{2}) = {\langle{\mathrmbfit{struc}(\sigma)(M_{2}),\mathrmbfit{spec}(\sigma)(T_{2})}\rangle}$~\footnote{since 
$T_{2} \geq_{\Sigma_{2}} \mathrmbfit{int}_{\Sigma_{2}}(M_{2})$ implies
$\mathrmbfit{spec}(\sigma)(T_{2}) \geq_{\Sigma_{1}} 
\mathrmbfit{spec}(\sigma)(\mathrmbfit{int}_{\Sigma_{2}}(M_{2})) = 
\mathrmbfit{int}_{\Sigma_{1}}(\mathrmbfit{struc}(\sigma)(M_{2}))$}
and maps a $\Sigma_{2}$-logic morphism 
$f_{2} : {\langle{M_{2},T_{2}}\rangle} \rightarrow {\langle{M'_{2},T'_{2}}\rangle}$
between sound logics
to the $\Sigma_{1}$-logic morphism 
$\mathrmbfit{snd}(\sigma)(f_{2}) = \mathrmbfit{log}(\sigma)(f_{2}) : {\langle{\mathrmbfit{struc}(\sigma)(M_{2}),\mathrmbfit{spec}(\sigma)(T_{2})}\rangle} \rightarrow {\langle{\mathrmbfit{struc}(\sigma)(M'_{2}),\mathrmbfit{spec}(\sigma)(T'_{2})}\rangle}$
between sound logics.
For any language morphism $\sigma : \Sigma_{1} \rightarrow \Sigma_{2}$,
inclusion satisfies the naturality commutative diagram
$\mathrmbfit{inc}_{\Sigma_{2}} \circ \mathrmbfit{log}(\sigma) 
= \mathrmbfit{snd}(\sigma) \circ \mathrmbfit{inc}_{\Sigma_{1}}$;
or pointwise,
$\mathrmbfit{inc}_{\Sigma_{1}}(\mathrmbfit{snd}(\sigma)(M_{2},T_{2}))
= \mathrmbfit{log}(\sigma)(M_{2},T_{2})
= \mathrmbfit{log}(\sigma)(\mathrmbfit{inc}_{\Sigma_{2}}(M_{2},T_{2}))$
for every sound target logic ${\langle{M_{2},T_{2}}\rangle} \in \mathrmbfit{snd}(\Sigma_{2})$.
\begin{center}
\begin{tabular}{c}
\\ \\ \\
\setlength{\unitlength}{0.55pt}
\begin{picture}(130,80)(0,-5)
\put(0,0){\begin{picture}(0,0)(0,0)
\qbezier[40](-45,40)(-45,110)(0,110)
\qbezier[40](45,40)(45,110)(0,110)
\qbezier[40](-45,40)(-45,-30)(0,-30)
\qbezier[40](45,40)(45,-30)(0,-30)
\put(0,125){\makebox(0,0){$\overset{\text{\scriptsize\itshape{soundness}}}{\overbrace{\hspace{30pt}}}$}}
\put(0,80){\makebox(0,0){\footnotesize{$T$}}}
\put(5,40){\makebox(0,0){\rule[-0.5pt]{0.5pt}{6pt}\,\footnotesize{$\vee_{\Sigma}$}}}
\put(0,0){\makebox(0,0){\footnotesize{$\mathrmbfit{int}_{\Sigma}(M)$}}}
\put(0,-65){\makebox(0,0){$\underset{
\mbox{\scriptsize{\shortstack{${\langle{M,T}\rangle}$\\where\\$M\;\models_{\Sigma}\;T$}}}
}{\underbrace{\hspace{30pt}}}$}}
\end{picture}}
\put(130,0){\begin{picture}(0,0)(0,0)
\qbezier[40](-45,40)(-45,110)(0,110)
\qbezier[40](45,40)(45,110)(0,110)
\qbezier[40](-45,40)(-45,-30)(0,-30)
\qbezier[40](45,40)(45,-30)(0,-30)
\put(0,125){\makebox(0,0){$\overset{\text{\scriptsize\itshape{soundness}}}{\overbrace{\hspace{30pt}}}$}}
\put(0,80){\makebox(0,0){\footnotesize{$T'$}}}
\put(5,40){\makebox(0,0){\rule[-0.5pt]{0.5pt}{6pt}\,\footnotesize{$\vee_{\Sigma}$}}}
\put(0,0){\makebox(0,0){\footnotesize{$\mathrmbfit{int}_{\Sigma}(M')$}}}
\put(0,-65){\makebox(0,0){$\underset{
\mbox{\scriptsize{\shortstack{${\langle{M',T'}\rangle}$\\where\\$M'\;\models_{\Sigma}\;T'$}}}
}{\underbrace{\hspace{30pt}}}$}}
\end{picture}}
\put(-50,80){\makebox(0,0)[r]{$
\shortstack{\footnotesize{$\Sigma$}\text{-\scriptsize\itshape{specification}}\\\text{\scriptsize\itshape{morphism}}}
\;\left\{\rule{0pt}{14pt}\right.$}}
\put(-50,0){\makebox(0,0)[r]{$
\shortstack{\footnotesize{$\Sigma$}\text{-\scriptsize\itshape{structure}}\\\text{\scriptsize\itshape{morphism}}}
\;\left\{\rule{0pt}{14pt}\right.$}}
\put(65,80){\makebox(0,0){\footnotesize{$\geq_{\Sigma}$}}}
\put(65,0){\makebox(0,0){\footnotesize{$\geq_{\Sigma}$}}}
\put(65,-51){\makebox(0,0){\footnotesize{$\stackrel{f}{\longrightarrow}$}}}
\end{picture}
\\ \\ \\ \\ \\
\end{tabular}
\end{center}
The homogenization (Grothendieck construction) is the fibered subcategory 
$\mathrmbfit{inc} : \mathrmbf{Snd} \hookrightarrow \mathrmbf{Log}$ 
with an index projection 
$\mathrmbfit{pr} = \mathrmbfit{inc} \circ \mathrmbfit{pr} : \mathrmbf{Snd} \rightarrow \mathrmbf{Lang}$
and
product projection fibered functors
\footnotesize\[
\mathrmbf{Struc} \stackrel{\mathrmbfit{pr}_{0}}{\longleftarrow} \mathrmbf{Snd} \stackrel{\mathrmbfit{pr}_{1}}{\longrightarrow} \mathrmbfit{Spec},
\]\normalsize
define by
$\mathrmbfit{pr}_{0} = \mathrmbfit{inc} \circ \mathrmbfit{pr}_{0}$ and
$\mathrmbfit{pr}_{1} = \mathrmbfit{inc} \circ \mathrmbfit{pr}_{1}$,
where product projections commute with fibered projections 
$\mathrmbfit{pr}_{0} \circ \mathrmbfit{pr} 
= \mathrmbfit{pr}
= \mathrmbfit{pr}_{1} \circ \mathrmbfit{pr}$.
An object in $\mathrmbf{Snd}$,
called a sound logic,
is a triple $\mathcal{L} = (\Sigma,M,T)$,
where $\Sigma$ is a language, 
$M \in \mathrmbfit{struc}(\Sigma)$ is a $\Sigma$-structure
and 
$T \in \mathrmbfit{spec}(\Sigma)$ is a $\Sigma$-specification
that is true in $M$:
$M \models_{\Sigma} T$.
A morphism in $\mathrmbf{Snd}$
is a logic morphism between sound logics;
it is a pair $(\sigma,f) : (\Sigma_{1},M_{1},T_{1}) \rightarrow (\Sigma_{2},M_{2},T_{2})$, 
where $\sigma : \Sigma_{1} \rightarrow \Sigma_{2}$ is a language morphism, and
$f : (M_{1},T_{1}) \rightarrow 
\mathrmbfit{snd}(\sigma)(M_{2},T_{2}) = 
(\mathrmbfit{struc}(\sigma)(M_{2}),\mathrmbfit{inv}(\sigma)(T_{2}))$
is a $\Sigma_{1}$-logic morphism between sound logics; 
that is,
$f : M_{1} \rightarrow \mathrmbfit{struc}(\sigma)(M_{2})$ is a $\Sigma_{1}$-structure morphism
so that
$\mathrmbfit{int}_{\Sigma_{1}}(M_{1}) \geq_{\Sigma_{1}} 
\mathrmbfit{int}_{\Sigma_{1}}(\mathrmbfit{struc}(\sigma)(M_{2})) =
\mathrmbfit{spec}(\sigma)(\mathrmbfit{int}_{\Sigma_{2}}(M_{2}))$, and 
$T_{1} \geq_{\Sigma_{1}} \mathrmbfit{spec}(\sigma)(T_{2})$ is a $\Sigma_{1}$-specification ordering.
\begin{center}
\begin{tabular}{c}
\\ \\ \\
\setlength{\unitlength}{0.55pt}
\begin{picture}(130,80)(0,-5)
\put(0,0){\begin{picture}(0,0)(0,0)
\qbezier[40](-45,40)(-45,110)(0,110)
\qbezier[40](45,40)(45,110)(0,110)
\qbezier[40](-45,40)(-45,-30)(0,-30)
\qbezier[40](45,40)(45,-30)(0,-30)
\put(0,125){\makebox(0,0){$\overset{\text{\scriptsize\itshape{soundness}}}{\overbrace{\hspace{30pt}}}$}}
\put(0,80){\makebox(0,0){\footnotesize{$T_{1}$}}}
\put(5,40){\makebox(0,0){\rule[-0.5pt]{0.5pt}{6pt}\,\footnotesize{$\vee_{\Sigma_{1}}$}}}
\put(0,0){\makebox(0,0){\footnotesize{$\mathrmbfit{int}_{\Sigma_{1}}(M_{1})$}}}
\put(0,-65){\makebox(0,0){$\underset{
\mbox{\scriptsize{\shortstack{${\langle{M_{1},T_{1}}\rangle}$\\where\\$M_{1}\;\models_{\Sigma_{1}}\;T_{1}$}}}
}{\underbrace{\hspace{30pt}}}$}}
\end{picture}}
\put(130,0){\begin{picture}(0,0)(0,0)
\qbezier[40](-45,40)(-45,110)(0,110)
\qbezier[40](45,40)(45,110)(0,110)
\qbezier[40](-45,40)(-45,-30)(0,-30)
\qbezier[40](45,40)(45,-30)(0,-30)
\put(0,125){\makebox(0,0){$\overset{\text{\scriptsize\itshape{soundness}}}{\overbrace{\hspace{30pt}}}$}}
\put(0,80){\makebox(0,0){\footnotesize{$\mathrmbfit{spec}(\sigma)(T_{2})$}}}
\put(5,40){\makebox(0,0){\rule[-0.5pt]{0.5pt}{6pt}\,\footnotesize{$\vee_{\Sigma_{1}}$}}}
\put(-44,0){\makebox(0,0)[l]{\footnotesize{$
\mathrmbfit{int}_{\Sigma_{1}}(\mathrmbfit{struc}(\sigma)(M_{2}))=
\mathrmbfit{spec}(\sigma)(\mathrmbfit{int}_{\Sigma_{2}}(M_{2}))$}}}
\put(0,-65){\makebox(0,0){$\underset{
\mbox{\scriptsize{\shortstack{
\hspace{50pt}$
(\mathrmbfit{struc}(\sigma)(M_{2}),\mathrmbfit{spec}(\sigma)(T_{2}))
$\\where\\\hspace{30pt}$\mathrmbfit{struc}(\sigma)(M_{2})\;\models_{\Sigma_{1}}\;\mathrmbfit{spec}(\sigma)(T_{2})$}}}
}{\underbrace{\hspace{30pt}}}$}}
\end{picture}}
\put(-50,80){\makebox(0,0)[r]{$
\shortstack{\footnotesize{$\Sigma$}\text{-\scriptsize\itshape{specification}}\\\text{\scriptsize\itshape{morphism}}}
\;\left\{\rule{0pt}{14pt}\right.$}}
\put(-50,0){\makebox(0,0)[r]{$
\shortstack{\footnotesize{$\Sigma$}\text{-\scriptsize\itshape{structure}}\\\text{\scriptsize\itshape{morphism}}}
\;\left\{\rule{0pt}{14pt}\right.$}}
\put(65,80){\makebox(0,0){\footnotesize{$\geq_{\Sigma_{1}}$}}}
\put(65,0){\makebox(0,0){\footnotesize{$\geq_{\Sigma_{1}}$}}}
\put(65,-51){\makebox(0,0){\footnotesize{$\stackrel{f}{\longrightarrow}$}}}
\end{picture}
\\ \\ \\ \\ \\
\end{tabular}
\end{center}

\newpage
\begin{definition}\label{model:category}
{\normalsize$\blacksquare$}
A model is a structure that satisfies (models) a specification.
Since sound logics represent satisfaction (truth),
they provide a background for models.
Conversely,
models elaborate sound logics.
There is an indexed category 
\footnotesize\[
\mathrmbfit{mod} 
: \mathrmbf{Spec}^{\mathrm{op}} \rightarrow \mathrmbf{Cat}.
\]\normalsize
\end{definition}
For any specification ${\langle{\Sigma,T}\rangle}$,
there is a fiber category 
$\mathrmbfit{mod}(\Sigma,T) \hookrightarrow \mathrmbfit{snd}(\Sigma) : M \mapsto (M.T)$
of ${\langle{\Sigma,T}\rangle}$-models.
An object $M \in \mathrmbfit{mod}(\Sigma,T)$
is a ${\langle{\Sigma,T}\rangle}$-model,
a $\Sigma$-structure 
that satisfies (models) $T$, $M \models_{\Sigma} T$.
A morphism $f : M \rightarrow M'$ is a ${\langle{\Sigma,T}\rangle}$-model morphism,
a $\Sigma$-structure morphism $f : M \rightarrow M'$, 
so that 
$\mathrmbfit{int}_{\Sigma}(M) \geq_{\Sigma} \mathrmbfit{int}_{\Sigma}(M')$. 
For any specification morphism $\sigma : {\langle{\Sigma_{1},T_{1}}\rangle} \rightarrow {\langle{\Sigma_{2},T_{2}}\rangle}$
consisting of a language morphism $\sigma : \Sigma_{1} \rightarrow \Sigma_{2}$
with $T_{1} \geq_{\Sigma_{1}} \mathrmbfit{spec}(\sigma)(T_{2})$,
there is a fiber functor
$\mathrmbfit{mod}(\sigma) : \mathrmbfit{mod}(\Sigma_{2},T_{2}) \rightarrow \mathrmbfit{mod}(\Sigma_{1},T_{1})$,
which maps a ${\langle{\Sigma_{2},T_{2}}\rangle}$-model $M_{2}$
to the ${\langle{\Sigma_{1},T_{1}}\rangle}$-model $\mathrmbfit{struc}(\sigma)(M_{2})$
and maps a ${\langle{\Sigma_{2},T_{2}}\rangle}$-model morphism $f_{2} : M_{2} \rightarrow M_{2}'$
to the ${\langle{\Sigma_{1},T_{1}}\rangle}$-model morphism 
$\mathrmbfit{struc}(\sigma)(f_{2}) : \mathrmbfit{struc}(\sigma)(M_{2}) \rightarrow \mathrmbfit{struc}(\sigma)(M_{2}')$~\footnote{\raggedright
$T_{1} \geq_{\Sigma_{1}} \mathrmbfit{spec}(\sigma)(T_{2}) \geq_{\Sigma_{1}} \mathrmbfit{spec}(\sigma)(\mathrmbfit{int}_{\Sigma_{2}}(M_{2}))
= \mathrmbfit{int}_{\Sigma_{1}}(\mathrmbfit{struc}(\sigma)(M_{2}))$,
since $\sigma$ is a specification morphism and $M_{2}$ models $T_{2}$;
$\mathrmbfit{int}_{\Sigma_{1}}(\mathrmbfit{struc}(\sigma)(M_{2}))
= \mathrmbfit{spec}(\sigma)(\mathrmbfit{int}_{\Sigma_{2}}(M_{2}))
\geq_{\Sigma_{1}} \mathrmbfit{spec}(\sigma)(\mathrmbfit{int}_{\Sigma_{2}}(M_{2}'))
= \mathrmbfit{int}_{\Sigma_{1}}(\mathrmbfit{struc}(\sigma)(M_{2}'))$,
since $\mathrmbfit{spec}(\sigma)$ is monotonic
and $f_{2}$ is a structure morphism.}.


The homogenization (Grothendieck construction) is the fibered subcategory $\mathrmbf{Mod}$ 
with an index projection functor $\mathrmbfit{pr} : \mathrmbf{Mod} \rightarrow \mathrmbf{Spec}$.
An object in $\mathrmbf{Mod}$ is a model
consisting of a triple ${\langle{\Sigma,T,M}\rangle}$,
where ${\langle{\Sigma,T}\rangle}$ is a specification
and $M \in \mathrmbfit{mod}(\Sigma,T)$ is a ${\langle{\Sigma,T}\rangle}$-model,
$M \models_{\Sigma} T$.
A morphism in $\mathrmbf{Mod}$ is a model morphism
$(\sigma,f) : {\langle{\Sigma_{1},T_{1},M_{1}}\rangle} \rightarrow {\langle{\Sigma_{2},T_{2},M_{2}}\rangle}$
consisting of a theory morphism 
$\sigma : (\Sigma_{1},T_{1}) \rightarrow (\Sigma_{2},T_{2})$
and a ${\langle{\Sigma_{1},T_{1}}\rangle}$-model morphism
$f : M_{1} \rightarrow \mathrmbfit{mod}(\sigma)(M_{2})$;
hence,
$\sigma : \Sigma_{1} \rightarrow \Sigma_{2}$ is a language morphism
with 
$T_{1} \geq_{\Sigma_{1}} \mathrmbfit{spec}(\sigma)(T_{2})$ is a $\Sigma_{1}$-specification ordering,
and
$f : M_{1} \rightarrow \mathrmbfit{struc}(\sigma)(M_{2})$ is a $\Sigma_{1}$-structure morphism,
so that
$\mathrmbfit{int}_{\Sigma_{1}}(M_{1}) \geq_{\Sigma_{1}} 
\mathrmbfit{int}_{\Sigma_{1}}(\mathrmbfit{struc}(\sigma)(M_{2})) =
\mathrmbfit{spec}(\sigma)(\mathrmbfit{int}_{\Sigma_{2}}(M_{2}))$.
Clearly,
we have the isomorphism
$\mathrmbf{Mod} \cong \mathrmbf{Snd}$,
the fibered category of models is the fibered category of sound logics.

\begin{fact}
The extent $\mathrmbfit{ext}(\Sigma,T)$
of any specification ${\langle{\Sigma,T}\rangle}$
is initial in the fiber category $\mathrmbfit{mod}(\Sigma,T)$.
\end{fact}

\begin{proof}
A $\Sigma$-model $M$ satisfies the $\Sigma$-specification $T$
when $T \geq_{\Sigma} \mathrmbfit{int}_{\Sigma}(M)$
iff there is a unique morphism $\mathrmbfit{ext}_{\Sigma}(T) \xrightarrow{\gtrdot} M$. 
\end{proof}


\begin{figure}
\begin{center}
\begin{tabular}{@{\hspace{60pt}}c@{\hspace{120pt}}c}
\setlength{\unitlength}{0.8pt}
\begin{picture}(0,130)(0,0)
\put(60,0){\begin{picture}(0,0)(0,0)
\put(40,130){\makebox(0,0){\footnotesize{$\mathrmbfit{log}(\Sigma_{1})$}}}
\put(18,100){\makebox(0,0)[l]{\footnotesize{$\mathrmbfit{snd}(\Sigma_{1})$}}}
\put(0,60){\makebox(0,0){\footnotesize{$\mathrmbfit{struc}(\Sigma_{1})$}}}
\put(80,60){\makebox(0,0){\footnotesize{$\mathrmbfit{spec}(\Sigma_{1})$}}}
\put(41,52){\makebox(0,0){\scriptsize{$\models_{\Sigma_{1}}$}}}
\put(30,60){\line(1,0){20}}
\put(40,115){\makebox(0,0){\footnotesize{$\cup$}}}
\put(24,120){\vector(-1,-2){24}}
\put(56,120){\vector(1,-2){24}}
\put(30,90){\vector(-1,-1){20}}
\put(50,90){\vector(1,-1){20}}
\qbezier[12](40,100)(40,115)(40,130)
\qbezier[24](0,60)(20,80)(40,100)
\qbezier[24](80,60)(60,80)(40,100)
\qbezier[34](0,60)(40,60)(80,60)
\end{picture}}
\put(0,0){\begin{picture}(0,0)(0,0)
\put(32,110){\makebox(0,0){\footnotesize{$\mathrmbfit{log}(\sigma)$}}}
\put(32,80){\makebox(0,0){\footnotesize{$\mathrmbfit{snd}(\sigma)$}}}
\put(14,32){\makebox(0,0)[l]{\footnotesize{$\mathrmbfit{struc}(\sigma)$}}}
\put(78,32){\makebox(0,0)[l]{\footnotesize{\shortstack{$\mathrmbfit{dir}(\sigma) 
\dashv 
\overset{\textstyle\mathrmbfit{spec}(\sigma)}{\overbrace{\mathrmbfit{inv}(\sigma)}}$}}}}
\put(0,80){\vector(2,1){80}}
\put(0,50){\vector(2,1){80}}
\put(120,50){\vector(-2,-1){80}}
\put(-40,10){\vector(2,1){80}}
\qbezier[60](100,130)(40,100)(-20,70)
\qbezier[60](100,100)(40,70)(-20,40)
\qbezier[60](60,60)(0,30)(-60,0)
\qbezier[60](140,60)(80,30)(20,0)
\qbezier[34](-60,0)(-20,0)(20,0)
\end{picture}}
\put(-60,-60){\begin{picture}(0,0)(0,0)
\put(40,130){\makebox(0,0){\footnotesize{$\mathrmbfit{log}(\Sigma_{2})$}}}
\put(18,100){\makebox(0,0)[l]{\footnotesize{$\mathrmbfit{snd}(\Sigma_{2})$}}}
\put(0,60){\makebox(0,0){\footnotesize{$\mathrmbfit{struc}(\Sigma_{2})$}}}
\put(80,60){\makebox(0,0){\footnotesize{$\mathrmbfit{spec}(\Sigma_{2})$}}}
\put(43,52){\makebox(0,0){\scriptsize{$\models_{\Sigma_{2}}$}}}
\put(30,60){\line(1,0){20}}
\put(40,115){\makebox(0,0){\footnotesize{$\cup$}}}
\put(24,120){\vector(-1,-2){24}}
\put(56,120){\vector(1,-2){24}}
\put(30,90){\vector(-1,-1){20}}
\put(50,90){\vector(1,-1){20}}
\qbezier[12](40,100)(40,115)(40,130)
\qbezier[24](0,60)(20,80)(40,100)
\qbezier[24](80,60)(60,80)(40,100)
\qbezier[34](0,60)(40,60)(80,60)
\end{picture}}
\end{picture}

&

\setlength{\unitlength}{1.0pt}
\begin{picture}(0,120)(0,0)
\put(-40,10){\begin{picture}(0,120)(0,0)
\put(40,70){\makebox(0,0){\normalsize{$\mathrmbf{Log}$}}}
\put(40,40){\makebox(0,0){\normalsize{$\mathrmbf{Snd}$}}}
\put(0,0){\makebox(0,0){\normalsize{$\mathrmbf{Struc}$}}}
\put(80,0){\makebox(0,0){\normalsize{$\mathrmbf{Spec}$}}}
\put(40,-8){\makebox(0,0){\small{$\mathrmbfit{int}$}}}
\thicklines
\put(22,0){\vector(1,0){36}}
\put(42,47){\vector(0,1){16}}\qbezier(42,47)(39,47)(36,47)\qbezier(36,47)(36,48.5)(36,50)
\put(27,60){\vector(-1,-2){24}}
\put(53,60){\vector(1,-2){24}}
\put(30,30){\vector(-1,-1){20}}
\put(50,30){\vector(1,-1){20}}
\thinlines
\qbezier[12](40,40)(40,55)(40,70)
\qbezier[24](0,0)(20,20)(40,40)
\qbezier[24](80,0)(60,20)(40,40)
\qbezier[34](0,0)(40,0)(80,0)
\end{picture}}
\end{picture}
\\
& \\
& \\
\itshape{indexed} & \itshape{fibered}
\\
& \\
& \\
\multicolumn{2}{p{360pt}}{
\begin{minipage}{360pt}
\begin{itemize}
\item 
For each language $\Sigma$,
truth (i.e., satisfaction) is concretely represented 
by the classification
$\mathrmbfit{cls}(\Sigma)={\langle{\mathrmbfit{spec}(\Sigma),\mathrmbfit{struc}(\Sigma),\models_{\Sigma}}\rangle}$
with the classification relation
between collections of objects of indexed categories.
An instance of classification 
$M \models_{\Sigma} T$
holds when the specification is more general than the intent of the structure
$T \geq_{\Sigma} \mathrmbfit{int}_{\Sigma}(M)$.
This asserts soundness for the $\Sigma$-logic
${\langle{M,T}\rangle} \in \mathrmbfit{log}(\Sigma) = \mathrmbfit{struc}(\Sigma){\times}\mathrmbfit{spec}(\Sigma)$.
The set of all sound $\Sigma$-logics,
the underlying set for the $\Sigma$-satisfaction relation, 
is also denoted by $\mathrmbfit{snd}(\Sigma) = {\models}_{\Sigma} \subseteq \mathrmbfit{log}(\Sigma)$.
Truth (i.e., satisfaction) is abstractly represented by membership 
${\langle{\Sigma,M,T}\rangle} \in \mathrmbf{Snd}$
in the fibered category of sound logics.
\newline
\item 
For each language morphism $\sigma : \Sigma_{1} \rightarrow \Sigma_{2}$,
truth (i.e., satisfaction) invariancy under change of notation is concretely represented 
by the infomorphism
\footnotesize\[
{\langle{\mathrmbfit{dir}(\sigma),\mathrmbfit{struc}(\sigma)}\rangle} :
\mathrmbfit{cls}(\Sigma_{1}) \rightleftarrows \mathrmbfit{cls}(\Sigma_{2}),
\]\normalsize
whose defining condition is equivalent to the naturality condition
$\mathrmbfit{struc}(\sigma) \circ \mathrmbfit{spec}(\Sigma_{1}) = 
 \mathrmbfit{spec}(\Sigma_{2}) \circ \mathrmbfit{spec}(\sigma)$
for the intent indexed functor 
$\mathrmbfit{int} : \mathrmbfit{struc} \Rightarrow \mathrmbfit{spec}$ 
(see subsection~\ref{intent})~\footnote{$T_1 \geq_{\Sigma_{1}} \mathrmbfit{int}_{\Sigma_{1}}(\mathrmbfit{struc}(\sigma)(M_{2}))
\;\text{\underline{iff}}\; \mathrmbfit{struc}(\sigma)(M_{2}) \models_{\Sigma_{1}} T_1
\;\text{\underline{iff}}\; M_{2} \models_{\Sigma_{2}} \mathrmbfit{dir}(\sigma)(T_1)
\newline
\;\text{\underline{iff}}\; \mathrmbfit{dir}(\sigma)(T_1) \geq_{\Sigma_{2}} \mathrmbfit{int}_{\Sigma_{2}}(M_{2})
\;\text{\underline{iff}}\; T_1 \geq_{\Sigma_{1}} \mathrmbfit{inv}(\sigma)(\mathrmbfit{int}_{\Sigma_{2}}(M_{2}))$}.
Truth (i.e., satisfaction) invariancy is abstractly represented by the intent fibered functor 
$\mathrmbfit{int} : \mathrmbf{Struc} \rightarrow \mathrmbf{Spec}$~\footnote{ These comments demonstrate that 
the classification logical environment $\mathtt{Cls}$ 
is applicable at the meta-level to any logical environment.}.
\newline
\end{itemize}
\end{minipage}
}
\\ & \\
\end{tabular}
\end{center}
\caption{Truth Architecture}
\label{truth:architecture}
\end{figure}

%% file: parta.tex
\newpage
\section{The Strict Aspect}

\input{intent}

\input{fiber-meets}


%% file: intent.tex
\subsection{Intent}\label{intent}

\begin{meta-axiom}\label{intent:functor}
{\normalsize$\blacksquare$}
There is an intent indexed functor
\footnotesize\[
\mathrmbfit{int} 
: \mathrmbfit{struc} \Rightarrow \mathrmbfit{spec} 
: \mathrmbf{Lang}^{\mathrm{op}} \rightarrow \mathrmbf{Cat}.
\]\normalsize
\end{meta-axiom}
For any language $\Sigma$,
there is a fiber functor
$\mathrmbfit{int}_{\Sigma} : \mathrmbfit{struc}(\Sigma) \rightarrow \mathrmbfit{spec}(\Sigma)$
that maps 
a $\Sigma$-structure $M \in \mathrmbfit{struc}(\Sigma)$ 
to its intent $\Sigma$-specification 
$\mathrmbfit{int}_{\Sigma}(M) \in \mathrmbfit{spec}(\Sigma)$, and
maps a $\Sigma$-structure morphism
$f : M \rightarrow M'$ in $\mathrmbfit{struc}(\Sigma)$
to its generalization-specialization $\Sigma$-specification ordering
$\mathrmbfit{int}_{\Sigma}(M) \geq_{\Sigma} \mathrmbfit{int}_{\Sigma}(M')$.
We interpret the latter as the assertion that $f : M \rightarrow M'$ points downward in the concept lattice.
Hence,
the fiber category $\mathrmbfit{struc}(\Sigma)$ 
generalizes 
the concept lattice of closed $\Sigma$-specifications.
We say that a $\Sigma$-structure $M$ \emph{satisfies} a $\Sigma$-specification $T$,
denoted
$M \models_{\Sigma} T$,
when $T$ is more general than the intent of $M$:
$T \geq_{\Sigma} \mathrmbfit{int}_{\Sigma}(M)$~\footnote{This is 
an abstract (point-less) axiomatization of specifications and the intents of structures.
The usual version of specifications given in the theory of institutions is point-wise, 
and starts with the relation of satisfaction.
In this paper truth semantics and satisfaction are given more algebraically.}.
For any language morphism $\sigma : \Sigma_{1} \rightarrow \Sigma_{2}$ in $\mathrmbf{Lang}$,
intent satisfies the naturality commutative diagram
\begin{center}
\begin{tabular}{c}
\\
\setlength{\unitlength}{0.65pt}
\begin{picture}(100,80)(0,0)
\put(0,80){\makebox(0,0){\footnotesize{$\mathrmbfit{struc}(\Sigma_{1})$}}}
\put(100,80){\makebox(0,0){\footnotesize{$\mathrmbfit{spec}(\Sigma_{1})$}}}
\put(0,0){\makebox(0,0){\footnotesize{$\mathrmbfit{struc}(\Sigma_{2})$}}}
\put(100,0){\makebox(0,0){\footnotesize{$\mathrmbfit{spec}(\Sigma_{2})$}}}
\put(50,92){\makebox(0,0){\scriptsize{$\mathrmbfit{int}_{\Sigma_{1}}$}}}
\put(105,40){\makebox(0,0)[l]{\scriptsize{$\mathrmbfit{spec}(\sigma)$}}}
\put(-5,40){\makebox(0,0)[r]{\scriptsize{$\mathrmbfit{struc}(\sigma)$}}}
\put(50,-14){\makebox(0,0){\scriptsize{$\mathrmbfit{int}_{\Sigma_{2}}$}}}
\put(0,20){\vector(0,1){40}}
\put(100,20){\vector(0,1){40}}
\put(35,80){\vector(1,0){30}}
\put(35,0){\vector(1,0){30}}
\end{picture}
\\ \\
\end{tabular}
\end{center}
or pointwise,
$\mathrmbfit{int}_{\Sigma_{1}}(\mathrmbfit{struc}(\sigma)(M_{2}))
= \mathrmbfit{inv}(\sigma)(\mathrmbfit{int}_{\Sigma_{2}}(M_{2}))$
for every target structure $M_{2} \in \mathrmbfit{struc}(\Sigma_{2})$.
This asserts the invariancy of truth (i.e., satisfaction) under change of notation
\footnotesize\[
\mathrmbfit{struc}(\sigma)(M_{2}) \models_{\Sigma_{1}} T_{1} 
\;\;\text{iff}\;\;
M_{2} \models_{\Sigma_{2}} \mathrmbfit{dir}(\sigma)(T_{1})
\]\normalsize
for every source specification $T_{2} \in \mathrmbfit{spec}(\Sigma_{1})$,
since
$T_{1} \geq_{\Sigma_{1}} \mathrmbfit{int}_{\Sigma_{1}}(\mathrmbfit{struc}(\sigma)(M_{2}))$ iff 
$T_{1} \geq_{\Sigma_{1}} \mathrmbfit{inv}(\sigma)(\mathrmbfit{int}_{\Sigma_{2}}(M_{2}))$ iff 
$\mathrmbfit{dir}(\sigma)(T_{1}) \geq_{\Sigma_{2}} \mathrmbfit{int}_{\Sigma_{2}}(M_{2})$.

The homogenization (Grothendieck construction) is the intent fibered functor 
$\mathrmbfit{int} : \mathrmbf{Struc} \rightarrow \mathrmbf{Spec}$ 
that commutes with index projections
$\mathrmbfit{int} \circ \mathrmbfit{pr} = \mathrmbfit{pr}$.
Intent
maps a structure ${\langle{\Sigma,M}\rangle}$
to the specification
$\mathrmbfit{int}(\Sigma,M) = {\langle{\Sigma,\mathrmbfit{int}_{\Sigma}(M)}\rangle}$,
and maps a structure morphism~\footnote{Recall that
a structure morphism,
is a pair $(\sigma,h) : (\Sigma_{1},M_{1}) \rightarrow (\Sigma_{2},M_{2})$, 
where $\sigma : \Sigma_{1} \rightarrow \Sigma_{2}$ is a language morphism 
and $h : M_{1} \rightarrow \mathrmbfit{struc}(\sigma)(M_{2})$ is a $\Sigma_{1}$-structure morphism,
a morphism in the fiber category $\mathrmbfit{struc}(\Sigma_{1})$,
implying the ordering
$\mathrmbfit{int}_{\Sigma_{1}}(M_{1}) \geq_{\Sigma_{1}} 
\mathrmbfit{int}_{\Sigma_{1}}(\mathrmbfit{struc}(\sigma)(M_{2}))$.}
$(\sigma,h) : (\Sigma_{1},M_{1}) \rightarrow (\Sigma_{2},M_{2})$
to the specification morphism
$\mathrmbfit{int}(\sigma,h) = \sigma :
\mathrmbfit{int}(\Sigma_{1},M_{1}) = {\langle{\Sigma_{1},\mathrmbfit{int}_{\Sigma_{1}}(M_{1})}\rangle} \rightarrow 
{\langle{\Sigma_{2},\mathrmbfit{int}_{\Sigma_{2}}(M_{2})}\rangle} = \mathrmbfit{int}(\Sigma_{2},M_{2})$,
where
$\mathrmbfit{int}_{\Sigma_{1}}(M_{1}) \geq_{\Sigma_{1}} 
\mathrmbfit{inv}(\sigma)(\mathrmbfit{int}_{\Sigma_{2}}(M_{2})) =
\mathrmbfit{int}_{\Sigma_{1}}(\mathrmbfit{struc}(\sigma)(M_{2}))$.


\begin{definition}\label{natural:logic:functor}
{\normalsize$\blacksquare$}
There is a natural logic indexed functor
\footnotesize\[
\widehat{\mathrmbfit{nat}} 
= (\mathrmbfit{1}_{\mathrmbf{struc}},\mathrmbfit{int}) 
: \mathrmbfit{struc} \Rightarrow \mathrmbfit{log}
: \mathrmbf{Lang}^{\mathrm{op}} \rightarrow \mathrmbf{Cat}.
\]\normalsize
that mediates the identity and intent indexed functors,
so that 
\footnotesize\[
\widehat{\mathrmbfit{nat}} \bullet \mathrmbfit{pr}_{0} 
= \mathrmbfit{1}_{\mathrmbf{struc}}
\;\;\text{and}\;\; 
\widehat{\mathrmbfit{nat}} \bullet \mathrmbfit{pr}_{1} 
= \mathrmbfit{int}.
\]\normalsize
\end{definition}
For any language $\Sigma$,
there is a natural logic fiber functor
$\widehat{\mathrmbfit{nat}}_{\Sigma} 
= (\mathrmbfit{1}_{\mathrmbf{struc}(\Sigma)},\mathrmbfit{int}_{\Sigma}) 
: \mathrmbfit{struc}(\Sigma) \rightarrow \mathrmbfit{log}(\Sigma)$
that maps a $\Sigma$-structure $M \in \mathrmbfit{struc}(\Sigma)$ 
to the (sound) $\Sigma$-logic $(M,\mathrmbfit{int}_{\Sigma}(M)) \in \mathrmbfit{log}(\Sigma)$, and
maps a $\Sigma$-structure morphism
$f : M \rightarrow M'$ in $\mathrmbfit{struc}(\Sigma)$
to the $\Sigma$-logic morphism
$\widehat{\mathrmbfit{nat}}(f) = f : 
(M,\mathrmbfit{int}_{\Sigma}(M)) \rightarrow
(M',\mathrmbfit{int}_{\Sigma}(M')) 
$.
For any language morphism $\sigma : \Sigma_{1} \rightarrow \Sigma_{2}$ in $\mathrmbf{Lang}$,
natural logic satisfies the naturality commutative diagram
\begin{center}
\begin{tabular}{c}
\\
\setlength{\unitlength}{0.6pt}
\begin{picture}(100,80)(0,0)
\put(-5,80){\makebox(0,0){\footnotesize{$\mathrmbfit{struc}(\Sigma_{1})$}}}
\put(100,80){\makebox(0,0){\footnotesize{$\mathrmbfit{log}(\Sigma_{1})$}}}
\put(-5,0){\makebox(0,0){\footnotesize{$\mathrmbfit{struc}(\Sigma_{2})$}}}
\put(100,0){\makebox(0,0){\footnotesize{$\mathrmbfit{log}(\Sigma_{2})$}}}
\put(55,92){\makebox(0,0){\scriptsize{$\widehat{\mathrmbfit{nat}}_{\Sigma_{1}}$}}}
\put(105,40){\makebox(0,0)[l]{\scriptsize{$\mathrmbfit{log}(\sigma)$}}}
\put(-5,40){\makebox(0,0)[r]{\scriptsize{$\mathrmbfit{struc}(\sigma)$}}}
\put(55,-14){\makebox(0,0){\scriptsize{$\widehat{\mathrmbfit{nat}}_{\Sigma_{2}}$}}}
\put(0,20){\vector(0,1){40}}
\put(100,20){\vector(0,1){40}}
\put(35,80){\vector(1,0){30}}
\put(35,0){\vector(1,0){30}}
\end{picture}
\\ \\
\end{tabular}
\end{center}
or pointwise,
$
\widehat{\mathrmbfit{nat}}_{\Sigma_{1}}(\mathrmbfit{struc}(\sigma)(M_{2}))
= (\mathrmbfit{struc}(\sigma)(M_{2}),\mathrmbfit{int}_{\Sigma_{1}}(\mathrmbfit{struc}(\sigma)(M_{2})))
= (\mathrmbfit{struc}(\sigma)(M_{2}),\mathrmbfit{inv}(\sigma)(\mathrmbfit{int}_{\Sigma_{2}}(M_{2}))))
= \mathrmbfit{log}(\sigma)(M_{2},\mathrmbfit{int}_{\Sigma_{2}}(M_{2}))
= \mathrmbfit{log}(\sigma)(\widehat{\mathrmbfit{nat}}_{\Sigma_{2}}(M_{2}))
$
for every target structure $M_{2} \in \mathrmbfit{struc}(\Sigma_{2})$.

The homogenization (Grothendieck construction) is the natural logic fibered functor 
$\widehat{\mathrmbfit{nat}} : \mathrmbf{Struc} \rightarrow \mathrmbf{Log}$ 
that commutes with index projections
$\widehat{\mathrmbfit{nat}} \circ \mathrmbfit{pr} = \mathrmbfit{pr}$.
Natural logic
maps a structure ${\langle{\Sigma,M}\rangle}$
to the logic
$\widehat{\mathrmbfit{nat}}(\Sigma,M) = {\langle{\Sigma,M,\mathrmbfit{int}_{\Sigma}(M)}\rangle}$,
and maps a structure morphism~\footnote{Recall that
a structure morphism,
is a pair $(\sigma,h) : (\Sigma_{1},M_{1}) \rightarrow (\Sigma_{2},M_{2})$, 
where $\sigma : \Sigma_{1} \rightarrow \Sigma_{2}$ is a language morphism 
and $h : M_{1} \rightarrow \mathrmbfit{struc}(\sigma)(M_{2})$ is a $\Sigma_{1}$-structure morphism,
a morphism in the fiber category $\mathrmbfit{struc}(\Sigma_{1})$,
implying the ordering
$\mathrmbfit{int}_{\Sigma_{1}}(M_{1}) \geq_{\Sigma_{1}} 
\mathrmbfit{int}_{\Sigma_{1}}(\mathrmbfit{struc}(\sigma)(M_{2}))$.}
$(\sigma,h) : (\Sigma_{1},M_{1}) \rightarrow (\Sigma_{2},M_{2})$
to the logic morphism
$\widehat{\mathrmbfit{nat}}(\sigma,h) = {\langle{\sigma,h}\rangle} :
\widehat{\mathrmbfit{nat}}(\Sigma_{1},M_{1}) = {\langle{\Sigma_{1},M_{1},\mathrmbfit{int}_{\Sigma_{1}}(M_{1})}\rangle} \rightarrow 
{\langle{\Sigma_{2},M_{2},\mathrmbfit{int}_{\Sigma_{2}}(M_{2})}\rangle} = \widehat{\mathrmbfit{nat}}(\Sigma_{2},M_{2})$,
where
$\mathrmbfit{int}_{\Sigma_{1}}(M_{1}) \geq_{\Sigma_{1}} 
\mathrmbfit{inv}(\sigma)(\mathrmbfit{int}_{\Sigma_{2}}(M_{2})) =
\mathrmbfit{int}_{\Sigma_{1}}(\mathrmbfit{struc}(\sigma)(M_{2}))$.

\begin{definition}\label{restricted:natural:logic:functor}
{\normalsize$\blacksquare$}
There is a restricted natural logic indexed functor
\footnotesize\[
\mathrmbfit{nat}
: \mathrmbfit{struc} \Rightarrow \mathrmbfit{snd}
: \mathrmbf{Lang}^{\mathrm{op}} \rightarrow \mathrmbf{Cat}.
\]\normalsize
that is the restriction of natural logic to sound logics,
so that
\footnotesize\[
\mathrmbfit{nat} \bullet \mathrmbfit{inc} = \widehat{\mathrmbfit{nat}} 
: \mathrmbfit{struc} \Rightarrow \mathrmbfit{log}.
\]\normalsize
\end{definition}
Clearly,
\footnotesize\[
\mathrmbfit{nat} \bullet \mathrmbfit{pr}_{0} 
= \mathrmbfit{1}_{\mathrmbf{Struc}}
\;\;\text{and}\;\; 
\mathrmbfit{nat} \bullet \mathrmbfit{pr}_{1} 
= \mathrmbfit{int}.
\]\normalsize

\begin{proposition}
There is an adjunction (reflection)
\footnotesize\[
\pi 
= {\langle{\mathrmbfit{pr}_{0} \dashv \mathrmbfit{nat},\eta,1}\rangle} 
: \mathrmbf{Snd} \rightarrow \mathrmbf{Struc}.
\]\normalsize
\end{proposition}

\begin{proof}
The counit is the identity
$\mathrmbfit{nat} \bullet \mathrmbfit{pr}_{0} = \mathrmbfit{1}_{\mathrmbf{Struc}}$.
The unit
$\eta : \mathrmbfit{1}_{\mathrmbf{Snd}} \Rightarrow \mathrmbfit{pr}_{0} \bullet \mathrmbfit{nat}$
has 
the sound logic morphism
$\eta_{\langle{\Sigma,M,T}\rangle} = {\langle{1_{\Sigma},1_{M}}\rangle} :
{\langle{\Sigma,M,T}\rangle} \rightarrow {\langle{\Sigma,M,\mathrmbfit{int}_{\Sigma}(M)}\rangle}$
as its ${\langle{\Sigma,M,T}\rangle}^{\mathrm{th}}$ component.
For any sound logic morphism~\footnote{Recall that
a logic morphism is a pair 
$(\sigma,f) : (\Sigma_{1},M_{1},T_{1}) \rightarrow (\Sigma_{2},M_{2},T_{2})$, 
where $\sigma : \Sigma_{1} \rightarrow \Sigma_{2}$ is a language morphism, and
$f : (M_{1},T_{1}) \rightarrow 
\mathrmbfit{log}(\sigma)(M_{2},T_{2}) = (\mathrmbfit{struc}(\sigma)(M_{2}),\mathrmbfit{inv}(\sigma)(T_{2}))$
is a $\Sigma_{1}$-logic morphism; 
that is,
$f : M_{1} \rightarrow \mathrmbfit{struc}(\sigma)(M_{2})$ is a $\Sigma_{1}$-structure morphism, and 
$T_{1} \geq_{\Sigma_{1}} \mathrmbfit{inv}(\sigma)(T_{2})$ is a $\Sigma_{1}$-specification ordering.
Hence,
a logic morphism $(\sigma,f) : (\Sigma_{1},M_{1},T_{1}) \rightarrow (\Sigma_{2},M_{2},T_{2})$
consists of
a structure morphism $(\sigma,f) : (\Sigma_{1},M_{1}) \rightarrow (\Sigma_{2},M_{2})$ and
a specification morphism $\sigma : (\Sigma_{1},T_{1}) \rightarrow (\Sigma_{2},T_{2})$.}
$(\sigma,f) : 
{\langle{\Sigma_{1},M_{1},T_{1}}\rangle} \rightarrow 
{\langle{\Sigma_{2},M_{2},\mathrmbfit{int}_{\Sigma_{2}}(M_{2})}\rangle} = \mathrmbfit{nat}(\Sigma_{2},M_{2})$,
the structure morphism
$\mathrmbfit{pr}_{0}(\sigma,f) = (\sigma,f)
: {\langle{\Sigma_{1},M_{1}}\rangle} \rightarrow {\langle{\Sigma_{2},M_{2}}\rangle}$
is the unique structure morphism
$\mathrmbfit{pr}_{0}(\Sigma_{1},M_{1},T_{1}) = {\langle{\Sigma_{1},M_{1}}\rangle}
\rightarrow {\langle{\Sigma_{2},M_{2}}\rangle}$
such that
$\eta_{\Sigma_{1},M_{1},T_{1}} \cdot \mathrmbfit{nat}(\sigma,f) 
= (\sigma,f)$.
\rule{5pt}{5pt}
\end{proof}

\newpage

\begin{eg}
Consider the classification logical environment $\mathtt{Cls}$.

\begin{itemize}

\item 
In the classification logical environment $\mathtt{Cls}$ 
(used in the metatheory of information flow),
the fibered category of $\mathrmbf{Th}$ theories and theory morphisms
with the underlying type projection functor
$\mathrmbfit{typ} : \mathrmbf{Th} \rightarrow \mathrmbf{Set}$
serves as the 
adjointly indexed preorder of specifications.
An object (specification) in $\mathrmbf{Th}$ is a theory ${\langle{Y,T}\rangle}={\langle{Y,{\vdash}}\rangle}$,
and a (specification) morphism in $\mathrmbf{Th}$ is a theory morphism
$f : {\langle{Y_{1},T_{1}}\rangle}={\langle{Y_{1},{\vdash}_{1}}\rangle} \rightarrow {\langle{Y_{2},{\vdash}_{2}}\rangle}={\langle{Y_{2},T_{2}}\rangle}$.
The category $\mathrmbf{Th}$ is the homogenization (Grothendieck construction) of 
an adjointly indexed preorder of theories
$\mathrmbfit{th} : \mathrmbf{Set}^{\mathrm{op}} \rightarrow \mathrmbf{Adj}$.
For any set $Y$,
there is a fiber preorder $\mathrmbfit{th}(Y) 
= {\langle{\mathrmbfit{th}(\Sigma),\geq_{Y}}\rangle}$,
where an object $T = {\vdash} \in \mathrmbfit{th}(Y)$ is a $Y$-theory~\footnote{
Let $Y$ be any (type) set.
A $Y$-sequent is a pair $(\Gamma,\Delta) \in \mathrmbfit{seq}(Y)$ of subsets of types $\Gamma,\Delta \subseteq Y$.
A $Y$-theory $T = {\vdash} \in \mathrmbfit{th}(Y) = {\wp}\mathrmbfit{seq}(Y)$ is a subset of $Y$-sequents.
We use the notation $\Gamma \vdash \Delta$ for a sequent of the theory $(\Gamma,\Delta) \in T$.
Given a $Y$-classification ${\langle{X,\models}\rangle}$,
define the type set of an instance $x \in X$ 
to be $\tau(x) = \{y \in Y \mid x \models y\}$.
An instance $x \in X$ satisfies a $Y$-sequent $(\Gamma,\Delta)$,
denoted $x \models_{Y} (\Gamma,\Delta)$,
when $x \models^{\forall} \Gamma$ implies $x \models^{\exists} \Delta$;
that is,
when $\tau(x) \supseteq \Gamma$ implies $\tau(x) \cap \Delta \not= \emptyset$
(otherwise, it is a counterexample to the sequent).
A $Y$-classification ${\langle{X,\models}\rangle}$ satisfies a $Y$-sequent $(\Gamma,\Delta)$,
denoted ${\langle{X,\models}\rangle} \models_{Y} (\Gamma,\Delta)$,
when $x \models_{Y} (\Gamma,\Delta)$ for all instances $x \in X$.
The intersection of all regular theories containing a $Y$-theory $T = {\vdash}$, 
denoted $T^{\scriptscriptstyle\bullet} = {\vdash}^{\scriptscriptstyle\bullet}$,
is called its closure. 
This is the smallest regular theory containing $T$.
Two $Y$-theories $T_{1},T_{2} \in \mathrmbfit{th}(Y)$ are ordered by entailment, 
symbolized by $T_{1} \leq_{Y} T_{2}$
when $T_{1}^{\scriptscriptstyle\bullet} \supseteq T_{2}^{\scriptscriptstyle\bullet}$.
This defines a fiber preorder $\mathrmbfit{th}(Y) 
= {\langle{\mathrmbfit{th}(\Sigma),\geq_{Y}}\rangle}$.
using the reverse entailment ordering on $Y$-theories. 
}
and a morphism 
$f : {\langle{Y_{1},T_{1}}\rangle}={\langle{Y_{1},{\vdash}_{1}}\rangle} \rightarrow {\langle{Y_{2},{\vdash}_{2}}\rangle}={\langle{Y_{2},T_{2}}\rangle}$
is a (type) function $f : Y_{1} \rightarrow Y_{2}$ where
$\mathrmbfit{dir}(f)(T_{1}) \geq_{2} T_{2}$
equivalently
$T_{1} \geq_{1} \mathrmbfit{inv}(f)(T_{2})$~\footnote{
For each function (language translation) $f : Y_{1} \rightarrow Y_{2}$,
sentence translation along $f$ is direct image squared on types
$\mathrmbfit{seq}(f)
: \mathrmbfit{seq}(Y_{1}) \rightarrow \mathrmbfit{seq}(Y_{2})
: (\Gamma_{1},\Delta_{1}) \mapsto ({\wp}f(\Gamma_{1}),{\wp}f(\Delta_{1}))$.
The inverse flow of a $Y_{2}$-theory $T_{2} = {\vdash}_{2}$ under $f$,
written $\mathrmbfit{inv}(f)(T_{2})$,
is the $Y_{1}$-theory $T_{1} = {\vdash}_{1}$
with $\Gamma \vdash_{1} \Delta$ when ${\wp}f(\Gamma) \vdash_{2}^{\scriptscriptstyle\bullet} {\wp}f(\Delta)$.
The direct flow of a $Y_{1}$-theory $T_{1} = {\vdash}_{1}$ under $f$,
written $\mathrmbfit{dir}(f)(T_{1})$,
is the $Y_{2}$-theory $T_{2} = {\vdash}_{2}$
with 
${\wp}f(\Gamma_{1}) \vdash_{Y_{2}} {\wp}f(\Delta_{1})$
when $\Gamma_{1} \vdash_{Y_{1}} \Delta_{1}$.
}.

\item 
In the classification logical environment $\mathtt{Cls}$,
for any set $Y$,
there is a fiber functor
$\mathrmbfit{int}_{Y} : \mathrmbfit{cls}(Y) \rightarrow \mathrmbfit{th}(Y)$
that maps 
a $Y$-classification ${\langle{X,\models}\rangle} \in \mathrmbfit{cls}(Y)$ 
to its intent $Y$-theory 
$\mathrmbfit{int}_{Y}(X,\models) \in \mathrmbfit{th}(Y)$
consisting of the collection of 
all $Y$-sequents $(\Gamma,\Delta)$ satisfied by all instances $x \in X$~\footnote{
The $Y$-theory generated by a $Y$-classification ${\langle{X,\models}\rangle}$,
denoted $\mathrmbfit{int}_{Y}(X,\models)$,
is the set of all $Y$-sequents satisfied by ${\langle{X,\models}\rangle}$.
The $Y$-theory $\mathrmbfit{int}_{Y}(X,\models)$
satisfies various axioms, such as identity, weakening, global cut, finite cut and partition.
A theory is regular when it satisfies these axioms.
It can be shown that all regular theories are of the form $\mathrmbfit{int}_{Y}(X,\models)$.
}, and
maps a $Y$-infomorphism
$g : {\langle{X,\models}\rangle} \rightleftarrows {\langle{X',\models'}\rangle}$
in $\mathrmbfit{cls}(Y)$
to its generalization-specialization $Y$-theory ordering
$\mathrmbfit{int}_{Y}(X,\models) = {\vdash} \geq_{Y} {\vdash}' = \mathrmbfit{int}_{Y}(X',\models')$;
that is,
${\vdash} \subseteq {\vdash}'$~\footnote{
Since $g(x') \models y$ iff $x' \models' y$,
the $Y$-theories
$\mathrmbfit{int}_{Y}(X,\models) = {\vdash}$ and
$\mathrmbfit{int}_{Y}(X',\models') = {\vdash}'$
are ordered
${\vdash} \subseteq {\vdash}'$.
For suppose
$\Gamma {\vdash} \Delta$;
that is,
$x \models (\Gamma,\Delta)$ for all $x \in X$.
In particular,
$g(x') \models (\Gamma,\Delta)$ for all $x' \in X'$.
Now,
$x' \models'^{\forall} \Gamma$ iff $g(x') \models^{\forall} \Gamma$
implies 
$g(x') \models^{\exists} \Delta$ iff $x' \models'^{\exists} \Delta$.
Hence,
$x' \models' (\Gamma,\Delta)$
for all $x' \in X'$. 
Hence,
$\Gamma {\vdash}' \Delta$.}

\end{itemize}

\end{eg}

\begin{eg}
Consider the diagram logical environment $\mathtt{Dgm}$ 
(used in the metatheory of sketches).
For any graph $G$,
a diagram $\langle D, \mathrmbf{C} \rangle \in \mathrmbf{dgm}(G)$ 
with adjoint $\mathrmbfit{D} : G^{\ast} \rightarrow \mathrmbf{C}$
satisfies an equation $\epsilon \in \mathrmbfit{eqn}(G)$,
symbolized by 
\[
\langle D, \mathrmbf{C} \rangle \models_{G} \epsilon,
\]
when $\mathrmbfit{D}(\epsilon_{0}) = \mathrmbfit{D}(\epsilon_{1})$.
\footnote{We also say that
$\epsilon$ holds in $\langle D, \mathrmbf{C} \rangle$, 
$\epsilon$ is true in $\langle D, \mathrmbf{C} \rangle$, 
or $\langle D, \mathrmbf{C} \rangle$ models $\epsilon$.}
We could use the notation `$\epsilon_{0} =_{\mathcal{D}} \epsilon_{1}$'
for such satisfaction
in terms of the indexed diagram $\mathcal{D} = \langle G, D, \mathrmbf{C} \rangle$.
This notion of satisfaction (for ``linear'' sketches) is like a reduced version of satisfaction for arbitrary sketchs.
The diagram $D : G \rightarrow |\mathrmbf{C}|$ is commutative when $D$ satisfies all equations of $G$;
that is, $D \models_G \epsilon$ for all $\epsilon \in \mathrmbfit{eqn}(G)$.
For any graph morphism $H : G \rightarrow G'$, 
if $\langle D', \mathrmbf{C} \rangle \in \mathrmbf{dgm}(G')$ is a target diagram,  
and $\epsilon \in \mathrmbfit{eqn}(G)$ is a source equation,
then
\[
\mathrmbfit{dgm}(H)(D', \mathrmbf{C}) \models_{G} \epsilon,
\text{ iff }
\langle D', \mathrmbf{C} \rangle \models_{G'} \mathrmbfit{eqn}(H)(\epsilon),
\]
since
$\mathrmbfit{dgm}(H)(D', \mathrmbf{C}) = \langle H \circ D', \mathrmbf{C} \rangle$,
$\mathrmbfit{eqn}(H)(\epsilon) = \epsilon \circ |H^{\ast}|$,
and hence
$(H^{\ast} \circ \mathrmbfit{D}')(\epsilon_{0}) = (H^{\ast} \circ \mathrmbfit{D}')(\epsilon_{1}) 
\text{ iff }
\mathrmbfit{D}'(H^{\ast}(\epsilon_{0})) = \mathrmbfit{D}'(H^{\ast}(\epsilon_{1}))$.
Thus,
satisfaction is invariant under change of notation.
\begin{center}
\begin{tabular}{c}
\\
\setlength{\unitlength}{0.6pt}
\begin{picture}(120,80)(-30,0)
\put(-100,38){\makebox(0,0){\small{$\mathrmbf{ppr}$}}}
\put(4,80){\makebox(0,0){\small{$|G^{\ast}|$}}}
\put(4,0){\makebox(0,0){\small{$|G'^{\ast}|$}}}
\put(120,38){\makebox(0,0){\small{$|\mathrmbf{C}|$}}}
\put(-80,52){\vector(3,1){60}}
\put(-80,28){\vector(3,-1){60}}
\put(23,73){\vector(3,-1){74}}
\put(23,5){\vector(3,1){74}}
\put(0,65){\vector(0,-1){50}}
\put(-50,76){\makebox(0,0){\footnotesize$\epsilon$}}
\put(-64,0){\makebox(0,0){\footnotesize$\epsilon \circ |H^{\ast}|$}}
\put(-6,40){\makebox(0,0)[r]{\footnotesize$|H^{\ast}|$}}
\put(70,82){\makebox(0,0){\footnotesize$|H^{\ast} {\circ}\, \mathrmbfit{D}'|$}}
\put(63,0){\makebox(0,0){\footnotesize$|\mathrmbfit{D}'|$}}
\end{picture}
\end{tabular}
\end{center}
\end{eg}

%% file: fiber-meets.tex
\subsection{Fiber Meets}

In the usual theory of institutions,
the specification fiber preorders for any language
have all finite meets and joins.
However,
logical environments are an abstraction of institutions,
and only the minimal collection of specification meets~\footnote{Just as for 
the concept lattice intent notation $\mathrmbfit{int}_{\Sigma}$,
in the fiber ${\langle{\mathrmbfit{spec}(\Sigma),\geq_{\Sigma}}\rangle}$ 
of all $\Sigma$-specifications with reverse entailment order
(generalized reverse concept lattice order),
we also use 
the concept lattice bottom notation ${\bot}_{\Sigma}$ for the empty meet and 
the concept lattice join notation $\vee_{\Sigma}$ for binary meets.}.
are axiomatized.

\begin{meta-axiom}
{\normalsize$\blacksquare$}
There is 
a bottom indexed functor
\footnotesize\[
{\bot} 
: \mathrmbfit{lang} \Rightarrow \mathrmbfit{spec} 
: \mathrmbf{Lang}^{\mathrm{op}} \rightarrow \mathrmbf{Cat}.
\]\normalsize
\end{meta-axiom}
For any language $\Sigma$,
there is a fiber functor
${\bot}_{\Sigma} : \mathrmbfit{lang}(\Sigma) = \mathrmbf{1} \rightarrow \mathrmbfit{spec}(\Sigma)$
that maps 
the single object $\bullet \in \mathrmbfit{lang}(\Sigma) = \mathrmbf{1}$ 
to the bottom $\Sigma$-specification 
${\bot}_{\Sigma}(\bullet) \in \mathrmbfit{spec}(\Sigma)$, and
the single (identity) morphism
$1_{\bullet} : {\bullet} \rightarrow {\bullet}$ in $\mathrmbfit{lang}(\Sigma) = \mathrmbf{1}$
to the identity $\Sigma$-specification ordering
${\bot}_{\Sigma}(\bullet) \geq_{\Sigma} {\bot}_{\Sigma}(\bullet)$.
For any language morphism $\sigma : \Sigma_{1} \rightarrow \Sigma_{2}$,
bottom satisfies the naturality commutative diagram
\begin{center}
\begin{tabular}{c}
\\
\setlength{\unitlength}{0.65pt}
\begin{picture}(100,80)(0,0)
\put(-10,80){\makebox(0,0){\footnotesize{$\mathrmbf{1}=\mathrmbfit{lang}(\Sigma_{1})$}}}
\put(100,80){\makebox(0,0){\footnotesize{$\mathrmbfit{spec}(\Sigma_{1})$}}}
\put(-10,0){\makebox(0,0){\footnotesize{$\mathrmbf{1}=\mathrmbfit{lang}(\Sigma_{2})$}}}
\put(100,0){\makebox(0,0){\footnotesize{$\mathrmbfit{spec}(\Sigma_{2})$}}}
\put(50,92){\makebox(0,0){\scriptsize{${\bot}_{\Sigma_{1}}$}}}
\put(106,40){\makebox(0,0)[l]{\scriptsize{$\mathrmbfit{spec}(\sigma)=\mathrmbfit{inv}(\sigma)$}}}
\put(-4,40){\makebox(0,0)[r]{\scriptsize{$\mathrmit{1}_{\mathrmbf{1}}=\mathrmbfit{lang}(\sigma)$}}}
\put(50,-14){\makebox(0,0){\scriptsize{${\bot}_{\Sigma_{2}}$}}}
\put(0,20){\vector(0,1){40}}
\put(100,20){\vector(0,1){40}}
\put(35,80){\vector(1,0){30}}
\put(35,0){\vector(1,0){30}}
\end{picture}
\\ \\
\end{tabular}
\end{center}
or pointwise,
${\bot}_{\Sigma_{1}}(\mathrmbfit{lang}(\sigma)(\bullet))
= {\bot}_{\Sigma_{1}}(\bullet)
= \mathrmbfit{inv}(\sigma)({\bot}_{\Sigma_{2}}(\bullet))$
the single object $\bullet \in \mathrmbfit{lang}(\Sigma_{2})$.
This asserts that the specification inverse image operator preserves bottom,
${\bot}_{\Sigma_{1}} = 
\mathrmbfit{inv}(\sigma)({\bot}_{\Sigma_{2}})$.

The homogenization (Grothendieck construction) is the bottom fibered functor 
${\bot} : \mathrmbf{Lang} \rightarrow \mathrmbf{Spec}$ 
that commutes with index projections
${\bot} \circ \mathrmbfit{pr} = \mathrmit{1}_{\mathrmbf{Lang}}$.
Bottom
maps a language $\Sigma$
to the specification
${\bot}(\Sigma) = {\langle{\Sigma,{\bot}_{\Sigma}}\rangle}$,
and maps a language morphism
$\sigma : \Sigma_{1} \rightarrow \Sigma_{2}$
to the specification morphism
${\bot}(\sigma) = \sigma :
{\bot}(\Sigma_{1}) = {\langle{\Sigma_{1},{\bot}_{\Sigma_{1}}}\rangle} \rightarrow 
{\langle{\Sigma_{2},{\bot}_{\Sigma_{2}}}\rangle} = {\bot}(\Sigma_{2})$.
Clearly,
the bottom functor commutes with projections
${\bot} \circ \mathrmbfit{pr} = \mathrmbfit{1}_{\mathrmbf{Lang}}$.
Terminality, 
the fact that bottom
${\bot}_{\Sigma}(\bullet_{\Sigma})$ is the terminal or top object in the specification fiber (reverse entailment) preorder $\mathrmbf{Spec}(\Sigma)$,
is axiomatized by the following.
\begin{meta-axiom}
{\normalsize$\blacksquare$}
There is an adjunction (reflection)
\footnotesize\[
\kappa 
= {\langle{\mathrmbfit{pr} \dashv \bot,\eta,1}\rangle} 
: \mathrmbf{Spec} \rightarrow \mathrmbf{Lang}.
\]\normalsize
\end{meta-axiom}
The counit is the identity
${\bot} \circ \mathrmbfit{pr} = \mathrmbfit{1}_{\mathrmbf{Lang}}$.
The unit is a natural transformation
$\eta : \mathrmit{1}_{\mathrmbf{Spec}} \Rightarrow \mathrmbfit{pr} \circ {\bot}$.
For any specification ${\langle{\Sigma,T}\rangle}$,
the ${\langle{\Sigma,T}\rangle}^{\mathrm{th}}$ component of the unit natural transformation $\eta$
is the specification morphism
$\eta_{\Sigma,T} = 1_{\Sigma} : {\langle{\Sigma,T}\rangle} \rightarrow {\langle{\Sigma,{\bot}_{\Sigma}}\rangle}$,
which asserts the $\Sigma$-ordering $T \geq_{\Sigma} {\bot}_{\Sigma}$.
Existence of this specification morphism alone
shows that there is a $\Sigma$-specification ordering from any $\Sigma$-structure $T$ to ${\bot}_{\Sigma}$.
The full statement of universality (or couniversality) adds nothing more.

\begin{definition}
{\normalsize$\blacksquare$}
There is a bottom indexed functor
\footnotesize\[
{\bot} = \mathrmbfit{pr} \circ {\bot}
: \mathrmbfit{struc} \Rightarrow \mathrmbfit{lang} \Rightarrow \mathrmbfit{spec} 
: \mathrmbf{Lang}^{\mathrm{op}} \rightarrow \mathrmbf{Cat}.
\]\normalsize
\end{definition}
For any language $\Sigma$,
there is a (constant) bottom fiber functor
${\bot}_{\Sigma} : \mathrmbfit{struc}(\Sigma) \rightarrow \mathrmbfit{spec}(\Sigma)$
that maps 
a $\Sigma$-structure $M \in \mathrmbfit{struc}(\Sigma)$ 
to the bottom $\Sigma$-specification ${\bot}_{\Sigma}$, and
maps a $\Sigma$-structure morphism
$f : M \rightarrow M'$ in $\mathrmbfit{struc}(\Sigma)$
with
$\mathrmbfit{int}_{\Sigma}(M) \geq_{\Sigma} \mathrmbfit{int}_{\Sigma}(M')$,
to the trivial identity $\Sigma$-specification ordering
${\bot}_{\Sigma} \geq_{\Sigma} {\bot}_{\Sigma}$.
For any language morphism $\sigma : \Sigma_{1} \rightarrow \Sigma_{2}$ in $\mathrmbf{Lang}$,
bottom satisfies the naturality commutative diagram
$\mathrmbfit{struc}(\sigma) \cdot {\bot}_{\Sigma_{1}} = {\bot}_{\Sigma_{2}} \cdot \mathrmbfit{spec}(\sigma)$;
or pointwise,
${\bot}_{\Sigma_{1}}(\mathrmbfit{struc}(\sigma)(M_{2}))
= {\bot}_{\Sigma_{1}}
= \mathrmbfit{inv}(\sigma)({\bot}_{\Sigma_{2}})
= \mathrmbfit{inv}(\sigma)({\bot}_{\Sigma_{2}}(M_{2}))$
for every target structure $M_{2} \in \mathrmbfit{struc}(\Sigma_{2})$.

The homogenization (Grothendieck construction) is the bottom fibered functor 
${\bot} = \mathrmbfit{pr} \circ {\bot} :	\mathrmbf{Struc} \rightarrow \mathrmbf{Spec}$
that commutes with index projections
${\bot} \circ \mathrmbfit{pr} = \mathrmbfit{pr}$.

\begin{definition}
{\normalsize$\blacksquare$}
There is a bottom logic indexed functor
\footnotesize\[
{\bot} 
= (\mathrmbfit{1}_{\mathrmbf{struc}},{\bot}) 
: \mathrmbfit{struc} \Rightarrow \mathrmbfit{log}
: \mathrmbf{Lang}^{\mathrm{op}} \rightarrow \mathrmbf{Cat}.
\]\normalsize
that mediates the identity and bottom indexed functors,
so that 
\footnotesize\[
{\bot} \bullet \mathrmbfit{pr}_{0} 
= \mathrmbfit{1}_{\mathrmbf{struc}}
\;\;\text{and}\;\; 
{\bot} \bullet \mathrmbfit{pr}_{1} 
= {\bot}.
\]\normalsize
\end{definition}
For any language $\Sigma$,
there is a bottom logic fiber functor
${\bot}_{\Sigma} 
= (\mathrmbfit{1}_{\mathrmbf{struc}(\Sigma)},{\bot}_{\Sigma}) 
: \mathrmbfit{struc}(\Sigma) \rightarrow \mathrmbfit{log}(\Sigma)$
that maps a $\Sigma$-structure $M \in \mathrmbfit{struc}(\Sigma)$ 
to the $\Sigma$-logic $(M,{\bot}_{\Sigma}) \in \mathrmbfit{log}(\Sigma)$, and
maps a $\Sigma$-structure morphism
$f : M \rightarrow M'$ in $\mathrmbfit{struc}(\Sigma)$
to the $\Sigma$-logic morphism
${\bot}(f) = f : (M,{\bot}_{\Sigma}) \rightarrow (M',{\bot}_{\Sigma})$.

For any language morphism $\sigma : \Sigma_{1} \rightarrow \Sigma_{2}$ in $\mathrmbf{Lang}$,
bottom logic satisfies the naturality commutative diagram
\begin{center}
\begin{tabular}{c}
\\
\setlength{\unitlength}{0.6pt}
\begin{picture}(100,80)(0,0)
\put(0,80){\makebox(0,0){\footnotesize{$\mathrmbfit{struc}(\Sigma_{1})$}}}
\put(100,80){\makebox(0,0){\footnotesize{$\mathrmbfit{log}(\Sigma_{1})$}}}
\put(0,0){\makebox(0,0){\footnotesize{$\mathrmbfit{struc}(\Sigma_{2})$}}}
\put(100,0){\makebox(0,0){\footnotesize{$\mathrmbfit{log}(\Sigma_{2})$}}}
\put(50,92){\makebox(0,0){\scriptsize{${\bot}_{\Sigma_{1}}$}}}
\put(105,40){\makebox(0,0)[l]{\scriptsize{$\mathrmbfit{log}(\sigma)$}}}
\put(-5,40){\makebox(0,0)[r]{\scriptsize{$\mathrmbfit{struc}(\sigma)$}}}
\put(50,-14){\makebox(0,0){\scriptsize{${\bot}_{\Sigma_{2}}$}}}
\put(0,20){\vector(0,1){40}}
\put(100,20){\vector(0,1){40}}
\put(35,80){\vector(1,0){30}}
\put(35,0){\vector(1,0){30}}
\end{picture}
\\ \\
\end{tabular}
\end{center}
or pointwise,
$
{\bot}_{\Sigma_{1}}(\mathrmbfit{struc}(\sigma)(M_{2}))
= (\mathrmbfit{struc}(\sigma)(M_{2}),{\bot}_{\Sigma_{1}})
= (\mathrmbfit{struc}(\sigma)(M_{2}),\mathrmbfit{inv}(\sigma)({\bot}_{\Sigma_{2}})))
= \mathrmbfit{log}(\sigma)(M_{2},{\bot}_{\Sigma_{2}})
= \mathrmbfit{log}(\sigma)({\bot}_{\Sigma_{2}}(M_{2}))
$
for every target structure $M_{2} \in \mathrmbfit{struc}(\Sigma_{2})$.

The homogenization (Grothendieck construction) is the bottom logic fibered functor 
${\bot} : \mathrmbf{Struc} \rightarrow \mathrmbf{Log}$ 
that commutes with index projections
${\bot} \circ \mathrmbfit{pr} = \mathrmbfit{pr}$.
Bottom logic
maps a structure ${\langle{\Sigma,M}\rangle}$
to the logic
${\bot}(\Sigma,M) = {\langle{\Sigma,M,{\bot}_{\Sigma}}\rangle}$,
and maps a structure morphism
$(\sigma,h) : (\Sigma_{1},M_{1}) \rightarrow (\Sigma_{2},M_{2})$
to the logic morphism
${\bot}(\sigma,h) = {\langle{\sigma,h}\rangle} :
{\bot}(\Sigma_{1},M_{1}) = {\langle{\Sigma_{1},M_{1},{\bot}_{\Sigma_{1}}}\rangle} \rightarrow 
{\langle{\Sigma_{2},M_{2},{\bot}_{\Sigma_{2}}}\rangle} = {\bot}(\Sigma_{2},M_{2})$,
where
${\bot}_{\Sigma_{1}} \geq_{\Sigma_{1}} 
\mathrmbfit{inv}(\sigma)({\bot}_{\Sigma_{2}}) =
{\bot}_{\Sigma_{1}}$.

\begin{proposition}
There is an adjunction (reflection)
\footnotesize\[
\pi^{\!\scriptscriptstyle\bot} 
= {\langle{\mathrmbfit{pr}_{0} \dashv {\bot},\eta,1}\rangle} 
: \mathrmbf{Log} \rightarrow \mathrmbf{Struc}.
\]\normalsize
\end{proposition}

\begin{proof}
The unit
$\eta : \mathrmbfit{1}_{\mathrmbf{Snd}} \Rightarrow \mathrmbfit{pr}_{0} \circ {\bot}$
has 
the logic morphism
$\eta_{\langle{\Sigma,M,T}\rangle} 
= {\langle{1_{\Sigma},1_{M}}\rangle} :
{\langle{\Sigma,M,T}\rangle} \rightarrow {\langle{\Sigma,M,{\bot}_{\Sigma}}\rangle}$
as its ${\langle{\Sigma,M,T}\rangle}^{\mathrm{th}}$ component.
\rule{5pt}{5pt}
\end{proof}

\begin{meta-axiom}
{\normalsize$\blacksquare$}
There is a join indexed functor
\footnotesize\[
\mathrmbfit{join} 
: \mathrmbfit{log} \Rightarrow \mathrmbfit{spec} 
: \mathrmbf{Lang}^{\mathrm{op}} \rightarrow \mathrmbf{Cat}.
\]\normalsize
\end{meta-axiom}
For any language $\Sigma$,
there is a fiber functor
$\mathrmbfit{join}_{\Sigma} : \mathrmbfit{log}(\Sigma) \rightarrow \mathrmbfit{spec}(\Sigma)$
that maps 
a $\Sigma$-logic $(M,T) \in \mathrmbfit{log}(\Sigma)$ 
to the join $\Sigma$-specification
$\mathrmbfit{int}_{\Sigma}(M) \vee_{\Sigma} T$, and
maps a $\Sigma$-logic morphism
$f : (M,T) \rightarrow (M',T')$ in $\mathrmbfit{log}(\Sigma)$,
consisting of a $\Sigma$-structure morphism $f : M \rightarrow M'$
with
$\mathrmbfit{int}_{\Sigma}(M) \geq_{\Sigma} \mathrmbfit{int}_{\Sigma}(M')$ 
and a $\Sigma$-specification ordering $T \geq_{\Sigma} T'$,
to the $\Sigma$-specification ordering
$\mathrmbfit{join}_{\Sigma}(M,T) = \mathrmbfit{int}_{\Sigma}(M){\vee_{\Sigma}}T \geq_{\Sigma} 
\mathrmbfit{int}_{\Sigma}(M'){\vee_{\Sigma}}T' = \mathrmbfit{join}_{\Sigma}(M',T')$.
Note that a $\Sigma$-logic $(M,T)$ is sound iff
$\mathrmbfit{join}_{\Sigma}(M,T) = \mathrmbfit{int}_{\Sigma}(M) \vee_{\Sigma} T \cong T$.

For any language morphism $\sigma : \Sigma_{1} \rightarrow \Sigma_{2}$ in $\mathrmbf{Lang}$,
join satisfies the naturality commutative diagram
\begin{center}
\begin{tabular}{c}
\\
\setlength{\unitlength}{0.6pt}
\begin{picture}(100,80)(0,0)
\put(0,80){\makebox(0,0){\footnotesize{$\mathrmbfit{log}(\Sigma_{1})$}}}
\put(100,80){\makebox(0,0){\footnotesize{$\mathrmbfit{spec}(\Sigma_{1})$}}}
\put(0,0){\makebox(0,0){\footnotesize{$\mathrmbfit{log}(\Sigma_{2})$}}}
\put(100,0){\makebox(0,0){\footnotesize{$\mathrmbfit{spec}(\Sigma_{2});$}}}
\put(50,92){\makebox(0,0){\scriptsize{$\mathrmbfit{join}_{\Sigma_{1}}$}}}
\put(105,40){\makebox(0,0)[l]{\scriptsize{$\mathrmbfit{spec}(\sigma)$}}}
\put(-5,40){\makebox(0,0)[r]{\scriptsize{$\mathrmbfit{log}(\sigma)$}}}
\put(50,-14){\makebox(0,0){\scriptsize{$\mathrmbfit{join}_{\Sigma_{2}}$}}}
\put(0,20){\vector(0,1){40}}
\put(100,20){\vector(0,1){40}}
\put(35,80){\vector(1,0){30}}
\put(35,0){\vector(1,0){30}}
\end{picture}
\\ \\
\end{tabular}
\end{center}
or pointwise,
$
\mathrmbfit{join}_{\Sigma_{1}}(\mathrmbfit{log}(\sigma)(M_{2},T_{2}))
= \mathrmbfit{join}_{\Sigma_{1}}(\mathrmbfit{struc}(\sigma)(M_{2}),\mathrmbfit{inv}(\sigma)(T_{2}))
= \mathrmbfit{int}_{\Sigma_{1}}(\mathrmbfit{struc}(\sigma)(M_{2})) \vee_{\Sigma_{1}} \mathrmbfit{inv}(\sigma)(T_{2})
= \mathrmbfit{inv}(\sigma)(\mathrmbfit{int}_{\Sigma_{2}}(M_{2})) \vee_{\Sigma_{1}} \mathrmbfit{inv}(\sigma)(T_{2})
= \mathrmbfit{inv}(\sigma)(\mathrmbfit{int}_{\Sigma_{2}}(M_{2}) \vee_{\Sigma_{2}} T_{2})
= \mathrmbfit{inv}(\sigma)(\mathrmbfit{join}_{\Sigma_{2}}(M_{2},T_{2}))
$
for every target logic $(M_{2},T_{2}) \in \mathrmbfit{log}(\Sigma_{2})$.~\footnote{Recall that
direct and inverse flow of specifications
are adjoint monotonic functions w.r.t. specification order:
$\mathrmbfit{inv}(\sigma)(T_{2}) \leq_{\Sigma_{1}} T_{1}$ 
\underline{iff}
$T_{2} \leq_{\Sigma_{2}} \mathrmbfit{dir}(\sigma)(T_{1})$;
symbolized by
$\langle \mathrmbfit{inv}(\sigma) \dashv \mathrmbfit{dir}(\sigma) \rangle 
: \mathrmbfit{spec}(\Sigma_{2})^{\mathrm{op}} \rightarrow \mathrmbfit{spec}(\Sigma_{1})^{\mathrm{op}}$
between entailment preorders
or by
$\langle \mathrmbfit{dir}(\sigma) \dashv \mathrmbfit{inv}(\sigma) \rangle 
: \mathrmbfit{spec}(\Sigma_{1}) \rightarrow \mathrmbfit{spec}(\Sigma_{2})$
between opposite preorders.
Hence,
inverse flow preeserves all joins and direct flow preserves all meets (with respect to entailment order).}

The homogenization (Grothendieck construction) is the join fibered functor 
$\mathrmbfit{join} :	\mathrmbf{Log} \rightarrow \mathrmbf{Spec}$
that commutes with index projections
$\mathrmbfit{join} \circ \mathrmbfit{pr} = \mathrmbfit{pr}$.
Join maps a logic ${\langle{\Sigma,M,T}\rangle}$
to the specification ${\langle{\Sigma,\mathrmbfit{join}_{\Sigma}(M,T)}\rangle}$,
and
maps a logic morphism
$(\sigma,f) : (\Sigma_{1},M_{1},T_{1}) \rightarrow (\Sigma_{2},M_{2},T_{2})$, 
where $\sigma : \Sigma_{1} \rightarrow \Sigma_{2}$ is a language morphism, and
$\mathrmbfit{int}_{\Sigma_{1}}(M_{1}) \geq_{\Sigma_{1}} 
\mathrmbfit{int}_{\Sigma_{1}}(\mathrmbfit{struc}(\sigma)(M_{2}))=
\mathrmbfit{inv}(\sigma)(\mathrmbfit{int}_{\Sigma_{2}}(M_{2}))$ 
and 
$T_{1} \geq_{\Sigma_{1}} \mathrmbfit{inv}(\sigma)(T_{2})$ are $\Sigma_{1}$-specification orderings,
to the specification morphism
$\mathrmbfit{join}(\sigma,f) = \sigma : 
\mathrmbfit{join}(\Sigma_{1},M_{1},T_{1}) 
= {\langle{\Sigma_{1},\mathrmbfit{join}_{\Sigma_{1}}(M_{1},T_{1})}\rangle}
= {\langle{\Sigma_{1},\mathrmbfit{int}_{\Sigma_{1}}(M_{1}){\vee_{\Sigma_{1}}}T_{1}}\rangle}
\rightarrow 
{\langle{\Sigma_{2},\mathrmbfit{int}_{\Sigma_{2}}(M_{2}){\vee_{\Sigma_{2}}}T_{2}}\rangle}
= {\langle{\Sigma_{2},\mathrmbfit{join}_{\Sigma_{2}}(M_{2},T_{2})}\rangle}
= \mathrmbfit{join}(\Sigma_{2},M_{2},T_{2})$,
where $\sigma : \Sigma_{1} \rightarrow \Sigma_{2}$ is a language morphism, and
$\mathrmbfit{join}(\Sigma_{1},M_{1},T_{1})=
\mathrmbfit{int}_{\Sigma_{1}}(M_{1}){\vee_{\Sigma}}T_{1} 
\geq_{\Sigma_{1}} 
\mathrmbfit{inv}(\sigma)(\mathrmbfit{int}_{\Sigma_{2}}(M_{2})){\vee_{\Sigma_{1}}}\mathrmbfit{inv}(\sigma)(T_{2})=
\mathrmbfit{inv}(\sigma)(\mathrmbfit{int}_{\Sigma_{2}}(M_{2}){\vee_{\Sigma_{1}}}T_{2})=
\mathrmbfit{inv}(\sigma)(\mathrmbfit{join}(\Sigma_{2},M_{2},T_{2}))$ 
is a $\Sigma_{1}$-specification ordering.

\begin{definition}
{\normalsize$\blacksquare$}
There is a join logic indexed functor
\footnotesize\[
\widehat{\mathrmbfit{join}} 
= (\mathrmbfit{pr}_{0},\mathrmbfit{join}) 
: \mathrmbfit{log} \Rightarrow \mathrmbfit{log}
: \mathrmbf{Lang}^{\mathrm{op}} \rightarrow \mathrmbf{Cat}.
\]\normalsize
that mediates the projection and join indexed functors,
so that 
\footnotesize\[
\widehat{\mathrmbfit{join}} \bullet \mathrmbfit{pr}_{0} 
= \mathrmbfit{pr}_{0}
\;\;\text{and}\;\; 
\widehat{\mathrmbfit{join}} \bullet \mathrmbfit{pr}_{1} 
= \mathrmbfit{join}.
\]\normalsize
\end{definition}
For any language $\Sigma$,
there is a join logic fiber functor
$\widehat{\mathrmbfit{join}}_{\Sigma} 
= (\mathrmbfit{pr}_{\mathrmbf{struc},\Sigma},\mathrmbfit{join}_{\Sigma}) 
: \mathrmbfit{log}(\Sigma) \rightarrow \mathrmbfit{log}(\Sigma)$
that maps a $\Sigma$-logic $(M,T) \in \mathrmbfit{log}(\Sigma)$ 
to the join (sound) $\Sigma$-logic $(M,\mathrmbfit{join}_{\Sigma}(M,T)) \in \mathrmbfit{log}(\Sigma)$, and
maps a $\Sigma$-logic morphism
$f : (M,T) \rightarrow (M',T')$ in $\mathrmbfit{log}(\Sigma)$
to the $\Sigma$-logic morphism (between sound logics)
$\widehat{\mathrmbfit{join}}_{\Sigma}(f) = f : 
(M,\mathrmbfit{join}_{\Sigma}(M,T)) \rightarrow
(M',\mathrmbfit{join}_{\Sigma}(M',T')) 
$.

For any language morphism $\sigma : \Sigma_{1} \rightarrow \Sigma_{2}$ in $\mathrmbf{Lang}$,
join logic satisfies the naturality commutative diagram
\begin{center}
\begin{tabular}{c}
\\
\setlength{\unitlength}{0.6pt}
\begin{picture}(100,80)(0,0)
\put(0,80){\makebox(0,0){\footnotesize{$\mathrmbfit{log}(\Sigma_{1})$}}}
\put(100,80){\makebox(0,0){\footnotesize{$\mathrmbfit{log}(\Sigma_{1})$}}}
\put(0,0){\makebox(0,0){\footnotesize{$\mathrmbfit{log}(\Sigma_{2})$}}}
\put(100,0){\makebox(0,0){\footnotesize{$\mathrmbfit{log}(\Sigma_{2})$}}}
\put(50,92){\makebox(0,0){\scriptsize{$\widehat{\mathrmbfit{join}}_{\Sigma_{1}}$}}}
\put(105,40){\makebox(0,0)[l]{\scriptsize{$\mathrmbfit{log}(\sigma)$}}}
\put(-5,40){\makebox(0,0)[r]{\scriptsize{$\mathrmbfit{log}(\sigma)$}}}
\put(50,-14){\makebox(0,0){\scriptsize{$\widehat{\mathrmbfit{join}}_{\Sigma_{2}}$}}}
\put(0,20){\vector(0,1){40}}
\put(100,20){\vector(0,1){40}}
\put(35,80){\vector(1,0){30}}
\put(35,0){\vector(1,0){30}}
\end{picture}
\\ \\
\end{tabular}
\end{center}
or pointwise,
$\widehat{\mathrmbfit{join}}_{\Sigma_{1}}(\mathrmbfit{log}(\sigma)(M_{2},T_{2}))
= \widehat{\mathrmbfit{join}}_{\Sigma_{1}}(\mathrmbfit{struc}(\sigma)(M_{2}),
\mathrmbfit{inv}(\sigma)(T_{2}))
\newline
= (\mathrmbfit{struc}(\sigma)(M_{2}),\mathrmbfit{join}_{\Sigma_{1}}(\mathrmbfit{struc}(\sigma)(M_{2}),
\mathrmbfit{inv}(\sigma)(T_{2})))
\newline
= (\mathrmbfit{struc}(\sigma)(M_{2}),\mathrmbfit{join}_{\Sigma_{1}}(\mathrmbfit{log}(\sigma)(M_{2},T_{2})))
\newline
= (\mathrmbfit{struc}(\sigma)(M_{2}),\mathrmbfit{inv}(\sigma)(\mathrmbfit{join}_{\Sigma_{2}}(M_{2},T_{2})))
= \mathrmbfit{log}(\sigma)(M_{2},\mathrmbfit{join}_{\Sigma_{2}}(M_{2},T_{2}))
\newline
= \mathrmbfit{log}(\sigma)(\widehat{\mathrmbfit{join}}_{\Sigma_{2}}(M_{2},T_{2}))$
for every target logic $(M_{2},T_{2}) \in \mathrmbfit{log}(\Sigma_{2})$.

The homogenization (Grothendieck construction) is the
join logic fibered functor 
$\widehat{\mathrmbfit{join}} : \mathrmbf{Log} \rightarrow \mathrmbf{Log}$ 
that commutes with index projections
$\widehat{\mathrmbfit{join}} \circ \mathrmbfit{pr} = \mathrmbfit{pr}$.
Join maps a logic ${\langle{\Sigma,M,T}\rangle}$
to the (sound) logic 
${\langle{\Sigma,M,\mathrmbfit{join}_{\Sigma}(M,T)}\rangle}=
{\langle{\Sigma,M,\mathrmbfit{int}_{\Sigma}(M){\vee_{\Sigma}}T}\rangle}$,
and
maps a logic morphism
$(\sigma,f) : (\Sigma_{1},M_{1},T_{1}) \rightarrow (\Sigma_{2},M_{2},T_{2})$
to the morphism (between sound logics)
$\widehat{\mathrmbfit{join}}(\sigma,f) = (\sigma,f) : 
\widehat{\mathrmbfit{join}}(\Sigma_{1},M_{1},T_{1}) 
= {\langle{\Sigma_{1},M_{1},\mathrmbfit{join}_{\Sigma_{1}}(M_{1},T_{1})}\rangle}
= {\langle{\Sigma_{1},M_{1},\mathrmbfit{int}_{\Sigma_{1}}(M_{1}){\vee_{\Sigma}}T_{1}}\rangle}
\rightarrow 
{\langle{\Sigma_{2},M_{2},\mathrmbfit{int}_{\Sigma_{2}}(M_{2}){\vee_{\Sigma}}T_{2}}\rangle}
= {\langle{\Sigma_{2},M_{2},\mathrmbfit{join}_{\Sigma_{2}}(M_{2},T_{2})}\rangle}
= \widehat{\mathrmbfit{join}}(\Sigma_{2},M_{2},T_{2})$.


\begin{definition}
{\normalsize$\blacksquare$}
There is a restricted join logic indexed functor
\footnotesize\[
\overset{\scriptscriptstyle\vee}{\mathrmbfit{res}}
: \mathrmbfit{log} \Rightarrow \mathrmbfit{snd}
: \mathrmbf{Lang}^{\mathrm{op}} \rightarrow \mathrmbf{Cat}.
\]\normalsize
that is the restriction of join logic to sound logics,
so that
\footnotesize\[
\overset{\scriptscriptstyle\vee}{\mathrmbfit{res}} \bullet \mathrmbfit{inc} = \widehat{\mathrmbfit{join}} 
: \mathrmbfit{log} \Rightarrow \mathrmbfit{log}.
\]\normalsize
\end{definition}
Clearly,
\footnotesize\[
\overset{\scriptscriptstyle\vee}{\mathrmbfit{res}} \bullet \mathrmbfit{pr}_{0} 
= \mathrmbfit{pr}_{0} 
\;\;\text{and}\;\; 
\overset{\scriptscriptstyle\vee}{\mathrmbfit{res}} \bullet \mathrmbfit{pr}_{1} 
= \mathrmbfit{join}.
\]\normalsize

\begin{proposition}
There is an adjunction (coreflection)
\footnotesize\[
\overset{\scriptscriptstyle\vee}{\rho} 
= {\langle{\mathrmbfit{inc} \dashv \overset{\scriptscriptstyle\vee}{\mathrmbfit{res}},1,\varepsilon}\rangle} 
: \mathrmbf{Snd} \rightarrow \mathrmbf{Log}.
\]\normalsize
\end{proposition}

\begin{proof}
The unit is the identity
$\mathrmbfit{1}_{\mathrmbf{Snd}}
= \mathrmbfit{inc} \bullet \overset{\scriptscriptstyle\vee}{\mathrmbfit{res}}$.
The counit
$\varepsilon : \overset{\scriptscriptstyle\vee}{\mathrmbfit{res}} \bullet \mathrmbfit{inc} = \widehat{\mathrmbfit{join}} 
\Rightarrow \mathrmbfit{1}_{\mathrmbf{Log}}$
has 
the logic morphism
$\varepsilon_{\langle{\Sigma,M,T}\rangle} = {\langle{1_{\Sigma},1_{M}}\rangle} :
{\langle{\Sigma,M,\mathrmbfit{join}_{\Sigma}(M,T)}\rangle}
= {\langle{\Sigma,M,\mathrmbfit{int}_{\Sigma}(M){\vee_{\Sigma}}T}\rangle}
\rightarrow {\langle{\Sigma,M,T}\rangle}$
as its ${\langle{\Sigma,M,T}\rangle}^{\mathrm{th}}$ component.
Let 
${\langle{\sigma,f}\rangle} : {\langle{\Sigma_{1},M_{1},T_{1}}\rangle} \rightarrow {\langle{\Sigma_{2},M_{2},T_{2}}\rangle}$ 
be any logic morphism
whose source logic is sound.
Hence,
$\sigma : \Sigma_{1} \rightarrow \Sigma_{2}$ is a language morphism,
$f : M_{1} \rightarrow \mathrmbfit{struc}(\sigma)(M_{2})$ is a $\Sigma_{1}$-structure morphism, with
$\mathrmbfit{int}_{\Sigma_{1}}(M_{1}) 
\geq_{\Sigma_{1}} \mathrmbfit{int}_{\Sigma_{1}}(\mathrmbfit{struc}(\sigma)(M_{2}))
= \mathrmbfit{inv}(\sigma)(\mathrmbfit{int}_{\Sigma_{2}}(M_{2}))$,
$T_{1} \geq_{\Sigma_{1}} \mathrmbfit{inv}(\sigma)(T_{2})$ is a $\Sigma_{1}$-specification ordering, and
$T_{1} \geq_{\Sigma_{1}} \mathrmbfit{int}_{\Sigma_{1}}(M_{1})$ (soundnesss).
We want to show that
there is a unique morphism between sound logics
${\langle{\sigma',f'}\rangle} : {\langle{\Sigma_{1},M_{1},T_{1}}\rangle} \rightarrow \overset{\scriptscriptstyle\vee}{\mathrmbfit{res}}(\Sigma_{2},M_{2},T_{2})
= {\langle{\Sigma_{2},M_{2},\mathrmbfit{join}_{\Sigma_{2}}(M_{2},T_{2})}\rangle}$
such that 
${\langle{\sigma',f'}\rangle} \cdot \varepsilon_{\langle{\Sigma_{2},M_{2},T_{2}}\rangle} 
= {\langle{\sigma',f'}\rangle} \cdot {\langle{1_{\Sigma_{2}},1_{M_{2}}}\rangle} 
= {\langle{\sigma,f}\rangle}$;
that is,
such that $\sigma' = \sigma$ and $f' = f$.
Clearly, uniqueness must hold.
Hence,
we want to show that
${\langle{\sigma,f}\rangle} : {\langle{\Sigma_{1},M_{1},T_{1}}\rangle} \rightarrow 
{\langle{\Sigma_{2},M_{2},\mathrmbfit{join}_{\Sigma_{2}}(M_{2},T_{2})}\rangle}$ 
is a logic morphism.
But,
$T_{1} 
\geq_{\Sigma_{1}} 
\mathrmbfit{inv}(\sigma)(\mathrmbfit{join}_{\Sigma_{2}}(M_{2},T_{2}))
= \mathrmbfit{inv}(\sigma)(\mathrmbfit{int}_{\Sigma_{2}}(M_{2}) {\vee_{\Sigma_{2}}} T_{2})
= \mathrmbfit{inv}(\sigma)(\mathrmbfit{int}_{\Sigma_{2}}(M_{2})) {\vee_{\Sigma_{1}}} \mathrmbfit{inv}(\sigma)(T_{2})$
by the orderings above.
\rule{5pt}{5pt}
\end{proof}

%% file: partb.tex
\newpage
\section{The Lax Aspect} 

\input{extent}
\input{fiber-sums}

%% file: extent.tex
\subsection{Extent}\label{extent}

\begin{meta-axiom}\label{extent:functor}
{\normalsize$\blacksquare$}
There is an extent lax indexed functor 
\footnotesize\[
\mathrmbfit{ext} 
: \mathrmbfit{spec} \Rightarrow \mathrmbfit{struc}
: \mathrmbf{Lang}^{\mathrm{op}} \rightarrow \mathrmbf{Cat}.
\]\normalsize
\end{meta-axiom}
%
For any language $\Sigma$,
there is a fiber functor
$\mathrmbfit{ext}_{\Sigma} : \mathrmbfit{spec}(\Sigma) \rightarrow \mathrmbfit{struc}(\Sigma)$
that maps 
a $\Sigma$-specification $T \in \mathrmbfit{spec}(\Sigma)$ 
to its extent $\Sigma$-structure 
$\mathrmbfit{ext}_{\Sigma}(T) \in \mathrmbfit{struc}(\Sigma)$, and
maps a $\Sigma$-specification ordering
$T \geq_{\Sigma} T'$ in $\mathrmbfit{spec}(\Sigma)$
to its extent generalization-specialization $\Sigma$-structure morphism~\footnote{We regard this morphism, 
as we do all $\Sigma$-structure morphisms,
as an abstract generalization-specialization $\Sigma$-structure ordering 
pointing downward in the concept lattice at $\Sigma$.}
$\mathrmbfit{ext}_{\Sigma}(T \geq_{\Sigma} T') = {\gtrdot}_{\Sigma}(T,T') :
\mathrmbfit{ext}_{\Sigma}(T) \rightarrow \mathrmbfit{ext}_{\Sigma}(T')$,
so that
$\mathrmbfit{int}_{\Sigma}(\mathrmbfit{ext}_{\Sigma}(T)) = T^{\scriptscriptstyle\bullet}
\geq_{\Sigma} 
T'^{\scriptscriptstyle\bullet} = \mathrmbfit{int}_{\Sigma}(\mathrmbfit{ext}_{\Sigma}(T'))$.
For any language morphism $\sigma : \Sigma_{1} \rightarrow \Sigma_{2}$,
extent satisfies the lax naturality commutative diagram
(asserting existence of a natural transformation $\mathrmbfit{ext}_{\sigma}$)
\begin{center}
\begin{tabular}{@{\hspace{15pt}}c@{\hspace{40pt}}c}
\\
\setlength{\unitlength}{0.65pt}
\begin{picture}(100,0)(10,95)
\put(0,80){\makebox(0,0){\footnotesize{$\mathrmbfit{spec}(\Sigma_{1})$}}}
\put(100,80){\makebox(0,0){\footnotesize{$\mathrmbfit{struc}(\Sigma_{1})$}}}
\put(0,0){\makebox(0,0){\footnotesize{$\mathrmbfit{spec}(\Sigma_{2})$}}}
\put(100,0){\makebox(0,0){\footnotesize{$\mathrmbfit{struc}(\Sigma_{2})$}}}
\put(50,92){\makebox(0,0){\scriptsize{$\mathrmbfit{ext}_{\Sigma_{1}}$}}}
\put(105,40){\makebox(0,0)[l]{\scriptsize{$\mathrmbfit{struc}(\sigma)$}}}
\put(-5,40){\makebox(0,0)[r]{\scriptsize{$\mathrmbfit{spec}(\sigma)$}}}
\put(50,-14){\makebox(0,0){\scriptsize{$\mathrmbfit{ext}_{\Sigma_{2}}$}}}
\put(0,20){\vector(0,1){40}}
\put(100,20){\vector(0,1){40}}
\put(35,80){\vector(1,0){30}}
\put(35,0){\vector(1,0){30}}
\put(35,50){\begin{picture}(0,0)(0,0)
\setlength{\unitlength}{0.5pt}
\put(30,-4){\makebox(0,0){\large{$\scriptstyle\mathrmbfit{ext}_{\sigma}$}}}
\drawline(3,0)(23,-20)
\drawline(0,-3)(20,-23)
\qbezier(23,-16)(23,-20)(25,-25)
\qbezier(16,-23)(20,-23)(25,-25)
\end{picture}}
\end{picture}
&
\begin{tabular}[b]{p{210pt}}
\footnotesize\noindent
The extent operator semantically characterizes specfications, whereas 
the intent operator formally characterizes structures.
The diagram to the left asserts that change of notation laxly preserves semantic characterization, whereas
the comparable (commuting) diagram for intent asserts the invariancy of truth under change of notation:
change of notation exactly preserves formal characterization.
\end{tabular}
\\ \\
\end{tabular}
\end{center}
or componentwise,
for every target specification $T_{2} \in \mathrmbfit{spec}(\Sigma_{2})$,
there is a laxification $\Sigma_{1}$-structure morphism
$\mathrmbfit{ext}_{\Sigma_{1}}(\mathrmbfit{spec}(\sigma)(T_{2}))
\stackrel{\mathrmbfit{ext}_{\sigma,T_{2}}}{\longrightarrow} 
\mathrmbfit{struc}(\sigma)(\mathrmbfit{ext}_{\Sigma_{2}}(T_{2}))$.
These diagrams vertically paste together under composition of language morphisms:
$\mathrmbfit{ext}_{\sigma_{1}\cdot\sigma_{2}}
= \mathrmbfit{spec}(\sigma_{2})\mathrmbfit{ext}_{\sigma_{1}} \bullet \mathrmbfit{ext}_{\sigma_{2}}\mathrmbfit{struc}(\sigma_{1})$.

The homogenization (Grothendieck construction) is the extent lax fibered functor 
$\mathrmbfit{ext} : \mathrmbf{Spec} \rightarrow \mathrmbf{Struc}$
that commutes with projections
$\mathrmbfit{ext} \circ \mathrmbfit{pr} = \mathrmbfit{pr}$.
Extent
maps a specification ${\langle{\Sigma,T}\rangle}$
to the structure
$\mathrmbfit{ext}(\Sigma,T) = {\langle{\Sigma,\mathrmbfit{ext}_{\Sigma}(T)}\rangle}$, and 
maps a specification morphism~\footnote{Recall that a specification morphism
$\sigma : (\Sigma_{1},T_{1}) \rightarrow (\Sigma_{2},T_{2})$
is a language morphism $\sigma : \Sigma_{1} \rightarrow \Sigma_{2}$,
where $T_{1} \geq_{\Sigma_{1}} \mathrmbfit{inv}(\sigma)(T_{2})=\mathrmbfit{spec}(\sigma)(T_{2})$,
or $\mathrmbfit{dir}(\sigma)(T_{1}) \geq_{\Sigma_{2}} T_{2}$.}
$\sigma : {\langle{\Sigma_{1},T_{1}}\rangle} \rightarrow {\langle{\Sigma_{2},T_{2}}\rangle}$,
to the structure morphism~\footnote{Recall that a structure morphism
$(\sigma,h) : (\Sigma_{1},M_{1}) \rightarrow (\Sigma_{2},M_{2})$
is a language morphism $\sigma : \Sigma_{1} \rightarrow \Sigma_{2}$  
and a $\Sigma_{1}$-structure morphism $h : M_{1} \rightarrow \mathrmbfit{struc}(\sigma)(M_{2})$.}
\footnotesize\[
\mathrmbfit{ext}(\sigma) = {\langle{\sigma,{\mathrmbfit{ext}_{\Sigma_{1}}(T_{1} \geq_{\Sigma_{1}} \mathrmbfit{spec}(\sigma)(T_{2}))}{\,\cdot\,}\mathrmbfit{ext}_{\sigma,T_{2}}}\rangle} :
{\langle{\Sigma_{1},\mathrmbfit{ext}_{\Sigma_{1}}(T_{1})}\rangle} \rightarrow 
{\langle{\Sigma_{2},\mathrmbfit{ext}_{\Sigma_{2}}(T_{2})}\rangle} 
\]\normalsize
where
\vspace{-10pt}
\begin{center}
{\footnotesize$\begin{array}{l}
T_{1} \geq_{\Sigma_{1}} \mathrmbfit{inv}(\sigma)(T_{2})=\mathrmbfit{spec}(\sigma)(T_{2}) \\
\mathrmbfit{ext}_{\Sigma_{1}}(T_{1} \geq_{\Sigma_{1}} \mathrmbfit{spec}(\sigma)(T_{2}))
\cdot \mathrmbfit{ext}_{\sigma,T_{2}} :
\mathrmbfit{ext}_{\Sigma_{1}}(T_{1}) \rightarrow
\mathrmbfit{ext}_{\Sigma_{1}}(\mathrmbfit{spec}(\sigma)(T_{2})) \rightarrow
\mathrmbfit{struc}(\sigma)(\mathrmbfit{ext}_{\Sigma_{2}}(T_{2}))
\end{array}$}
\end{center}

\begin{definition}\label{theory:logic:functor}
{\normalsize$\blacksquare$}
There is a theory logic lax indexed functor
\footnotesize\[
\widehat{\mathrmbfit{th}} 
= (\mathrmbfit{ext},\mathrmbfit{1}_{\mathrmbf{spec}}) 
: \mathrmbfit{spec} \Rightarrow \mathrmbfit{log}
: \mathrmbf{Lang}^{\mathrm{op}} \rightarrow \mathrmbf{Cat}.
\]\normalsize
that mediates the extent and identity indexed functors,
so that 
\footnotesize\[
\widehat{\mathrmbfit{th}} \bullet \mathrmbfit{pr}_{0} 
= \mathrmbfit{ext}
\;\;\text{and}\;\; 
\widehat{\mathrmbfit{th}} \bullet \mathrmbfit{pr}_{1} 
= \mathrmbfit{1}_{\mathrmbf{spec}}.
\]\normalsize
\end{definition}
For any language $\Sigma$,
there is a theory fiber functor
$\widehat{\mathrmbfit{th}}_{\Sigma} 
= (\mathrmbfit{ext}_{\Sigma},\mathrmbfit{1}_{\mathrmbf{spec}(\Sigma)}) 
: \mathrmbfit{spec}(\Sigma) \rightarrow \mathrmbfit{log}(\Sigma)$
that maps a $\Sigma$-specification $T \in \mathrmbfit{spec}(\Sigma)$ 
to the (sound) $\Sigma$-logic $(\mathrmbfit{ext}_{\Sigma}(T),T) \in \mathrmbfit{log}(\Sigma)$, and
maps a $\Sigma$-specification ordering
$T \geq_{\Sigma} T'$ in $\mathrmbfit{spec}(\Sigma)$
to the $\Sigma$-logic morphism
$\widehat{\mathrmbfit{th}}_{\Sigma}(T \geq_{\Sigma} T') = \mathrmbfit{ext}_{\Sigma}(T \geq_{\Sigma} T') :
(\mathrmbfit{ext}_{\Sigma}(T),T) \rightarrow (\mathrmbfit{ext}_{\Sigma}(T'),T')$
consisting of a $\Sigma$-structure morphism
$\mathrmbfit{ext}(T \geq_{\Sigma} T') : \mathrmbfit{ext}_{\Sigma}(T) \rightarrow \mathrmbfit{ext}_{\Sigma}(T')$~\footnote{Recall that a $\Sigma$-logic morphism
$f : L = {\langle{M,T}\rangle} \rightarrow {\langle{M',T'}\rangle} = L'$,
consists of a $\Sigma$-structure morphism $f : M \rightarrow M'$
and a possibly independent $\Sigma$-specification ordering $T \geq_{\Sigma} T'$.
Here of course,
the $\Sigma$-specification ordering $T \geq_{\Sigma} T'$ and
the $\Sigma$-structure morphism 
$\mathrmbfit{ext}(T \geq_{\Sigma} T') : \mathrmbfit{ext}_{\Sigma}(T) \rightarrow \mathrmbfit{ext}_{\Sigma}(T')$
with its associated $\Sigma$-specification ordering
$\mathrmbfit{int}_{\Sigma}(\mathrmbfit{ext}_{\Sigma}(T)) \geq_{\Sigma} \mathrmbfit{int}_{\Sigma}(\mathrmbfit{ext}_{\Sigma}(T'))$
are highly dependent.}.
For any language morphism $\sigma : \Sigma_{1} \rightarrow \Sigma_{2}$,
theory logic satisfies the lax naturality commutative diagram
(asserting existence of a natural transformation $\widehat{\mathrmbfit{th}}_{\sigma}$)
\begin{center}
\begin{tabular}{c@{\hspace{60pt}}c}
\\
\begin{tabular}{c}
\setlength{\unitlength}{0.65pt}
\begin{picture}(100,80)(0,0)
\put(0,80){\makebox(0,0){\footnotesize{$\mathrmbfit{spec}(\Sigma_{1})$}}}
\put(100,80){\makebox(0,0){\footnotesize{$\mathrmbfit{log}(\Sigma_{1})$}}}
\put(0,0){\makebox(0,0){\footnotesize{$\mathrmbfit{spec}(\Sigma_{2})$}}}
\put(100,0){\makebox(0,0){\footnotesize{$\mathrmbfit{log}(\Sigma_{2})$}}}
\put(50,92){\makebox(0,0){\scriptsize{$\widehat{\mathrmbfit{th}}_{\Sigma_{1}}$}}}
\put(105,40){\makebox(0,0)[l]{\scriptsize{$\mathrmbfit{log}(\sigma)$}}}
\put(-5,40){\makebox(0,0)[r]{\scriptsize{$\mathrmbfit{spec}(\sigma)$}}}
\put(50,-14){\makebox(0,0){\scriptsize{$\widehat{\mathrmbfit{th}}_{\Sigma_{2}}$}}}
\put(0,20){\vector(0,1){40}}
\put(100,20){\vector(0,1){40}}
\put(35,80){\vector(1,0){30}}
\put(35,0){\vector(1,0){30}}
\put(35,50){\begin{picture}(0,0)(0,0)
\setlength{\unitlength}{0.5pt}
\put(30,-4){\makebox(0,0){\large{$\scriptstyle\widehat{\mathrmbfit{th}}_{\sigma}$}}}
\drawline(3,0)(23,-20)
\drawline(0,-3)(20,-23)
\qbezier(23,-16)(23,-20)(25,-25)
\qbezier(16,-23)(20,-23)(25,-25)
\end{picture}}
\end{picture}
\end{tabular}
&
{\footnotesize$
\begin{array}{l}
\widehat{\mathrmbfit{th}}_{\sigma} \circ \mathrmbfit{pr}_{0}(\Sigma_{1}) = \mathrmbfit{ext}_{\sigma} \\
\widehat{\mathrmbfit{th}}_{\sigma} \circ \mathrmbfit{pr}_{1}(\Sigma_{1}) = \mathrmit{1}_{\mathrmbfit{spec}(\sigma)} 
\end{array}$}
\\ & \\
\end{tabular}
\end{center}
or componentwise,
for every target specification $T_{2} \in \mathrmbfit{spec}(\Sigma_{2})$,
there is a laxification $\Sigma_{1}$-logic morphism
\footnotesize\[
\overset{\textstyle\underbrace{\widehat{\mathrmbfit{th}}_{\Sigma_{1}}(\mathrmbfit{spec}(\sigma)(T_{2}))}}
{\mathrmbfit{ext}_{\Sigma_{1}}(\mathrmbfit{spec}(\sigma)(T_{2})),\mathrmbfit{spec}(\sigma)(T_{2})}
\xrightarrow{\overset{\scriptscriptstyle\underbrace{\widehat{\mathrmbfit{th}}_{\sigma,T_{2}}}}
{\textstyle\mathrmbfit{ext}_{\sigma,T_{2}}}}
\overset{\textstyle\underbrace{\mathrmbfit{log}(\sigma)(\widehat{\mathrmbfit{th}}_{\Sigma_{2}}(T_{2}))}}
{(\mathrmbfit{struc}(\sigma)(\mathrmbfit{ext}_{\Sigma_{2}}(T_{2})),
\mathrmbfit{spec}(\sigma)(T_{2}))}.
\]\normalsize
These diagrams vertically paste together under composition of language morphisms.

The homogenization (Grothendieck construction) is the theory logic lax fibered functor 
$\widehat{\mathrmbfit{th}} : \mathrmbf{Spec} \rightarrow \mathrmbf{Log}$
that commutes with projections
$\widehat{\mathrmbfit{th}} \circ \mathrmbfit{pr} = \mathrmbfit{pr}$.
Theory logic
maps a specification ${\langle{\Sigma,T}\rangle}$
to the (sound) logic
$\widehat{\mathrmbfit{th}}(\Sigma,T) = {\langle{\Sigma,\mathrmbfit{ext}_{\Sigma}(T),T}\rangle}$, and 
maps a specification morphism
$\sigma : {\langle{\Sigma_{1},T_{1}}\rangle} \rightarrow {\langle{\Sigma_{2},T_{2}}\rangle}$
to the logic morphism
\footnotesize\[
\widehat{\mathrmbfit{th}}(\sigma) = 
{\langle{\sigma,
\widehat{\mathrmbfit{th}}_{\Sigma_{1}}(T_{1} \geq_{\Sigma_{1}} \mathrmbfit{spec}(\sigma)(T_{2})) {\,\cdot}
\widehat{\mathrmbfit{th}}_{\sigma,T_{2}}}\rangle} :
{\langle{\Sigma_{1},\widehat{\mathrmbfit{th}}_{\Sigma_{1}}(T_{1})}\rangle} \rightarrow 
{\langle{\Sigma_{2},\widehat{\mathrmbfit{th}}_{\Sigma_{2}}(T_{2})}\rangle}
\]\normalsize
\vspace{-10pt}
where
\begin{center}
{\footnotesize$\begin{array}{l}
T_{1} \geq_{\Sigma_{1}} \mathrmbfit{inv}(\sigma)(T_{2})=\mathrmbfit{spec}(\sigma)(T_{2}) \\
\widehat{\mathrmbfit{th}}_{\Sigma_{1}}(T_{1} \geq_{\Sigma_{1}} \mathrmbfit{spec}(\sigma)(T_{2}))
\cdot \widehat{\mathrmbfit{th}}_{\sigma,T_{2}} :
\widehat{\mathrmbfit{th}}_{\Sigma_{1}}(T_{1}) \rightarrow
\widehat{\mathrmbfit{th}}_{\Sigma_{1}}(\mathrmbfit{spec}(\sigma)(T_{2})) \rightarrow
\mathrmbfit{struc}(\sigma)(\widehat{\mathrmbfit{th}}_{\Sigma_{2}}(T_{2}))
\end{array}$}
\end{center}

\begin{definition}\label{restricted:theory:logic:functor}
{\normalsize$\blacksquare$}
There is a restricted theory logic lax indexed functor
\footnotesize\[
\mathrmbfit{th}
: \mathrmbfit{spec} \Rightarrow \mathrmbfit{snd}
: \mathrmbf{Lang}^{\mathrm{op}} \rightarrow \mathrmbf{Cat}.
\]\normalsize
that is the restriction of theory logic to sound logics,
so that
\footnotesize\[
\mathrmbfit{th} \bullet \mathrmbfit{inc} = \widehat{\mathrmbfit{th}} 
: \mathrmbfit{struc} \Rightarrow \mathrmbfit{log}.
\]\normalsize
\end{definition}
Clearly,
\footnotesize\[
\mathrmbfit{th} \bullet \mathrmbfit{pr}_{0} = \mathrmbfit{ext}
\;\;\text{and}\;\; 
\mathrmbfit{th} \bullet \mathrmbfit{pr}_{1} = \mathrmbfit{1}_{\mathrmbf{Spec}}.
\]\normalsize

\begin{proposition}
There is an adjunction (coreflection)
\footnotesize\[
\lambda 
= {\langle{\mathrmbfit{th} \dashv \mathrmbfit{pr}_{1},1,\varepsilon}\rangle} 
: \mathrmbf{Spec} \rightarrow \mathrmbf{Snd}.
\]\normalsize
\end{proposition}

\begin{proof}
The unit is the identity
$\mathrmbfit{th} \bullet \mathrmbfit{pr}_{1} = \mathrmbfit{1}_{\mathrmbf{Spec}}$.
The counit
$\varepsilon : \mathrmbfit{pr}_{1} \bullet \mathrmbfit{th} \Rightarrow \mathrmbfit{1}_{\mathrmbf{Snd}}$
has 
the sound logic morphism
$\varepsilon_{\Sigma,M,T} 
= {\langle{1_{\Sigma},\mathrmbfit{ext}_{\Sigma}(T \geq_{\Sigma} \mathrmbfit{int}_{\Sigma}(M)) \cdot \varepsilon^{\mu}_{\Sigma,M}}\rangle} :
{\langle{\Sigma,\mathrmbfit{ext}_{\Sigma}(T),T}\rangle}
\rightarrow {\langle{\Sigma,M,T}\rangle}$
as its ${\langle{\Sigma,M,T}\rangle}^{\mathrm{th}}$ component.




\begin{sloppypar}
Assume that
${\langle{\sigma,f}\rangle} : 
\mathrmbfit{th}(\Sigma_{1},T_{1})
= {\langle{\Sigma_{1},\mathrmbfit{ext}_{\Sigma_{1}}(T_{1}),T_{1}}\rangle} \rightarrow
{\langle{\Sigma_{2},M_{2},T_{2}}\rangle}$
is a morphism between sound logics~\footnote{Recall that
a logic morphism is a pair 
${\langle{\sigma,f}\rangle} : {\langle{\Sigma_{1},M_{1},T_{1}}\rangle} \rightarrow {\langle{\Sigma_{2},M_{2},T_{2}}\rangle}$, 
where $\sigma : \Sigma_{1} \rightarrow \Sigma_{2}$ is a language morphism and
$f : (M_{1},T_{1}) \rightarrow 
\mathrmbfit{log}(\sigma)(M_{2},T_{2}) = (\mathrmbfit{struc}(\sigma)(M_{2}),\mathrmbfit{inv}(\sigma)(T_{2}))$
is a $\Sigma_{1}$-logic morphism; 
that is,
$f : M_{1} \rightarrow \mathrmbfit{struc}(\sigma)(M_{2})$ is a $\Sigma_{1}$-structure morphism, and 
$T_{1} \geq_{\Sigma_{1}} \mathrmbfit{inv}(\sigma)(T_{2})$ is a $\Sigma_{1}$-specification ordering.
Hence,
a logic morphism 
${\langle{\sigma,f}\rangle} : {\langle{\Sigma_{1},M_{1},T_{1}}\rangle} \rightarrow {\langle{\Sigma_{2},M_{2},T_{2}}\rangle}$
consists of
a structure morphism 
${\langle{\sigma,f}\rangle} : {\langle{\Sigma_{1},M_{1}}\rangle} \rightarrow {\langle{\Sigma_{2},M_{2}}\rangle}$ and
a specification morphism 
$\sigma : {\langle{\Sigma_{1},T_{1}}\rangle} \rightarrow {\langle{\Sigma_{2},T_{2}}\rangle}$.}.
Then
$\sigma : 
{\langle{\Sigma_{1},T_{1}}\rangle}  \rightarrow {\langle{\Sigma_{2},T_{2}}\rangle}$
is a specification morphism, and 
$f : \mathrmbfit{ext}_{\Sigma_{1}}(T_{1}) \rightarrow \mathrmbfit{struc}(\sigma)(M_{2})$ 
is the $\Sigma_{1}$-structure morphism
\footnotesize\[
\mathrmbfit{ext}_{\Sigma_{1}}(T_{1} \geq_{\Sigma_{1}} \mathrmbfit{int}_{\Sigma_{1}}\mathrmbfit{struc}(\sigma)((M_{2})))
\cdot \mathrmbfit{ext}_{\sigma,\mathrmbfit{int}_{\Sigma_{2}}(M_{2})}
\cdot \mathrmbfit{struc}(\sigma)(\varepsilon^{\mu}_{\Sigma_{2},M_{2}}), 
\]\normalsize
since 
$\sigma : {\langle{\Sigma_{1},T_{1}}\rangle} \rightarrow {\langle{\Sigma_{2},T_{2}}\rangle}$ is a specification morphism
($T_{1} \geq_{\Sigma_{1}} \mathrmbfit{spec}(\sigma)(T_{2})$),
${\langle{\Sigma_{2},M_{2},T_{2}}\rangle}$ is a sound logic
($T_{2} \geq_{\Sigma_{2}} \mathrmbfit{int}_{\Sigma_{2}}(M_{2})$), and
$\mathrmbfit{ext}_{\Sigma_{1}}(T_{1})$
is unified (Axiom~\ref{extent:unified}).
Hence,
by the naturality of the 
$\mathrmbfit{ext}_{\sigma} : 
\mathrmbfit{spec}(\sigma) \circ \mathrmbfit{int}_{\Sigma_{1}} \Rightarrow
\mathrmbfit{int}_{\Sigma_{2}} \circ \mathrmbfit{struc}(\sigma)$,
$\sigma : {\langle{\Sigma_{1},T_{1}}\rangle}  \rightarrow {\langle{\Sigma_{2},T_{2}}\rangle}$
is the unique specification morphism
such that
$\mathrmbfit{th}(\sigma) \cdot  \varepsilon_{\Sigma_{2},M_{2},T_{2}} 
= 
{\langle{\sigma,f}\rangle}$.
\rule{5pt}{5pt}
\end{sloppypar}
\end{proof}

%% file: fiber-sums.tex
\subsection{Fiber Sums}

For any language $\Sigma$,
the fiber categories $\mathrmbfit{struc}(\Sigma)$ have some finite sums.
Existence of these particular kinds of sums is specified in the following meta-axioms.

\begin{meta-axiom}
{\normalsize$\blacksquare$}
There is 
an initial lax indexed functor
\footnotesize\[
\mathrmbfit{0} 
: \mathrmbfit{lang} \Rightarrow \mathrmbfit{struc} 
: \mathrmbf{Lang}^{\mathrm{op}} \rightarrow \mathrmbf{Cat}.
\]\normalsize
\end{meta-axiom}
For any language $\Sigma$,
there is a fiber functor
$\mathrmbfit{0}_{\Sigma} : \mathrmbfit{lang}(\Sigma) = \mathrmbf{1} \rightarrow \mathrmbfit{struc}(\Sigma)$
that maps 
the single object $\bullet_{\Sigma} \in \mathrmbfit{lang}(\Sigma) = \mathrmbf{1}$ 
to the initial $\Sigma$-structure 
$\mathrmbfit{0}_{\Sigma}(\bullet_{\Sigma}) = 0_{\Sigma} \in \mathrmbfit{struc}(\Sigma)$, and
 maps the single (identity) morphism
$1_{\bullet_{\Sigma}} : {\bullet_{\Sigma}} \rightarrow {\bullet_{\Sigma}}$ in $\mathrmbfit{lang}(\Sigma) = \mathrmbf{1}$
to the identity $\Sigma$-structure morphism
$0_{\Sigma} \rightarrow 0_{\Sigma}$.
For any language morphism $\sigma : \Sigma_{1} \rightarrow \Sigma_{2}$,
inital structure satisfies the lax naturality commutative diagram
\begin{center}
\begin{tabular}{c}
\\
\setlength{\unitlength}{0.65pt}
\begin{picture}(100,80)(0,0)
\put(-10,80){\makebox(0,0){\footnotesize{$\mathrmbf{1}=\mathrmbfit{lang}(\Sigma_{1})$}}}
\put(100,80){\makebox(0,0){\footnotesize{$\mathrmbfit{struc}(\Sigma_{1})$}}}
\put(-10,0){\makebox(0,0){\footnotesize{$\mathrmbf{1}=\mathrmbfit{lang}(\Sigma_{2})$}}}
\put(100,0){\makebox(0,0){\footnotesize{$\mathrmbfit{struc}(\Sigma_{2})$}}}
\put(50,92){\makebox(0,0){\scriptsize{$\mathrmbfit{0}_{\Sigma_{1}}$}}}
\put(106,40){\makebox(0,0)[l]{\scriptsize{$\mathrmbfit{struc}(\sigma)$}}}
\put(-4,40){\makebox(0,0)[r]{\scriptsize{$\mathrmit{1}_{\mathrmbf{1}}=\mathrmbfit{lang}(\sigma)$}}}
\put(50,-14){\makebox(0,0){\scriptsize{$\mathrmbfit{0}_{\Sigma_{2}}$}}}
\put(0,20){\vector(0,1){40}}
\put(100,20){\vector(0,1){40}}
\put(35,80){\vector(1,0){30}}
\put(35,0){\vector(1,0){30}}
\put(35,50){\begin{picture}(0,0)(0,0)
\setlength{\unitlength}{0.5pt}
\put(28,-7){\makebox(0,0){{$\scriptstyle\mathrmbfit{0}_{\sigma}$}}}
\drawline(3,0)(23,-20)
\drawline(0,-3)(20,-23)
\qbezier(23,-16)(23,-20)(25,-25)
\qbezier(16,-23)(20,-23)(25,-25)
\end{picture}}
\end{picture}
\\ \\
\end{tabular}
\end{center}
Pointwise,
this asserts that the structure operator laxly preserves initial structure,
$0_{\Sigma_{1}} \xrightarrow{\scriptstyle\mathrmbfit{0}_{\sigma,\bullet_{\Sigma_{2}}}} \mathrmbfit{struc}(\sigma)(0_{\Sigma_{2}})$.
These diagrams vertically paste together under composition of language morphisms:
$\mathrmbfit{0}_{\sigma_{1}\cdot\sigma_{2}}
= \mathrmbfit{lang}(\sigma_{2})\mathrmbfit{0}_{\sigma_{1}} \bullet \mathrmbfit{0}_{\sigma_{2}}\mathrmbfit{struc}(\sigma_{1})
= \mathrmbfit{0}_{\sigma_{1}} \bullet \mathrmbfit{0}_{\sigma_{2}}\mathrmbfit{struc}(\sigma_{1})$
or
$\mathrmbfit{0}_{\sigma_{1}\cdot\sigma_{2},\bullet_{\Sigma_{3}}}
= \mathrmbfit{0}_{\sigma_{1},\bullet_{\Sigma_{2}}} \cdot \mathrmbfit{struc}(\sigma_{1})(\mathrmbfit{0}_{\sigma_{2},\bullet_{\Sigma_{3}}})$.

The homogenization (Grothendieck construction) is the initial lax fibered functor 
$\mathrmbfit{0} : \mathrmbf{Lang} \rightarrow \mathrmbf{Spec}$ 
that commutes with index projections
$\mathrmbfit{0} \circ \mathrmbfit{pr} = \mathrmit{1}_{\mathrmbf{Lang}}$.
Initial structure
maps a language $\Sigma$ 
(more precisely ${\langle{\Sigma,\bullet_{\Sigma}}\rangle}$)
to the initial structure $\mathrmbfit{0}(\Sigma) = {\langle{\Sigma,0_{\Sigma}}\rangle}$, and
maps a language morphism $\sigma : \Sigma_{1} \rightarrow \Sigma_{2}$
(more precisely $\sigma : {\langle{\Sigma_{1},\bullet_{\Sigma_{1}}}\rangle} \rightarrow {\langle{\Sigma_{2},\bullet_{\Sigma_{2}}}\rangle}$)
to the initial structure morphism
$\mathrmbfit{0}(\sigma) 
={\langle{\sigma,\mathrmbfit{0}_{\Sigma_{1}}(1_{\bullet_{\Sigma_{1}}}){\,\cdot\,}\mathrmbfit{0}_{\sigma,\bullet_{\Sigma_{2}}}}\rangle}
= {\langle{\sigma,\mathrmbfit{0}_{\sigma,\bullet_{\Sigma_{2}}}}\rangle}
: {\langle{\Sigma_{1},0_{\Sigma_{1}}}\rangle} \rightarrow {\langle{\Sigma_{2},0_{\Sigma_{2}}}\rangle}$.
Clearly,
the initial functor commutes with projections
$\mathrmbfit{0} \circ \mathrmbfit{pr} = \mathrmbfit{1}_{\mathrmbf{Lang}}$.
Initiality, 
the fact that
$\mathrmbfit{0}_{\Sigma}(\bullet_{\Sigma})$ is the initial object in the structure fiber category $\mathrmbf{Struc}(\Sigma)$,
is axiomatized by the following.
\begin{meta-axiom}
{\normalsize$\blacksquare$}
There is an adjunction (coreflection)
\footnotesize\[
\omega = 
{\langle{\mathrmbfit{0},\mathrmbfit{pr},1,\varepsilon}\rangle} 
: \mathrmbf{Lang} \rightarrow \mathrmbf{Struc}.
\]\normalsize
\end{meta-axiom}
The unit is the identity
$\mathrmbfit{1}_{\mathrmbf{Lang}} = \mathrmbfit{0} \bullet \mathrmbfit{pr}$.
The counit is a natural transformation
$\varepsilon : \mathrmbfit{pr} \bullet \mathrmbfit{0} \Rightarrow \mathrmbfit{1}_{\mathrmbf{Struc}}$.
For any structure ${\langle{\Sigma,M}\rangle}$,
the ${\langle{\Sigma,M}\rangle}^{\mathrm{th}}$ component of $\varepsilon$ is
the structure morphism
$\varepsilon_{\Sigma,M} = {\langle{1_{\Sigma},{!}_{M}}\rangle} : 
\mathrmbfit{0}(\mathrmbfit{pr}(\Sigma,M))={\langle{\Sigma,0_{\Sigma}}\rangle} \rightarrow {\langle{\Sigma,M}\rangle}$
with $\Sigma$-structure morphism
${!}_{M} : 0_{\Sigma} \rightarrow M$.
We show that ${!}_{M}$ is the unique $\Sigma$-structure morphism from the $\Sigma$-structure $0_{\Sigma}$ to $M$.

\begin{figure}
\begin{center}
\begin{tabular}{@{\hspace{-50pt}}c@{\hspace{170pt}}c@{\hspace{60pt}}}
\\ 
& \\
\setlength{\unitlength}{0.7pt}
\begin{picture}(120,80)(0,0)
\put(0,80){\makebox(0,0){$\Sigma_{1}$}}
\put(120,90){\makebox(0,0){$
\overset{\textstyle\underbrace{
\mathrmbfit{pr}(\mathrmbfit{0}(\Sigma_{1}))}}
{\Sigma_{1}}
$}}
\put(120,-12){\makebox(0,0){$
\overset{\textstyle\underbrace{\Sigma_{2}}}
{\mathrmbfit{pr}(\Sigma_{2},M_{2})}
$}}
\put(50,30){\makebox(0,0)[r]{\footnotesize{$\sigma$}}}
\put(50,85){\makebox(0,0){\footnotesize{$
\overset{\underbrace{\eta_{\Sigma_{1}}}}{=}
$}}}
\put(128,45){\makebox(0,0)[l]{\footnotesize{$
\overset{\underbrace{\mathrmbfit{pr}(\sigma,g)}}{\sigma}
$}}}
\put(20,80){\vector(1,0){75}}
\put(18,68){\vector(3,-2){75}}
\put(120,65){\vector(0,-1){45}}
\put(85,0){\begin{picture}(0,80)(0,0)
\put(120,92){\makebox(0,0){$
\overset{\underbrace{\mathrmbfit{0}(\Sigma_{1})}}
{{\langle{\Sigma_{1},0_{\Sigma_{1}}}\rangle}}
$}}
\put(120,0){\makebox(0,0){${\langle{\Sigma_{2},M_{2}}\rangle}$}}
\put(120,65){\vector(0,-1){45}}
\put(130,40){\makebox(0,0)[l]{\footnotesize{$
{\langle{\sigma,g}\rangle}
$}}}
\end{picture}}
\end{picture}
&
\setlength{\unitlength}{0.7pt}
\begin{picture}(120,80)(0,0)
\put(0,80){\makebox(0,0){${\langle{\Sigma_{2},M_{2}}\rangle}$}}
\put(0,0){\makebox(0,0){${\langle{\Sigma_{1},\mathrmbfit{struc}(\sigma)(M_{2})}\rangle}$}}
\put(120,92){\makebox(0,0){$
\overset{\underbrace{\mathrmbfit{0}(\mathrmbfit{pr}(\Sigma_{2},M_{2}))}}
{{\langle{\Sigma_{2},0_{\Sigma_{2}}}\rangle}}
$}}
\put(120,-12){\makebox(0,0){$
\overset{\textstyle{\langle{\Sigma_{1},0_{\Sigma_{1}}}\rangle}}
{\overbrace{\mathrmbfit{0}(\Sigma_{1})}}
$}}
\put(-5,35){\makebox(0,0)[r]{\footnotesize{${\langle{\sigma,1}\rangle}$}}}
\put(55,35){\makebox(0,0)[r]{\footnotesize{${\langle{\sigma,g}\rangle}$}}}
\put(55,105){\makebox(0,0){\footnotesize{$
\overset{\underbrace{\varepsilon_{\Sigma_{2},M_{2}}}}
{{\langle{1_{\Sigma_{2}},{!}_{M_{2}}}\rangle}}
$}}}
\put(50,-25){\makebox(0,0){\footnotesize{$
\underset{\overbrace{\varepsilon_{\Sigma_{1},\mathrmbfit{struc}(\sigma)(M_{2})}}}
{{\langle{1_{\Sigma_{1}},{!}_{\mathrmbfit{struc}(\sigma)(M_{2})}}\rangle}}
$}}}
\put(115,42){\makebox(0,0)[l]{\footnotesize{$
\overset{\underbrace{\mathrmbfit{0}(\sigma)}}
{{\langle{\sigma,\mathrmbfit{0}_{\sigma,\bullet_{\Sigma_{2}}}}\rangle}}
$}}}
\put(80,80){\vector(-1,0){45}}
\qbezier[10](80,0)(70,0)(60,0)\put(60,0){\vector(-1,0){0}}
\put(93,18){\vector(-3,2){75}}
\put(120,20){\vector(0,1){45}}
\qbezier[15](0,20)(0,32.5)(0,65)\put(0,65){\vector(0,1){0}}
\put(85,0){\begin{picture}(0,80)(0,0)
\put(120,92){\makebox(0,0){$
\overset{\underbrace{\mathrmbfit{pr}(\Sigma_{2},M_{2})}}
{\Sigma_{2}}
$}}
\put(120,0){\makebox(0,0){$\Sigma_{1}$}}
\put(120,20){\vector(0,1){45}}
\put(130,40){\makebox(0,0)[l]{\footnotesize{$\sigma$}}}
\end{picture}}
\end{picture}
\\ & \\ & \\ & \\ 
\multicolumn{2}{c}{
{\footnotesize$\begin{array}{l}
{\langle{\sigma,g}\rangle} =
{\langle{\sigma,\mathrmbfit{0}_{\sigma,\bullet_{\Sigma_{2}}}}\rangle} \circ 
{{\langle{1_{\Sigma_{2}},{!}_{M_{2}}}\rangle}} =
{{\langle{\sigma,\mathrmbfit{0}_{\sigma,\bullet_{\Sigma_{2}}} \cdot \mathrmbfit{struc}(\sigma)({!}_{M_{2}})}\rangle}} 
\\ \\
g = \mathrmbfit{0}_{\sigma,\bullet_{\Sigma_{2}}} \cdot \mathrmbfit{struc}(\sigma)({!}_{M_{2}}) :
0_{\Sigma_{1}} \rightarrow \mathrmbfit{struc}(\sigma)(0_{\Sigma_{2}})
\rightarrow
\mathrmbfit{struc}(\sigma)(M_{2})
\end{array}$}}
\end{tabular}
\end{center}
\caption{Initial-Projection Adjunction}
\label{initial-projection:adjunction}
\end{figure}

The adjunction is expressed either as universally or couniversally.
{\bfseries Universally:} 
for any structure ${\langle{\Sigma_{2},M_{2}}\rangle}$
and any language morphism $\sigma : \Sigma_{1} \rightarrow \Sigma_{2}=\mathrmbfit{pr}(\Sigma_{2},M_{2})$,
there is a unique $\Sigma_{1}$-morphism 
$g : 0_{\Sigma_{1}} \rightarrow \mathrmbfit{struc}(\sigma)(M_{2})$.
Hence,
$g = {!}_{\mathrmbfit{struc}(\sigma)(M_{2})}$.
{\bfseries Couniversally:} 
for any structure ${\langle{\Sigma_{2},M_{2}}\rangle}$,
    any language morphism $\sigma : \Sigma_{1} \rightarrow \Sigma_{2}$
and any $\Sigma_{1}$-morphism $g : 0_{\Sigma_{1}} \rightarrow \mathrmbfit{struc}(\sigma)(M_{2})$~\footnote{That is,
for any language $\Sigma_{1}$
and any structure morphism
${\langle{\sigma,g}\rangle} : \mathrmbfit{0}(\Sigma_{1})={\langle{\Sigma_{1},0_{\Sigma_{1}}}\rangle} \rightarrow {\langle{\Sigma_{2},M_{2}}\rangle}$.},
we must have
$g = \mathrmbfit{0}_{\sigma,\bullet_{\Sigma_{2}}} \cdot \mathrmbfit{struc}(\sigma)({!}_{M_{2}}) 
: 0_{\Sigma_{1}} \rightarrow \mathrmbfit{struc}(\sigma)(0_{\Sigma_{2}}) \rightarrow \mathrmbfit{struc}(\sigma)(M_{2})$.

In summary,
for any $\Sigma_{2}$-structure $M_{2}$
and any language morphism $\sigma : \Sigma_{1} \rightarrow \Sigma_{2}$, 
${!}_{\mathrmbfit{struc}(\sigma)(M_{2})}  = \mathrmbfit{0}_{\sigma,\bullet_{\Sigma_{2}}} \cdot \mathrmbfit{struc}(\sigma)({!}_{M_{2}}) 
: 0_{\Sigma_{1}} \rightarrow \mathrmbfit{struc}(\sigma)(M_{2})$
is the unique $\Sigma_{1}$-morphism
from $0_{\Sigma_{1}}$ to $\mathrmbfit{struc}(\sigma)(M_{2})$.
In particular,
(when $\sigma = 1_{\Sigma} : \Sigma \rightarrow \Sigma$)
for any $\Sigma$-structure $M$,
${!}_{M} : 0_{\Sigma} \rightarrow M$ is the unique $\Sigma$-morphism from $0_{\Sigma}$ to $M$.

\begin{definition}
{\normalsize$\blacksquare$}
There is an initial indexed functor
\footnotesize\[
\mathrmbfit{0} = \mathrmbfit{pr} \circ \mathrmbfit{0}
: \mathrmbfit{spec} \Rightarrow \mathrmbfit{lang} \Rightarrow \mathrmbfit{struc} 
: \mathrmbf{Lang}^{\mathrm{op}} \rightarrow \mathrmbf{Cat}.
\]\normalsize
\end{definition}
For any language $\Sigma$,
there is a (constant) initial fiber functor
$\mathrmbfit{0}_{\Sigma} : \mathrmbfit{struc}(\Sigma) \rightarrow \mathrmbfit{spec}(\Sigma)$
that maps 
a $\Sigma$-specification $T \in \mathrmbfit{spec}(\Sigma)$ 
to the initial $\Sigma$-structure $0_{\Sigma}$, and
maps a $\Sigma$-specification ordering
$T \geq_{\Sigma} T'$ in $\mathrmbfit{spec}(\Sigma)$
to the trivial identity $\Sigma$-structure morphism
$1 : 0_{\Sigma} \rightarrow 0_{\Sigma}$.
For any language morphism $\sigma : \Sigma_{1} \rightarrow \Sigma_{2}$ in $\mathrmbf{Lang}$,
initiality satisfies the lax naturality commutative diagram
$\mathrmbfit{0}_{\sigma} :
\mathrmbfit{spec}(\sigma) \cdot \mathrmbfit{0}_{\Sigma_{1}} \Rightarrow 
\mathrmbfit{0}_{\Sigma_{2}} \cdot \mathrmbfit{struc}(\sigma)$.
Pointwise,
this asserts that the structure operator laxly preserves initial structure,
$0_{\Sigma_{1}} \xrightarrow{\scriptstyle\mathrmbfit{0}_{\sigma,T_{2}}} \mathrmbfit{struc}(\sigma)(0_{\Sigma_{2}})$
for any $\Sigma$-specification $T_{2} \in \mathrmbfit{spec}(\Sigma_{2})$.
These diagrams vertically paste together under composition of language morphisms.

The homogenization (Grothendieck construction) is the initial fibered functor 
$\mathrmbfit{0} = \mathrmbfit{pr} \circ \mathrmbfit{0} :	\mathrmbf{Spec} \rightarrow \mathrmbf{Struc}$
that commutes with index projections
$\mathrmbfit{0} \circ \mathrmbfit{pr} = \mathrmbfit{pr}$.

\begin{definition}
{\normalsize$\blacksquare$}
There is a initial logic indexed functor
\footnotesize\[
\mathrmbfit{0} 
= (\mathrmbfit{0},\mathrmbfit{1}_{\mathrmbf{spec}}) 
: \mathrmbfit{spec} \Rightarrow \mathrmbfit{log}
: \mathrmbf{Lang}^{\mathrm{op}} \rightarrow \mathrmbf{Cat}.
\]\normalsize
that mediates the initial and identity indexed functors,
so that 
\footnotesize\[
\mathrmbfit{0} \bullet \mathrmbfit{pr}_{0} 
= \mathrmbfit{0}
\;\;\text{and}\;\; 
\mathrmbfit{0} \bullet \mathrmbfit{pr}_{1} 
= \mathrmbfit{1}_{\mathrmbf{spec}}.
\]\normalsize
\end{definition}
For any language $\Sigma$,
there is a initial logic fiber functor
$\mathrmbfit{0}_{\Sigma} 
= (\mathrmbfit{1}_{\mathrmbf{spec}(\Sigma)},\mathrmbfit{0}_{\Sigma}) 
: \mathrmbfit{spec}(\Sigma) \rightarrow \mathrmbfit{log}(\Sigma)$
that maps a $\Sigma$-structure $T \in \mathrmbfit{spec}(\Sigma)$ 
to the $\Sigma$-logic $(T,\mathrmbfit{0}_{\Sigma}) \in \mathrmbfit{log}(\Sigma)$, and
maps a $\Sigma$-structure morphism
$f : T \rightarrow T'$ in $\mathrmbfit{spec}(\Sigma)$
to the $\Sigma$-logic morphism
$\mathrmbfit{0}(f) = f : (T,\mathrmbfit{0}_{\Sigma}) \rightarrow (T',\mathrmbfit{0}_{\Sigma})$.

For any language morphism $\sigma : \Sigma_{1} \rightarrow \Sigma_{2}$ in $\mathrmbf{Lang}$,
initial logic satisfies the naturality commutative diagram
\begin{center}
\begin{tabular}{c}
\\
\setlength{\unitlength}{0.6pt}
\begin{picture}(100,80)(0,0)
\put(0,80){\makebox(0,0){\footnotesize{$\mathrmbfit{spec}(\Sigma_{1})$}}}
\put(100,80){\makebox(0,0){\footnotesize{$\mathrmbfit{log}(\Sigma_{1})$}}}
\put(0,0){\makebox(0,0){\footnotesize{$\mathrmbfit{spec}(\Sigma_{2})$}}}
\put(100,0){\makebox(0,0){\footnotesize{$\mathrmbfit{log}(\Sigma_{2})$}}}
\put(50,92){\makebox(0,0){\scriptsize{$\mathrmbfit{0}_{\Sigma_{1}}$}}}
\put(105,40){\makebox(0,0)[l]{\scriptsize{$\mathrmbfit{log}(\sigma)$}}}
\put(-5,40){\makebox(0,0)[r]{\scriptsize{$\mathrmbfit{spec}(\sigma)$}}}
\put(50,-14){\makebox(0,0){\scriptsize{$\mathrmbfit{0}_{\Sigma_{2}}$}}}
\put(0,20){\vector(0,1){40}}
\put(100,20){\vector(0,1){40}}
\put(35,80){\vector(1,0){30}}
\put(35,0){\vector(1,0){30}}
\end{picture}
\\ \\
\end{tabular}
\end{center}
or pointwise,
$
\mathrmbfit{0}_{\Sigma_{1}}(\mathrmbfit{spec}(\sigma)(T_{2}))
= (\mathrmbfit{spec}(\sigma)(T_{2}),\mathrmbfit{0}_{\Sigma_{1}})
= (\mathrmbfit{spec}(\sigma)(T_{2}),\mathrmbfit{inv}(\sigma)(\mathrmbfit{0}_{\Sigma_{2}})))
= \mathrmbfit{log}(\sigma)(T_{2},\mathrmbfit{0}_{\Sigma_{2}})
= \mathrmbfit{log}(\sigma)(\mathrmbfit{0}_{\Sigma_{2}}(T_{2}))
$
for every target structure $T_{2} \in \mathrmbfit{spec}(\Sigma_{2})$.

The homogenization (Grothendieck construction) is the initial logic fibered functor 
$\mathrmbfit{0} : \mathrmbf{Spec} \rightarrow \mathrmbf{Log}$ 
that commutes with index projections
$\mathrmbfit{0} \circ \mathrmbfit{pr} = \mathrmbfit{pr}$.
Initial logic
maps a specification ${\langle{\Sigma,T}\rangle}$
to the logic
$\mathrmbfit{0}(\Sigma,T) = {\langle{\Sigma,T,\mathrmbfit{0}_{\Sigma}}\rangle}$,
and maps a specification morphism
$(\sigma,h) : (\Sigma_{1},T_{1}) \rightarrow (\Sigma_{2},T_{2})$
to the logic morphism
$\mathrmbfit{0}(\sigma,h) = {\langle{\sigma,h}\rangle} :
\mathrmbfit{0}(\Sigma_{1},T_{1}) = {\langle{\Sigma_{1},T_{1},\mathrmbfit{0}_{\Sigma_{1}}}\rangle} \rightarrow 
{\langle{\Sigma_{2},T_{2},\mathrmbfit{0}_{\Sigma_{2}}}\rangle} = \mathrmbfit{0}(\Sigma_{2},T_{2})$,
where
$\mathrmbfit{0}_{\Sigma_{1}} \geq_{\Sigma_{1}} 
\mathrmbfit{inv}(\sigma)(\mathrmbfit{0}_{\Sigma_{2}}) =
\mathrmbfit{0}_{\Sigma_{1}}$.

\newpage

\begin{eg}
In the classification logical environment $\mathtt{Cls}$ 
(used in the metatheory of information flow),
the type power functor
${\wp} : \mathrmbf{Set} \rightarrow \mathrmbf{Cls}$
serves as the initial structure lax fibered functor.
The functor ${\wp}$ 
maps a (type) set $X$ to the type power classification ${\wp}(X) = {\langle{X,{\wp}(X),\ni_{X}}\rangle}$,
where $0_{X}= {\langle{{\wp}(X),\ni_{X}}\rangle}$ is the initial $X$-classification:
for any $X$-classification $\mathcal{A} = {\langle{Y,\models}\rangle}$,
the type set function for $\mathcal{A}$ is the unique $X$-infomorphism
${!}_{\mathcal{A}} = \tau_{\mathcal{A}} : 0_{X} \rightleftarrows \mathcal{A}$.
The functor ${\wp}$
maps a (type) function $f : X_{1} \rightarrow X_{2}$ to the type power infomorphism
${\wp}(f) = {\langle{f,f^{-1}}\rangle}
: {\langle{X_{1},{\wp}(X_{1}),\ni_{X_{1}}}\rangle} \rightleftarrows {\langle{X_{2},{\wp}(X_{2}),\ni_{X_{2}}}\rangle}$,
where inverse image
${\wp}_{f} = f^{-1} = {!}_{\mathrmbfit{cls}(f)(0_{X_{2}})} : 0_{X_{1}} \rightarrow \mathrmbfit{cls}(f)(0_{X_{2}})$
is the unique $X_{1}$-infomorphism 
from $0_{X_{1}}={\langle{{\wp}(X_{1}),\ni_{X_{1}}}\rangle}$
to $\mathrmbfit{cls}(f)(0_{X_{2}}) = {\langle{{\wp}(X_{2}),\ni_{f}}\rangle}$,
since
$x_{1} \in \tau_{\mathrmbfit{cls}(f)(0_{X_{2}})}(Z_{2})$
iff $f(x_{1}) \in Z_{2}$
iff $x_{1} \in f^{-1}(Z_{2})$
for all $Z_{2} \in {\wp}(X_{2})$.
Clearly,
the type power functor commutes with projections
${\wp} \circ \mathrmbfit{typ} = 1_{\mathrmbf{Set}}$.
There is a natural transformation
$\varepsilon : \mathrmbfit{typ} \circ {\wp} \Rightarrow 1_{\mathrmbf{Cls}}$,
where for any classification $\mathcal{A}={\langle{X,Y,\models}\rangle}$,
the $\mathcal{A}^{\mathrm{th}}$
component of $\varepsilon$ is defined to be the infomorphism 
$\varepsilon_{\mathcal{A}} = {\langle{1_{X},\tau_{\mathcal{A}}}\rangle} : 
{\wp}(\mathrmbfit{typ}(\mathcal{A})) 
= {\langle{X,{\wp}(X),\ni_{X}}\rangle} \rightarrow {\langle{X,Y,\models}\rangle} = \mathcal{A}$.
Naturality is expressed by the commuting diagram
${\wp}(f) \cdot \varepsilon_{\mathcal{A}_{2}} 
= {\langle{f,f^{-1}}\rangle} \cdot {\langle{1_{X_{2}},\tau_{\mathcal{A}_{2}}}\rangle} 
= {\langle{f,\tau_{\mathcal{A}_{2}}{\cdot}f^{-1}}\rangle} 
= {\langle{f,g{\cdot}\tau_{\mathcal{A}_{1}}}\rangle}
= {\langle{1_{X_{1}},\tau_{\mathcal{A}_{1}}}\rangle}{\cdot}{\langle{f,g}\rangle}
= \varepsilon_{\mathcal{A}_{1}} \cdot {\langle{f,g}\rangle}$
for any infomorphism
${\langle{f,g}\rangle} : \mathcal{A}_{1}={\langle{X_{1},Y_{1},\models_{1}}\rangle} \rightleftarrows {\langle{X_{2},Y_{2},\models_{2}}\rangle}=\mathcal{A}_{2}$,
since
$x_{1} \in \tau_{\mathcal{A}_{1}}(g(y_{2}))$
iff $g(y_{2}) \models_{\mathcal{A}_{1}} x_{1}$
iff $y_{2} \models_{\mathcal{A}_{2}} f(x_{1})$
iff $x_{1} \in f^{-1}(\tau_{\mathcal{A}_{2}}(y_{2}))$.
The lax fibered functor ${\wp} : \mathrmbf{Set} \rightarrow \mathrmbf{Cls}$ 
is the homogenization of a type power lax indexed functor
${\wp} : \mathrmbfit{set} \Rightarrow \mathrmbfit{cls} : \mathrmbf{Set}^{\mathrm{op}} \rightarrow \mathrmbf{Cat}$.
For any set $X$
there is a fiber functor ${\wp}_{X} : \mathrmbfit{set}(X) \rightarrow \mathrmbfit{cls}(X)$
that maps 
the $X$-set (a single object) ${\bullet} \in \mathrmbfit{set}(X)$ 
to the type power $X$-classification 
${\wp}_{X}(\bullet) = {\langle{{\wp}X,\ni_{X}}\rangle} \in \mathrmbfit{cls}(X)$, 
and maps the (identity) $X$-set morphism
$1_{\bullet} : {\bullet} \rightarrow {\bullet}$ in $\mathrmbfit{set}(X)$
to the identity $X$-infomorphism
${\wp}_{X}(\bullet) \rightleftarrows {\wp}_{X}(\bullet)$.
For any function $f : X_{1} \rightarrow X_{2}$,
type power satisfies the lax naturality commutative diagram
\begin{center}
\begin{tabular}{c}
\\
\setlength{\unitlength}{0.6pt}
\begin{picture}(100,80)(0,0)
\put(0,80){\makebox(0,0){\footnotesize{$\mathrmbfit{set}(X_{1})$}}}
\put(100,80){\makebox(0,0){\footnotesize{$\mathrmbfit{cls}(X_{1})$}}}
\put(0,0){\makebox(0,0){\footnotesize{$\mathrmbfit{set}(X_{2})$}}}
\put(100,0){\makebox(0,0){\footnotesize{$\mathrmbfit{cls}(X_{2})$}}}
\put(50,92){\makebox(0,0){\scriptsize{$\mathrmbfit{0}_{X_{1}}={\wp}_{X_{1}}$}}}
\put(105,40){\makebox(0,0)[l]{\scriptsize{$\mathrmbfit{cls}(f)$}}}
\put(-5,40){\makebox(0,0)[r]{\scriptsize{$\mathrmbfit{set}(f)$}}}
\put(50,-14){\makebox(0,0){\scriptsize{$\mathrmbfit{0}_{X_{2}}={\wp}_{X_{2}}$}}}
\put(0,20){\vector(0,1){40}}
\put(100,20){\vector(0,1){40}}
\put(35,80){\vector(1,0){30}}
\put(35,0){\vector(1,0){30}}
\put(30,50){\begin{picture}(0,0)(0,0)
\setlength{\unitlength}{0.5pt}
\put(23,-5){\makebox(0,0)[l]{\large{$\scriptstyle\mathrmbfit{0}_{f}$}{$\scriptscriptstyle=f^{-1}$}}}
\drawline(3,0)(23,-20)
\drawline(0,-3)(20,-23)
\qbezier(23,-16)(23,-20)(25,-25)
\qbezier(16,-23)(20,-23)(25,-25)
\end{picture}}
\end{picture}
\\
\end{tabular}
\end{center}
or pointwise,
${\wp}_{X_{1}}(\mathrmbfit{set}(f)(\bullet))
= {\wp}_{X_{1}}(\bullet)
= {\langle{{\wp}X_{1},\ni_{X_{1}}}\rangle}
\stackrel{f^{-1}}{\rightleftarrows} 
{\langle{{\wp}X_{2},\ni_{f}}\rangle}
= \mathrmbfit{cls}(f)({\wp}X_{2},\ni_{X_{2}})
= \mathrmbfit{cls}(f)({\wp}_{X_{2}}(\bullet))$
for the single target object $\bullet \in \mathrmbfit{set}(X_{2})$.
These diagrams vertically paste together under composition of functions.
\end{eg}

~\footnote{Type power initially appears to be a fibered functor,
since it commutes with projections
${\wp} \circ \mathrmbfit{typ} = 1_{\mathrmbf{Set}}$.
But it is not,
since it does not preserve cartesian morphisms:
any function $f : X_{1} \rightarrow X_{2}$ is (trivially) cartesian,
but the infomorphism 
${\wp}(f) = {\langle{f,f^{-1}}\rangle}
: {\langle{X_{1},{\wp}(X_{1}),\ni_{X_{1}}}\rangle} \rightleftarrows {\langle{X_{2},{\wp}(X_{2}),\ni_{X_{2}}}\rangle}$
is not cartesian,
since
${\langle{f,1_{{\wp}(X_{2})}}\rangle} : 
\mathrmbfit{cls}(f)(X_{2},{\wp}(X_{2}),\ni_{X_{2}}) \rightleftarrows {\langle{X_{2},{\wp}(X_{2}),\ni_{X_{2}}}\rangle}$
is cartesian,
but
although there is a $X_{1}$-infomorphism
between
${\langle{X_{1},{\wp}(X_{1}),\ni_{X_{1}}}\rangle}$ and $\mathrmbfit{cls}(f)(X_{2},{\wp}(X_{2}),\ni_{X_{2}})$,
namely
${\langle{1_{X_{1}},f^{-1}}\rangle} : 
{\langle{X_{1},{\wp}(X_{1}),\ni_{X_{1}}}\rangle} \rightleftarrows \mathrmbfit{cls}(f)(X_{2},{\wp}(X_{2}),\ni_{X_{2}})$,
by a cardinality argument (${\wp}(X_{1}) \not\cong {\wp}(X_{2})$),
there is not necessarily an isomorphism.}

\begin{eg}
In the diagram logical environment $\mathtt{Dgm}$ (used in the metatheory of sketches),
the path functor
${\wp} : \mathrmbf{Gph} \rightarrow \mathrmbf{Dgm}$
serves as the initial structure lax fibered functor.
The functor ${\wp}$
maps a graph $G$ to the path diagram 
${\wp}(G) = {\langle{G,\eta_{G},G^{\ast}}\rangle}$,
where $G^{\ast}$ is the path category generated by $G$ and 
$\eta_{G} : G \rightarrow |G^{\ast}|$ is the embedding graph morphism.
This is an initial $G$-diagram $0_{G} = {\langle{\eta_{G},G^{\ast}}\rangle}$:
for any $G$-diagram ${\langle{G,D,\mathrmbf{C}}\rangle}$ with $D : G \rightarrow |\mathrmbf{C}|$, 
there is a unique $G$-diagram morphism
${!}_{D,\mathrmbf{C}} = \mathrmbfit{D} : 0_{G}={\langle{\eta_{G},G^{\ast}}\rangle} \rightarrow {\langle{D,\mathrmbf{C}}\rangle}$
consisting of a functor $\mathrmbfit{D} : G^{\ast} \rightarrow \mathrmbf{C}$
satisfying the commutative diagram $\eta_{G} \circ |\mathrmbfit{D}| = D$.
The functor ${\wp}$ maps a graph morphism $H : G_{1} \rightarrow G_{2}$ 
to the path diagram morphism
${\wp}(H) = {\langle{H,{\wp}_{H}}\rangle} = {\langle{H,H^{\ast}}\rangle} : 
{\langle{G_{1},\eta_{G_{1}},G_{1}^{\ast}}\rangle} \rightrightarrows {\langle{G_{2},\eta_{G_{2}},G_{2}^{\ast}}\rangle}$,
where 
${\wp}_{H} = {!}_{\mathrmbfit{dgm}(H)(0_{G_{2}})} = H^{\ast} : 
0_{G_{1}} = {\langle{\eta_{G_{1}},G_{1}^{\ast}}\rangle} \rightarrow 
{\langle{H \circ \eta_{G_{2}},G_{2}^{\ast}}\rangle} = \mathrmbfit{dgm}(H)(0_{G_{2}})$
is the unique $G_{1}$-diagram morphism 
from the initial $G_{1}$-diagram $0_{G_{1}}$
to $\mathrmbfit{dgm}(H)(0_{G_{2}})$,
the unique functor $H^{\ast} : G_{1}^{\ast} \rightarrow G_{2}^{\ast}$
adjoint to the graph morphism 
$H \circ \eta_{G_{2}} : G_{1} \rightarrow |G_{2}^{\ast}|$
and hence satisfying commutative diagram $\eta_{G_{1}} \circ |H^{\ast}| = H \circ \eta_{G_{2}}$.
Clearly,
the path functor commutes with projections
${\wp} \circ \mathrmbfit{gph} = 1_{\mathrmbf{Gph}}$.
There is a natural transformation
$\varepsilon : \mathrmbfit{gph} \circ {\wp} \Rightarrow 1_{\mathrmbf{Dgm}}$,
where for any diagram $\mathcal{D}={\langle{G,D,\mathrmbf{C}}\rangle}$,
the $\mathcal{D}^{\mathrm{th}}$
component of $\varepsilon$ is defined to be the diagram morphism 
$\varepsilon_{\mathcal{D}} = {\langle{1_{G},\mathrmbfit{D}}\rangle} : 
{\wp}(\mathrmbfit{gph}(\mathcal{D})) 
= {\langle{G,\eta_{G},G^{\ast}}\rangle} \rightarrow {\langle{G,D,\mathrmbf{C}}\rangle} = \mathcal{D}$.
Naturality is expressed by the commuting diagram
${\wp}(H) \circ \varepsilon_{\mathcal{D}_{2}} 
= {\langle{H,H^{\ast}}\rangle} \circ {\langle{1_{G_{2}},\mathrmbfit{D}_{2}}\rangle} 
= {\langle{H,H^{\ast}{\circ}\mathrmbfit{D}_{2} }\rangle} 
= {\langle{H,\mathrmbfit{D}_{1}{\circ}\mathrmbfit{H}}\rangle}
= {\langle{1_{G_{1}},\mathrmbfit{D}_{1}}\rangle}{\circ}{\langle{H,\mathrmbfit{H}}\rangle}
= \varepsilon_{\mathcal{D}_{1}} \circ {\langle{H,\mathrmbfit{H}}\rangle}$
for any diagram morphism
${\langle{H,\mathrmbfit{H}}\rangle} : \mathcal{D}_{1}={\langle{G_{1},D_{1},\mathrmbf{C}_{1}}\rangle} \rightrightarrows {\langle{G_{2},D_{2},\mathrmbf{C}_{2}}\rangle}=\mathcal{D}_{2}$,
since
$\eta_{G_{1}} \circ |H^{\ast}| \circ |\mathrmbfit{D}_{2}| 
= H \circ \eta_{G_{2}} \circ |\mathrmbfit{D}_{2}| 
= H \circ D_{2} 
= D_{1} \circ |\mathrmbfit{H}|
= \eta_{G_{1}} \circ |\mathrmbfit{D}_{1}| \circ |\mathrmbfit{H}|$
implies by uniqueness that
$H^{\ast} \circ \mathrmbfit{D}_{2} = \mathrmbfit{D}_{1} \circ \mathrmbfit{H}$.
\begin{center}
\begin{tabular}{c}
\\
\setlength{\unitlength}{0.5pt}
\begin{picture}(200,110)(0,-14)
\put(0,80){\makebox(0,0){\footnotesize{$G_{1}$}}}
\put(100,80){\makebox(0,0){\footnotesize{$|G_{1}^{\ast}|$}}}
\put(200,80){\makebox(0,0){\footnotesize{$|\mathrmbf{C}_{1}|$}}}
\put(0,0){\makebox(0,0){\footnotesize{$G_{2}$}}}
\put(100,0){\makebox(0,0){\footnotesize{$|G_{1}^{\ast}|$}}}
\put(200,0){\makebox(0,0){\footnotesize{$|\mathrmbf{C}_{2}|$}}}
\put(-8,40){\makebox(0,0)[r]{\scriptsize{$H$}}}
\put(108,40){\makebox(0,0)[l]{\scriptsize{$|H^{\ast}|$}}}
\put(208,40){\makebox(0,0)[l]{\scriptsize{$|\mathrmbfit{H}|$}}}
\put(50,92){\makebox(0,0){\scriptsize{$\eta_{G_{1}}$}}}
\put(50,-14){\makebox(0,0){\scriptsize{$\eta_{G_{2}}$}}}
\put(146,92){\makebox(0,0){\scriptsize{$|\mathrmbfit{D}_{1}|$}}}
\put(146,-14){\makebox(0,0){\scriptsize{$|\mathrmbfit{D}_{2}|$}}}
\put(100,120){\makebox(0,0){\scriptsize{$D_{1}$}}}
\put(100,-45){\makebox(0,0){\scriptsize{$D_{2}$}}}
\put(30,80){\vector(1,0){40}}
\put(123,80){\vector(1,0){40}}
\put(30,0){\vector(1,0){40}}
\put(123,0){\vector(1,0){40}}
\put(0,60){\vector(0,-1){40}}
\put(100,60){\vector(0,-1){40}}
\put(200,60){\vector(0,-1){40}}
\put(100,100){\oval(190,20)[t]}\put(195,96){\vector(0,-1){0}}
\put(100,-20){\oval(190,20)[b]}\put(195,-16){\vector(0,1){0}}
\end{picture}
\\ \\
\end{tabular}
\end{center}
The lax fibered functor ${\wp} : \mathrmbf{Gph} \rightarrow \mathrmbf{Dgm}$ 
is the homogenization of a path lax indexed functor
${\wp} : \mathrmbfit{gph} \Rightarrow \mathrmbfit{dgm} : \mathrmbf{Set}^{\mathrm{op}} \rightarrow \mathrmbf{Cat}$.
For any graph $G$
there is a fiber functor ${\wp}_{G} : \mathrmbfit{gph}(G) \rightarrow \mathrmbfit{dgm}(G)$
that maps 
the $G$-graph (a single object) ${\bullet} \in \mathrmbfit{gph}(G)$ 
to the path $G$-diagram 
${\wp}_{G}(\bullet) = 
{\langle{\eta_{G},G^{\ast}}\rangle} \in \mathrmbfit{dgm}(G)$, 
and maps the (identity) $G$-graph morphism
$1_{\bullet} : {\bullet} \rightarrow {\bullet}$ in $\mathrmbfit{gph}(G)$
to the identity $G$-diagram morphism
${\wp}_{G}(\bullet) \rightrightarrows {\wp}_{G}(\bullet)$.
For any graph morphism $H : G_{1} \rightarrow G_{2}$,
path satisfies the lax naturality commutative diagram
\begin{center}
\begin{tabular}{c}
\\
\setlength{\unitlength}{0.6pt}
\begin{picture}(100,80)(0,0)
\put(0,80){\makebox(0,0){\footnotesize{$\mathrmbfit{gph}(G_{1})$}}}
\put(100,80){\makebox(0,0){\footnotesize{$\mathrmbfit{dgm}(G_{1})$}}}
\put(0,0){\makebox(0,0){\footnotesize{$\mathrmbfit{gph}(G_{2})$}}}
\put(100,0){\makebox(0,0){\footnotesize{$\mathrmbfit{dgm}(G_{2})$}}}
\put(50,92){\makebox(0,0){\scriptsize{$\mathrmbfit{0}_{G_{1}}={\wp}_{G_{1}}$}}}
\put(105,40){\makebox(0,0)[l]{\scriptsize{$\mathrmbfit{dgm}(H)$}}}
\put(-5,40){\makebox(0,0)[r]{\scriptsize{$\mathrmbfit{gph}(H)$}}}
\put(50,-14){\makebox(0,0){\scriptsize{$\mathrmbfit{0}_{G_{2}}={\wp}_{G_{2}}$}}}
\put(0,20){\vector(0,1){40}}
\put(100,20){\vector(0,1){40}}
\put(35,80){\vector(1,0){30}}
\put(35,0){\vector(1,0){30}}
\put(30,50){\begin{picture}(0,0)(0,0)
\setlength{\unitlength}{0.5pt}
\put(23,-5){\makebox(0,0)[l]{\large{$\scriptstyle\mathrmbfit{0}_{H}$}{$\scriptscriptstyle=H^{\ast}$}}}
\drawline(3,0)(23,-20)
\drawline(0,-3)(20,-23)
\qbezier(23,-16)(23,-20)(25,-25)
\qbezier(16,-23)(20,-23)(25,-25)
\end{picture}}
\end{picture}
\\
\end{tabular}
\end{center}
or pointwise,
${\wp}_{G_{1}}(\mathrmbfit{gph}(H)(\bullet))
= {\wp}_{G_{1}}(\bullet)
= {\langle{\eta_{G_{1}},G_{1}^{\ast}}\rangle}
\stackrel{H^{\ast}}{\rightrightarrows} 
{\langle{H{\circ}\eta_{G_{2}},G_{2}^{\ast}}\rangle}
= \mathrmbfit{dgm}(H)(\eta_{G_{2}},G_{2}^{\ast})
= \mathrmbfit{dgm}(H)({\wp}_{G_{2}}(\bullet))$
for the single target object $\bullet \in \mathrmbfit{gph}(G_{2})$.
These diagrams vertically paste together under composition of functions.
\end{eg}

~\footnote{Every category $\mathrmbf{C}$ has an underlying graph 
$\mathrmbfit{gph}(\mathrmbf{C}) = |\mathrmbf{C}|$. 
Let
$|\mbox{-}| : \mathrmbf{Cat} \rightarrow \mathrmbf{Gph}$ be the underlying graph functor.
Every graph $G$ has 
an associated path category $\mathrmbfit{cat}(G) = G^{\ast}$ 
and embedding graph morphism $\eta_{G} : G \rightarrow |G^{\ast}|$, 
such that 
there is a natural isomorphism between 
graph morphisms $D : G \rightarrow |\mathrmbf{C}|$ and 
their adjoints,
the unique functors $\mathrmbfit{D} : G^{\ast} \rightarrow \mathrmbf{C}$ 
satisfing $\eta_{G} \circ |\mathrmbfit{D}| = D$.
That is, there is an adjunction $(\mbox{-})^{\ast} \dashv |\mbox{-}|$ between 
the path functor ${\ast} : \mathrmbf{Gph} \rightarrow \mathrmbf{Cat}$ and 
the underlying graph functor $|\mbox{-}| : \mathrmbf{Cat} \rightarrow \mathrmbf{Gph}$.}

\newpage
\paragraph{The sum structure functor}

\begin{figure}
\begin{center}
\begin{tabular}{@{\hspace{50pt}}c@{\hspace{190pt}}c@{\hspace{60pt}}}
\\ 
& \\
\setlength{\unitlength}{0.8pt}
\begin{picture}(120,80)(0,0)
\put(0,80){\makebox(0,0){${\langle{\Sigma_{1},M_{1},T_{1}}\rangle}$}}
\put(120,110){\makebox(0,0){$
\overset{\textstyle\underbrace{\mathrmbfit{nat}(\mathrmbfit{sum}(\Sigma_{1},M_{1},T_{1}))}}
{{\langle{\Sigma_{1},M_{1}{+}\mathrmbfit{ext}_{\Sigma_{1}}(T_{1})
,\mathrmbfit{int}_{\Sigma_{1}}(M_{1}{+}\mathrmbfit{ext}_{\Sigma_{1}}(T_{1}))}\rangle}}
$}}
\put(120,-12){\makebox(0,0){$
\overset{\textstyle\underbrace{{\langle{\Sigma_{2},M_{2},\mathrmbfit{int}_{\Sigma_{2}}(M_{2})}\rangle}}}
{\mathrmbfit{nat}(\Sigma_{2},M_{2})}
$}}
\put(50,30){\makebox(0,0)[r]{\footnotesize{${\langle{\sigma,f}\rangle}$}}}
\put(60,85){\makebox(0,0){\footnotesize{$
\eta_{{\langle{\Sigma_{1},M_{1},T_{1}}\rangle}}
$}}}
\put(128,45){\makebox(0,0)[l]{\footnotesize{$
\overset{\underbrace{\mathrmbfit{nat}(\sigma,{!})}}{\sigma}
$}}}
\put(40,80){\vector(1,0){55}}
\put(18,68){\vector(3,-2){75}}
\put(120,65){\vector(0,-1){45}}
\put(185,0){\begin{picture}(0,80)(0,0)
\put(120,92){\makebox(0,0){$
\overset{\underbrace{\mathrmbfit{sum}(\Sigma_{1},M_{1},T_{1})}}
{{\langle{\Sigma_{1},M_{1}{+}\mathrmbfit{ext}_{\Sigma_{1}}(T_{1})}\rangle}}
$}}
\put(120,0){\makebox(0,0){${\langle{\Sigma_{2},M_{2}}\rangle}$}}
\put(120,65){\vector(0,-1){45}}
\put(130,40){\makebox(0,0)[l]{\footnotesize{$
{\langle{\sigma,{!}_{\mathrmbfit{struc}(\sigma)(M_{2})}}\rangle}
$}}}
\end{picture}}
\end{picture}
&
\\ & \\ & \\ & \\ 
\end{tabular}
\end{center}
\caption{Sum-Natural Adjunction}
\label{sum-natural:adjunction}
\end{figure}

\mbox{ }\newline
A morphism in $\mathrmbf{Log}$, 
called a logic morphism,
is a pair $(\sigma,f) : (\Sigma_{1},M_{1},T_{1}) \rightarrow (\Sigma_{2},M_{2},T_{2})$, 
where $\sigma : \Sigma_{1} \rightarrow \Sigma_{2}$ is a language morphism, and
$f : (M_{1},T_{1}) \rightarrow 
\mathrmbfit{log}(\sigma)(M_{2},T_{2}) = (\mathrmbfit{struc}(\sigma)(M_{2}),\mathrmbfit{inv}(\sigma)(T_{2}))$
is a $\Sigma_{1}$-logic morphism; 
that is,
$f : M_{1} \rightarrow \mathrmbfit{struc}(\sigma)(M_{2})$ is a $\Sigma_{1}$-structure morphism, and 
$T_{1} \geq_{\Sigma_{1}} \mathrmbfit{inv}(\sigma)(T_{2})$ is a $\Sigma_{1}$-specification ordering.
Hence,
a logic morphism $(\sigma,f) : (\Sigma_{1},M_{1},T_{1}) \rightarrow (\Sigma_{2},M_{2},T_{2})$
consists of
a structure morphism $(\sigma,f) : (\Sigma_{1},M_{1}) \rightarrow (\Sigma_{2},M_{2})$ and
a specification morphism $\sigma : (\Sigma_{1},T_{1}) \rightarrow (\Sigma_{2},T_{2})$.

\begin{center}
\begin{tabular}{c}
\setlength{\unitlength}{0.9pt}
\begin{picture}(80,60)(0,0)
\put(40,60){\makebox(0,0){$\mathrmbf{Log}$}}
\put(0,0){\makebox(0,0){$\mathrmbf{Struc}$}}
\put(80,0){\makebox(0,0){$\mathrmbf{Spec}$}}
\put(40,-7){\makebox(0,0){\scriptsize{$\mathrmbfit{ext}$}}}
\put(27,28){\makebox(0,0)[l]{\scriptsize{$\mathrmbfit{pr}$}}}
\put(68,28){\makebox(0,0)[l]{\scriptsize{$\mathrmbfit{pr}$}}}
\put(5,34){\makebox(0,0)[r]{\scriptsize{$\mathrmbfit{sum}$}}}
\put(15,8){\makebox(0,0){\scriptsize{$+$}}}
\qbezier[26](22,53)(10,35)(-4,14)\put(-4,14){\vector(-2,-3){0}}
\put(60,0){\vector(-1,0){40}}
\put(30,45){\vector(-2,-3){22}}
\put(50,45){\vector(2,-3){22}}
\end{picture}
\end{tabular}
\end{center}

\newpage

\begin{meta-axiom}
{\normalsize$\blacksquare$}
For any language $\Sigma$ and for any $\Sigma$-logic
$\mathcal{L} = {\langle{M,T}\rangle}$
there is a $\Sigma$-structure
that is the sum $M +_{\Sigma} \mathrmbfit{ext}_{\Sigma}(T)$ in the fiber category $\mathrmbf{Struc}(\Sigma)$
of the component structure and the extent of the component specification.
Hence,
there is a sum structure lax fibered functor
that commutes with projection
$\mathrmbfit{sum} \circ \mathrmbfit{pr}_{\mathrmbf{Struc}} = \mathrmbfit{pr}_{\mathrmbf{Log}}$.
The sum structure functor is left adjoint to the natural logic functor
\footnotesize\[
{\langle{\mathrmbfit{sum},\mathrmbfit{nat},\eta,\varepsilon}\rangle} : \mathrmbf{Log} \rightarrow \mathrmbf{Struc}
\]\normalsize
We have the following adjunction properties:
\begin{center}
{\footnotesize$\begin{array}{rcl}
\mathrmbf{Log} 
\stackrel{{\langle{\overset{\scriptscriptstyle+}{\mathrmbfit{res}},\mathrmbfit{inc}}\rangle}}{\longrightarrow}
\mathrmbf{Snd} 
\stackrel{{\langle{\mathrmbfit{pr}_{0},\mathrmbfit{nat}}\rangle}}{\longrightarrow} 
\mathrmbf{Struc}
&=&
\mathrmbf{Log} 
\stackrel{{\langle{\mathrmbfit{sum},\mathrmbfit{nat}}\rangle}}{\longrightarrow} 
\mathrmbf{Struc}
\\
\mathrmbf{Snd} 
\stackrel{{\langle{\mathrmbfit{inc},\overset{\scriptscriptstyle\vee}{\mathrmbfit{res}}}\rangle}}{\longrightarrow}
\mathrmbf{Log} 
\stackrel{{\langle{\mathrmbfit{pr}_{0},{\bot}}\rangle}}{\longrightarrow} 
\mathrmbf{Struc}
&=&
\mathrmbf{Log} 
\stackrel{{\langle{\mathrmbfit{pr}_{0},\mathrmbfit{nat}}\rangle}}{\longrightarrow} 
\mathrmbf{Struc}
\\
\mathrmbf{Snd} 
\stackrel{{\langle{\mathrmbfit{inc},\overset{\scriptscriptstyle\vee}{\mathrmbfit{res}}}\rangle}}{\longrightarrow}
\mathrmbf{Log} 
\stackrel{{\langle{\mathrmbfit{sum},\mathrmbfit{nat}}\rangle}}{\longrightarrow} 
\mathrmbf{Struc}
&\cong&
\mathrmbf{Log} 
\stackrel{{\langle{\mathrmbfit{pr}_{0},\mathrmbfit{nat}}\rangle}}{\longrightarrow} 
\mathrmbf{Struc}
\end{array}$}
\end{center}
\end{meta-axiom}
This lax fibered functor is the homogenization of a lax indexed functor 
$\mathrmbfit{sum} : \mathrmbfit{log} \Rightarrow \mathrmbfit{struc}
: \mathrmbf{Lang}^{\mathrm{op}} \rightarrow \mathrmbf{Cat}$.
For any language $\Sigma$,
there is a fiber functor
$\mathrmbfit{sum}_{\Sigma} : \mathrmbfit{log}(\Sigma) \rightarrow \mathrmbfit{struc}(\Sigma)$
that maps 
a $\Sigma$-logic $(M,T) \in \mathrmbfit{log}(\Sigma)$ 
to a sum $\Sigma$-structure 
$\mathrmbfit{sum}_{\Sigma}(M,T) 
= M +_{\Sigma} \mathrmbfit{ext}_{\Sigma}(T)$,
so that
$T \geq_{\Sigma} \mathrmbfit{int}_{\Sigma}(\mathrmbfit{sum}_{\Sigma}(M,T))$, and
maps a $\Sigma$-logic morphism
$f : (M,T) \rightarrow (M',T')$ in $\mathrmbfit{log}(\Sigma)$,
where $f : M \rightarrow M'$ in $\mathrmbfit{struc}(\Sigma)$
and $T \geq_{\Sigma} T'$ in $\mathrmbfit{spec}(\Sigma)$,
to the sum $\Sigma$-structure morphism
$\mathrmbfit{sum}_{\Sigma}(f) 
= f +_{\Sigma} \mathrmbfit{ext}_{\Sigma}(T \geq_{\Sigma} T') :
\mathrmbfit{sum}_{\Sigma}(M,T) \rightarrow_{\Sigma} \mathrmbfit{sum}_{\Sigma}(M',T')$.
For any language morphism $\sigma : \Sigma_{1} \rightarrow \Sigma_{2}$ in $\mathrmbf{Lang}$,
summation 
includes
a natural transformation
$\mathrmbfit{sum}_{\sigma} :
\mathrmbfit{log}(\sigma) \circ \mathrmbfit{sum}_{\Sigma_{1}} \Rightarrow
\mathrmbfit{sum}_{\Sigma_{2}} \circ \mathrmbfit{struc}(\sigma)$
that can be regarded as a lax naturality commutative diagram
\begin{center}
\begin{tabular}{c}
\\
\setlength{\unitlength}{0.6pt}
\begin{picture}(100,80)(0,0)
\put(0,80){\makebox(0,0){\footnotesize{$\mathrmbfit{log}(\Sigma_{1})$}}}
\put(0,0){\makebox(0,0){\footnotesize{$\mathrmbfit{log}(\Sigma_{2})$}}}
\put(100,80){\makebox(0,0){\footnotesize{$\mathrmbfit{struc}(\Sigma_{1})$}}}
\put(100,0){\makebox(0,0){\footnotesize{$\mathrmbfit{struc}(\Sigma_{2})$}}}
\put(50,92){\makebox(0,0){\scriptsize{$\mathrmbfit{sum}_{\Sigma_{1}}$}}}
\put(50,-14){\makebox(0,0){\scriptsize{$\mathrmbfit{sum}_{\Sigma_{2}}$}}}
\put(-5,40){\makebox(0,0)[r]{\scriptsize{$\mathrmbfit{log}(\sigma)$}}}
\put(105,40){\makebox(0,0)[l]{\scriptsize{$\mathrmbfit{struc}(\sigma)$}}}
\put(0,20){\vector(0,1){40}}
\put(100,20){\vector(0,1){40}}
\put(35,80){\vector(1,0){24}}
\put(35,0){\vector(1,0){24}}
\put(35,50){\begin{picture}(0,0)(0,0)
\setlength{\unitlength}{0.5pt}
\put(31,-7){\makebox(0,0){{$\scriptscriptstyle\mathrmbfit{sum}_{\sigma}$}}}
\drawline(3,0)(23,-20)
\drawline(0,-3)(20,-23)
\qbezier(23,-16)(23,-20)(25,-25)
\qbezier(16,-23)(20,-23)(25,-25)
\end{picture}}
\end{picture}
\\ \\
\end{tabular}
\end{center}
where, for every target logic $(M_{2},T_{2}) \in \mathrmbfit{struc}(\Sigma_{2})$
\footnotesize\[
\hspace{-20pt}
\underset{
[\mathrmbfit{struc}(\sigma)(\iota_{M_{2}}),
\mathrmbfit{ext}_{\sigma}(T_{2}){\cdot}\mathrmbfit{struc}(\sigma)(\iota_{T_{2}})]
}
{\underbrace{\mathrmbfit{sum}_{\sigma}(M_{2},T_{2})}} :
\underset{
\mathrmbfit{struc}(\sigma)(M_{2}) +_{\Sigma_{1}} \mathrmbfit{ext}_{\Sigma_{1}}(\mathrmbfit{inv}(\sigma)(T_{2}))
}{\underbrace{\mathrmbfit{sum}_{\Sigma_{1}}(\mathrmbfit{log}(\sigma)(M_{2},T_{2}))}}
\rightarrow 
\underset{
\mathrmbfit{struc}(\sigma)(M_{2} +_{\Sigma_{2}} \mathrmbfit{ext}_{\Sigma_{2}}(T_{2}))
}
{\underbrace{\mathrmbfit{struc}(\sigma)(\mathrmbfit{sum}_{\Sigma_{2}}(M_{2},T_{2}))}}.
\]\normalsize
These diagrams vertically paste together under composition of language morphisms.

\begin{eg}
In the information flow logical environment $\mathtt{IF}$,
the sum functor
$\mathrmbfit{sum} : \mathrmbf{Log} \rightarrow \mathrmbf{Cls}$
serves as a sum structure lax fibered functor.
The functor $\mathrmbfit{sum}$
maps a logic ${\langle{X,Y,\models,\vdash}\rangle}$ to the sum classification
$\mathrmbfit{sum}(X,Y,\models,\vdash) 
= {\langle{X,Y,\models}\rangle} +_{Y} \mathrmbfit{ext}(Y,\vdash)
= {\langle{X^{\vdash},Y,\models}\rangle}$,
where
$X^{\vdash} = \{ n \in X \mid n \models (\Gamma,\Delta),\text{if}\;\Gamma\vdash\Delta \}$ 
consists of all instances that satisfy all sequents in the theory.
~\footnote{
Consider ways of connecting
the component classification ${\langle{X,Y,\models}\rangle}$
and the extent
$\mathrmbfit{ext}(Y,\vdash) = {\langle{{\wp}Y^{\vdash},Y, \ni_{Y}}\rangle}$ 
of the component theory,
where ${\wp}Y^{\vdash} \subseteq {\wp}Y$
is the collection of type subsets that satisfy the theory.
Let $\tau : X \rightarrow {\wp}Y$ be the state description function,
where $\tau(x) = \{ y \in Y \mid x \models y \}$.
Let $X^{\vdash} = \{ n \in X \mid n \models (\Gamma,\Delta),\text{all}\;\Gamma\vdash\Delta \}$ 
be the subset of all normal instances;
i.e., those that satisfy all sequents in the theory.
The state description function restricts as
$\tau : X^{\vdash} \rightarrow {\wp}Y^{\vdash}$.
One way for connection is with the opspan consisting of
the inclusion $Y$-fiber infomorphism (on the left) and 
the restricted state description $Y$-fiber infomorphism (on the right)
\[
{\langle{X,Y,\models}\rangle} 
\stackrel{\langle{\iota^Y_X,1_{Y}}\rangle}{\rightleftarrows} 
{\langle{X^{\vdash},Y,\models}\rangle}
\stackrel{\langle{\tau,1_{Y}}\rangle}{\leftrightarrows} 
{\langle{{\wp}Y^{\vdash},Y,\ni_{Y}}\rangle} 
\]
If there is another opspan
${\langle{X,Y,\models}\rangle} 
\stackrel{\langle{g_X,1_{Y}}\rangle}{\rightleftarrows} 
{\langle{X',Y,\models'}\rangle}
\stackrel{\langle{g_{\vdash},1_{Y}}\rangle}{\leftrightarrows} 
{\langle{{\wp}Y^{\vdash},Y,\ni_{Y}}\rangle}$
connecting these two classifications,
then
$g_X(x') \models y$ iff $x' \models' y$ iff $g_{\vdash}(x') \ni_{Y} y$
for any $x' \in X'$.
Hence,
$g_{\vdash}(x')$ is the state description of $g_X(x')$ and
$g_X(x') \in X^{\vdash}$ is a normal instance.
Hence,
$g_X$ factors through the normal instances
$g_X \cdot \iota^Y_X = g_X$
and
$g_X \cdot \tau = g_{\vdash}$.
So we have a (clearly unique) $Y$-fiber infomorphism
${\langle{X^{\vdash},Y,\models}\rangle}
\stackrel{\langle{g_X,1_{Y}}\rangle}{\rightleftarrows} 
{\langle{X',Y,\models'}\rangle}$
with
${\langle{\iota^Y_X,1_{Y}}\rangle} \cdot {\langle{g_X,1_{Y}}\rangle}
= {\langle{g_X,1_{Y}}\rangle}$
and
${\langle{\tau,1_{Y}}\rangle} \cdot {\langle{g_X,1_{Y}}\rangle}
= {\langle{g_{\vdash},1_{Y}}\rangle}$.
This means that
${\langle{X^{\vdash},Y,\models}\rangle}
= {\langle{X,Y,\models}\rangle} +_{Y} {\langle{{\wp}Y^{\vdash},Y,\ni_{Y}}\rangle}$
is a $Y$-fiber binary sum.
We are not saying that all binary sums exist.
Also, note that
${\langle{{\wp}Y,Y,\ni_{Y}}\rangle} 
\stackrel{\langle{\tau,1_{Y}}\rangle}{\rightleftarrows} 
{\langle{X,Y,\models}\rangle}$
is the unique $Y$-fiber infomorphism
from the initial $Y$-fiber object
$\mathrmbfit{ext}(Y,\emptyset) = \mathrmbfit{ext}(Y,\top) = {\langle{{\wp}Y,Y,\ni_{Y}}\rangle}$.
``$\mathrmbfit{ext}_{Y}$ preserves initial object.''}
The functor $\mathrmbfit{sum}$
maps a logic infomorphism 
${\langle{g,f}\rangle} : 
{\langle{X_{1},Y_{1},\models_{1},\vdash_{1}}\rangle} \rightleftarrows {\langle{X_{2},Y_{2},\models_{2},\vdash_{2}}\rangle}$ 
to the sum structure infomorphism
$\mathrmbfit{sum}(g,f) = {\langle{g^{\vdash},f}\rangle} : 
{\langle{X^{\vdash_{1}}_{1},Y_{1},\models_{1}}\rangle} \rightleftarrows 
{\langle{X^{\vdash_{2}}_{2},Y_{2},\models_{2}}\rangle}$, 
where
$g^{\vdash} : X^{\vdash_{2}}_{2} \rightarrow X^{\vdash_{1}}_{1}$ 
is the sum restriction
of the instance function
$g : X_{2} \rightarrow X_{1}$. 
Clearly,
the sum operator $\mathrmbfit{sum}$ is functorial,
since restriction of a composite function is the composite of the component restrictions.
The sum functor commutes with projection
$\mathrmbfit{sum} \circ \mathrmbfit{typ} = \mathrmbfit{typ}$.
The lax fibered functor $\mathrmbfit{sum} : \mathrmbf{Log} \rightarrow \mathrmbf{Cls}$ 
is the homogenization of 
a sum lax indexed functor
$\mathrmbfit{sum} : \mathrmbfit{log} \Rightarrow \mathrmbfit{cls}
: \mathrmbf{Set}^{\mathrm{op}} \rightarrow \mathrmbf{Cat}$.
For any set $Y$,
there is a fiber functor
$\mathrmbfit{sum}_{Y} : \mathrmbfit{log}(Y) \rightarrow \mathrmbfit{cls}(Y)$
that maps 
a $Y$-logic $(X,\models,\vdash) \in \mathrmbfit{log}(Y)$ 
to a sum $Y$-classification 
${\langle{X^{\vdash},\models}\rangle} \in \mathrmbfit{cls}(Y)$
with all instances in $X^{\vdash}$ satisfying ${\vdash}$, and
maps a $Y$-logic infomorphism
$g : (X,\models,\vdash) \rightarrow (X',\models',\vdash')$ in $\mathrmbfit{log}(Y)$
to the sum $Y$-classification infomorphism
$g^{\vdash} : 
{\langle{X^{\vdash},Y,\models}\rangle} \rightleftarrows {\langle{X'^{\vdash'},Y',\models'}\rangle}$, 
For any (type) function $f : Y_{1} \rightarrow \Sigma_{2}$,
there is a natural transformation
\begin{center}
\begin{tabular}{c}
\\
\setlength{\unitlength}{0.6pt}
\begin{picture}(100,80)(0,0)
\put(0,80){\makebox(0,0){\footnotesize{$\mathrmbfit{log}(Y_{1})$}}}
\put(100,80){\makebox(0,0){\footnotesize{$\mathrmbfit{cls}(Y_{1})$}}}
\put(0,0){\makebox(0,0){\footnotesize{$\mathrmbfit{log}(Y_{2})$}}}
\put(100,0){\makebox(0,0){\footnotesize{$\mathrmbfit{cls}(Y_{2})$}}}
\put(50,92){\makebox(0,0){\scriptsize{$\mathrmbfit{sum}_{Y_{1}}$}}}
\put(105,40){\makebox(0,0)[l]{\scriptsize{$\mathrmbfit{cls}(f)$}}}
\put(-5,40){\makebox(0,0)[r]{\scriptsize{$\mathrmbfit{log}(f)$}}}
\put(50,-14){\makebox(0,0){\scriptsize{$\mathrmbfit{sum}_{Y_{2}}$}}}
\put(0,20){\vector(0,1){40}}
\put(100,20){\vector(0,1){40}}
\put(35,80){\vector(1,0){24}}
\put(35,0){\vector(1,0){24}}
\put(35,50){\begin{picture}(0,0)(0,0)
\setlength{\unitlength}{0.5pt}
\put(31,-7){\makebox(0,0){{$\scriptscriptstyle\mathrmbfit{sum}_{f}$}}}
\drawline(3,0)(23,-20)
\drawline(0,-3)(20,-23)
\qbezier(23,-16)(23,-20)(25,-25)
\qbezier(16,-23)(20,-23)(25,-25)
\end{picture}}
\end{picture}
\\ \\
\end{tabular}
\end{center}
where
for every $Y_{2}$-logic $(X_{2},\models_{2},\vdash_{2}) \in \mathrmbfit{log}(Y_{2})$,
the $Y_{1}$-infomorphism
\footnotesize\[
\mathrmbfit{sum}_{f}(X_{2},\models_{2},\vdash_{2}) :
\underset{(X_{2}^{f^{-1}(\vdash_{2})},\models_{2}^{f})} 
{\underbrace{\mathrmbfit{sum}_{Y_{1}}(\mathrmbfit{log}(f)(X_{2},\models_{2},\vdash_{2}))}}
\rightleftarrows 
\underset{(X_{2}^{\vdash_{2}},\models_{2}^{f})}
{\underbrace{\mathrmbfit{cls}(f)(\mathrmbfit{sum}_{Y_{2}}(X_{2},\models_{2},\vdash_{2}))}}
\]\normalsize
has inclusion instance function
$X_{2}^{\vdash_{2}} \hookrightarrow X_{2}^{f^{-1}(\vdash_{2})}$,
since any normal instance of the logic 
${\langle{X_{2},\models_{2},\vdash_{2} }\rangle}$
is a normal instance of the logic 
${\langle{X_{2},\models_{2}^{f},f^{-1}(\vdash_{2})}\rangle}$.
These diagrams vertically paste together under composition of functions.
\end{eg}

\begin{eg}
In the sketch logical environment $\mathtt{Sk}$ (used in the metatheory of sketches),
the sum functor
$\mathrmbfit{sum} : \mathrmbf{Log} \rightarrow \mathrmbf{Dgm}$
serves as a sum structure lax fibered functor.
The functor $\mathrmbfit{sum}$
maps a logic ${\langle{G,D,\mathrmbf{C},E}\rangle}$ to the sum diagram
$\mathrmbfit{sum}(G,D,\mathrmbf{C},E) 
= {\langle{G,(D,\mathrmbf{C})+_{G}\mathrmbfit{ext}_{G}(E)}\rangle}$
~\footnote{
$\mathrmbfit{ext}_{G}(E) 
= ({\between}_{G},G/E)
= (\raisebox{1.5pt}[0pt]{${\scriptstyle\{\}}_{G}$}{\circ}|\raisebox{1.5pt}[0pt]{${\scriptstyle[]}_{E}$}|,G/E)$
is the free diagram generated by the $G$-specification $E$.
$E$ induces a congruence $\ddot{E}$ on the graph $G$,
which in turn induces a ``quotient ''category $G/E$
with canonical functor
$\raisebox{1.5pt}[0pt]{${\scriptstyle[]}_{E}$} : G^{\ast} \rightarrow G/E$.
Any diagram
$D = {\scriptstyle\{\}}_{G} \circ \mathrmbfit{D} : G \rightarrow |\mathrmbf{C}|$ 
with iteration $\mathrmbfit{D}^{\ast} : G^{\ast} \rightarrow \mathrmbf{C}$
that satisfies the specification $E$ 
\[
(D,\mathrmbf{C}) \models_{G} E 
\;\;\text{iff}\;\;
E \geq_{G} \mathrmbf{int}_{G}(D,\mathrmbf{C}) 
\;\;\text{iff}\;\;
\mathrmbf{ext}_{G}(E) \stackrel{\exists{!}}{\rightarrow} (D,\mathrmbf{C})
\;\;\text{iff}\;\;
\mathrmbfit{D} : G^{\ast} \rightarrow \mathrmbf{C}\;\;\text{respects}\;\;\ddot{E},
\]
uniquely factors through $G/E$,
${\between}_{G} \circ |{\exists{!}}_{G}| = D$,
via a functor ${\exists{!}}_{G} : G/E \rightarrow \mathrmbf{C}$.
},
\begin{center}
\begin{tabular}{c}
\\
\setlength{\unitlength}{0.5pt}
\begin{picture}(200,100)(0,0)
\put(0,80){\makebox(0,0){\footnotesize{$G$}}}
\put(100,80){\makebox(0,0){\footnotesize{$|G^{\ast}|$}}}
\put(204,80){\makebox(0,0){\footnotesize{$|G/E|$}}}
\put(0,0){\makebox(0,0){\footnotesize{$|\mathrmbf{C}|$}}}
\put(220,0){\makebox(0,0){\footnotesize{$|\mathrmbf{C}{+}_{G}{G/E}|$}}}
\put(-8,40){\makebox(0,0)[r]{\scriptsize{$D$}}}
\put(208,40){\makebox(0,0)[l]{\scriptsize{$|\iota_{E}|$}}}
\put(100,120){\makebox(0,0){\scriptsize{${\scriptstyle\between}_{G}$}}}
\put(50,92){\makebox(0,0){\scriptsize{${\scriptstyle\{\}}_{G}$}}}
\put(146,92){\makebox(0,0){\scriptsize{$|{\scriptstyle[]}_{E}|$}}}
\put(100,-14){\makebox(0,0){\scriptsize{$|\iota_{\mathrmbf{C}}|$}}}
\put(90,30){\makebox(0,0){\scriptsize{${\Diamond}_{E}$}}}
\put(30,80){\vector(1,0){40}}
\put(123,80){\vector(1,0){40}}
\put(26,0){\vector(1,0){140}}
\put(25,63){\vector(3,-1){142}}
\put(0,60){\vector(0,-1){40}}
\put(200,60){\vector(0,-1){40}}
\put(100,100){\oval(190,20)[t]}\put(195,96){\vector(0,-1){0}}
\end{picture}
\\ \\
\end{tabular}
\end{center}
The functor $\mathrmbfit{sum}$ maps a logic infomorphism 
${\langle{H,\mathrmbfit{F}}\rangle} : 
{\langle{G_{1},D_{1},\mathrmbf{C}_{1},E_{1}}\rangle} \rightleftarrows {\langle{G_{2},D_{2},\mathrmbf{C}_{2},E_{2}}\rangle}$ 
with diagram morphism
${\langle{H,\mathrmbfit{F}}\rangle} : 
{\langle{G_{1},D_{1},\mathrmbf{C}_{1}}\rangle} \rightleftarrows {\langle{G_{2},D_{2},\mathrmbf{C}_{2}}\rangle}$
and specification morphism
$H : {\langle{G_{1},E_{1}}\rangle} \rightleftarrows {\langle{G_{2},E_{2}}\rangle}$ 
with unique factor functor
$\mathrmbfit{H} : G_{1}/E_{1} \rightarrow G_{2}/E_{2}$,
so that
$H \circ D_{2} = D_{1} \circ \mathrmbfit{F}$,
and
$H \circ {\scriptstyle\between}_{G_{2}}
= {\scriptstyle\between}_{G_{1}} \circ \mathrmbfit{H}$,
to the sum diagram morphism
${\langle{H,\mathrmbfit{F}{+}\mathrmbfit{H}}\rangle} : 
{\langle{G_{1},\Diamond_{E_{1}},\mathrmbf{C}_{1}/E_{1}}\rangle} 
\rightleftarrows 
{\langle{G_{2},\Diamond_{E_{2}},\mathrmbf{C}_{2}/E_{2}}\rangle}$,
where
$\mathrmbfit{F}{+}\mathrmbfit{H} : 
\mathrmbf{C}_{1}{+}_{G_{1}}{G_{1}/E_{1}} \rightarrow \mathrmbf{C}_{2}{+}_{G_{2}}{G_{2}/E_{2}}$
is the unique functor
such that 
$\iota_{\mathrmbf{C}_{1}} \circ (\mathrmbfit{F}{+}\mathrmbfit{H}) 
= \mathrmbfit{F} \circ \iota_{\mathrmbf{C}_{2}}$
and
$\iota_{E_{1}} \circ (\mathrmbfit{F}{+}\mathrmbfit{H}) 
= \mathrmbfit{H} \circ \iota_{E_{2}}$.
\begin{center}
\begin{tabular}{c}
\\
\setlength{\unitlength}{0.5pt}
\begin{picture}(200,110)(0,-14)
\put(0,80){\makebox(0,0){\footnotesize{$G_{1}$}}}
\put(100,80){\makebox(0,0){\footnotesize{$|\mathrmbf{C}_{1}|$}}}
\put(228,80){\makebox(0,0){\footnotesize{$|\mathrmbf{C}_{1}{+}_{G_{1}}{G_{1}/E_{1}}|$}}}
\put(0,0){\makebox(0,0){\footnotesize{$G_{2}$}}}
\put(100,0){\makebox(0,0){\footnotesize{$|\mathrmbf{C}_{2}|$}}}
\put(228,0){\makebox(0,0){\footnotesize{$|\mathrmbf{C}_{2}{+}_{G_{2}}{G_{2}/E_{2}}|$}}}
\put(-8,40){\makebox(0,0)[r]{\scriptsize{$H$}}}
\put(108,40){\makebox(0,0)[l]{\scriptsize{$|\mathrmbfit{F}|$}}}
\put(208,40){\makebox(0,0)[l]{\scriptsize{$|\mathrmbfit{F}{+}\mathrmbfit{H}|$}}}
\put(50,92){\makebox(0,0){\scriptsize{$D_{1}$}}}
\put(50,-14){\makebox(0,0){\scriptsize{$D_{2}$}}}
\put(146,92){\makebox(0,0){\scriptsize{$|\iota_{\mathrmbf{C}_{1}}|$}}}
\put(146,-14){\makebox(0,0){\scriptsize{$|\iota_{\mathrmbf{C}_{2}}|$}}}
\put(100,120){\makebox(0,0){\scriptsize{$\Diamond_{E_{1}}$}}}
\put(100,-45){\makebox(0,0){\scriptsize{$\Diamond_{E_{2}}$}}}
\put(30,80){\vector(1,0){40}}
\put(123,80){\vector(1,0){40}}
\put(30,0){\vector(1,0){40}}
\put(123,0){\vector(1,0){40}}
\put(0,60){\vector(0,-1){40}}
\put(100,60){\vector(0,-1){40}}
\put(200,60){\vector(0,-1){40}}
\put(100,100){\oval(190,20)[t]}\put(195,96){\vector(0,-1){0}}
\put(100,-20){\oval(190,20)[b]}\put(195,-16){\vector(0,1){0}}
\end{picture}
\\ \\
\end{tabular}
\end{center}
Clearly,
the sum operator $\mathrmbfit{sum}$ is functorial,
since mediating functor of a composite is the composite of the component mediating functors.
The sum functor commutes with projection
$\mathrmbfit{sum} \circ \mathrmbfit{typ} = \mathrmbfit{typ}$.
\end{eg}

\newpage
\subsubsection{Relating Extent to Sum.} 

An object $A$ in a category $\mathrmbf{C}$ is cohesive
when it is its own binary copower:
$A \cong 2{\cdot}A = A{+}A$.
This means that any two parallel $\mathrmbf{C}$-morphisms
$f,g : A \rightarrow B$
are identical $f = g$.
Clearly,
initial objects are cohesive~\footnote{If $\mathrmbf{C}$ has an initial object $0_{\mathrmbf{C}}$,
then any two parallel $\mathrmbf{C}$-morphisms $f,g : 0_{\mathrmbf{C}} \rightarrow B$ are identical $f = g$,
since there is a unique such morphism ${!}_{B} : 0_{\mathrmbf{C}} \rightarrow B$.}.
If $A$ is cohesive in $\mathrmbf{C}$,
then existence of a (necessarily unique) $\mathrmbf{C}$-morphism $f : A \rightarrow B$
means that $B$ is the $\mathrmbf{C}$-coproduct $A + B$
with coproduct injections
\footnotesize\[
A \stackrel{f}{\longrightarrow} B \stackrel{1_{B}}{\longleftarrow} B,
\]\normalsize
since cohesion of $A$ implies that
for any two $\mathrmbf{C}$-morphisms
$g : A \rightarrow C$ and $h : B \rightarrow C$,
we have the identity $g = f \cdot h : A \rightarrow C$.

\newpage

\begin{meta-axiom}
{\normalsize$\blacksquare$}
For any language $\Sigma$,
and for any $\Sigma$-specification $T$,
the $\Sigma$-structure
$\mathrmbfit{ext}_{\Sigma}(T)$ is cohesive
in the fiber category $\mathrmbf{Struc}(\Sigma)$.
\end{meta-axiom}

\begin{itemize}
\item 
The meta-axiom,
which asserts existence of the left adjoint initial functor 
$\mathrmbfit{0} : \mathrmbf{Lang} \rightarrow \mathrmbf{Struc}$,
guarantees that for any language $\Sigma$
the fiber category
$\mathrmbf{Struc}(\Sigma)$  
has the initial (hence, cohesive) structure $\mathrmbfit{0}(\Sigma) \in \mathrmbf{Struc}(\Sigma)$.
\item 
The meta-axiom, 
which asserts existence of the left adjoint extent functor 
$\mathrmbfit{ext} : \mathrmbf{Spec} \rightarrow \mathrmbf{Struc}$,
guarantees that 
for any language $\Sigma$
and any $\Sigma$-specification $T$
the fiber subcategory
$\ddddot{\mathrmbfit{ext}}_{\Sigma}(T) \subseteq \mathrmbf{Struc}(\Sigma)$  
has the initial (hence, cohesive) structure $\mathrmbfit{ext}_{\Sigma}(T) \in \ddddot{\mathrmbfit{ext}}_{\Sigma}(T)$.
\item 
The above meta-axiom
guarantees that for any language $\Sigma$
the structures $\mathrmbfit{ext}_{\Sigma}(T)$,
although not initial,
are at least cohesive in $\mathrmbf{Struc}(\Sigma)$.
\end{itemize}
For any sound logic ${\langle{\Sigma,M,T}\rangle}$,
since the $\Sigma$-structure $\mathrmbfit{ext}_{\Sigma}(T)$ is cohesive in $\mathrmbf{Struc}(\Sigma)$ and
$\gtrdot_{\Sigma} : \mathrmbfit{ext}_{\Sigma}(T) \rightarrow M$ is a (necessarily unique) $\Sigma$-structure morphism,
$M$ is the $\mathrmbf{Struc}(\Sigma)$-coproduct $M = \mathrmbfit{ext}_{\Sigma}(T) + M$
~\footnote{In particular,
for any $\Sigma$-structure $M$
the $\Sigma$-structure 
${M}^{\scriptscriptstyle\circ} = \mathrmbfit{ext}_{\Sigma}(\mathrmbfit{int}_{\Sigma}(M))$
is cohesive in $\mathrmbf{Struc}(\Sigma)$.
Hence,
the pair of $\Sigma$-structures
${M}^{\scriptscriptstyle\circ}$ and $M$
have the $\mathrmbf{Struc}(\Sigma)$-coproduct $M = {M}^{\scriptscriptstyle\circ} + M$
with coproduct injections
${M}^{\scriptscriptstyle\circ} \stackrel{\varepsilon_{\Sigma}(M)}{\longrightarrow} M \stackrel{1_{M}}{\longleftarrow} M$. 
}
with coproduct injections
$\mathrmbfit{ext}_{\Sigma}(T) \stackrel{\gtrdot_{\Sigma}}{\longrightarrow} M \stackrel{1_{M}}{\longleftarrow} M$.

\begin{itemize}

\item 
We have the adjunction
\footnotesize\[
{\langle{\mathrmbfit{sum}_{\Sigma} \dashv \mathrmbfit{nat}_{\Sigma},\eta_{\Sigma},1}\rangle} : 
\mathrmbf{Log}(\Sigma) \rightarrow \mathrmbf{Struc}(\Sigma). 
\]\normalsize
since
$\mathrmbfit{sum}_{\Sigma}(\widehat{\mathrmbfit{nat}}_{\Sigma}(M)) 
= {M}^{\scriptscriptstyle\circ} + M = M$
implies
$\widehat{\mathrmbfit{nat}} \circ \mathrmbfit{sum} = 1_{\mathrmbf{Struc}}$.
This gives the trivial counit component for the adjunction.
\newline

\item 
We have the functor composition
$\widehat{\mathrmbfit{th}} \circ \mathrmbfit{sum} = \mathrmbfit{ext}$, 
\newline
since
$\mathrmbfit{sum}_{\Sigma}(\widehat{\mathrmbfit{th}}_{\Sigma}(T)) 
= \mathrmbfit{ext}_{\Sigma}(T) + \mathrmbfit{ext}_{\Sigma}(T) = \mathrmbfit{ext}_{\Sigma}(T)$.
\newline

\item 
Thus,
we have the adjunction composition (factorization of truth)
\footnotesize\[
{\langle{\widehat{\mathrmbfit{th}} \dashv \mathrmbfit{pr}_{1}}\rangle}
\circ
{\langle{\mathrmbfit{sum} \dashv \mathrmbfit{nat}}\rangle}
=
{\langle{\mathrmbfit{ext} \dashv \mathrmbfit{int}}\rangle}
\]\normalsize

\end{itemize}

We have the following adjunction properties:
\begin{center}
{\footnotesize$\begin{array}{rcl}
\mathrmbf{Log} 
\stackrel{{\langle{\overset{\scriptscriptstyle+}{\mathrmbfit{res}},\mathrmbfit{inc}}\rangle}}{\longrightarrow}
\mathrmbf{Snd} 
\stackrel{{\langle{\mathrmbfit{pr}_{0},\mathrmbfit{nat}}\rangle}}{\longrightarrow} 
\mathrmbf{Struc}
&=&
\mathrmbf{Log} 
\stackrel{{\langle{\mathrmbfit{sum},\mathrmbfit{nat}}\rangle}}{\longrightarrow} 
\mathrmbf{Struc}
\\
\mathrmbf{Snd} 
\stackrel{{\langle{\mathrmbfit{inc},\overset{\scriptscriptstyle\vee}{\mathrmbfit{res}}}\rangle}}{\longrightarrow}
\mathrmbf{Log} 
\stackrel{{\langle{\mathrmbfit{pr}_{0},{\bot}}\rangle}}{\longrightarrow} 
\mathrmbf{Struc}
&=&
\mathrmbf{Log} 
\stackrel{{\langle{\mathrmbfit{pr}_{0},\mathrmbfit{nat}}\rangle}}{\longrightarrow} 
\mathrmbf{Struc}
\\
\mathrmbf{Snd} 
\stackrel{{\langle{\mathrmbfit{inc},\overset{\scriptscriptstyle\vee}{\mathrmbfit{res}}}\rangle}}{\longrightarrow}
\mathrmbf{Log} 
\stackrel{{\langle{\mathrmbfit{sum},\mathrmbfit{nat}}\rangle}}{\longrightarrow} 
\mathrmbf{Struc}
&\cong&
\mathrmbf{Log} 
\stackrel{{\langle{\mathrmbfit{pr}_{0},\mathrmbfit{nat}}\rangle}}{\longrightarrow} 
\mathrmbf{Struc}
\end{array}$}
\end{center}

\begin{eg}
In the information flow logical environment $\mathtt{IF}$,
for a classification ${\langle{X,Y,\models}\rangle}$,
two instances $x,x' \in X$ are indistinguishable,
denoted $x \equiv x'$,
when they have the same intent.
Any two infomorphisms 
${\langle{g,f}\rangle},{\langle{g',f}\rangle} : 
{\langle{X_{1},Y_{1},\models_{1}}\rangle} \rightarrow {\langle{X_{2},Y_{2},\models_{2}}\rangle}$
over a (type) function $f : Y_{1} \rightarrow Y_{2}$,
have ``equivalent'' (instance) functions $g,g' : X_{2} \rightarrow X_{2}$,
since $g(x_{2}) \models_{1} y_{1}$ iff $x_{2} \models_{2} f(y_{1})$ iff $g(x_{1}) \models_{1} y_{1}$
for all target instances $x_{2} \in X_{2}$.
A classification ${\langle{X,Y,\models}\rangle}$
is separated when there is no pair of indistinguishable instances in $X$.
The type powerset classification ${\langle{{\wp}(Y),Y,\ni_{Y}}\rangle}$ 
(and any of its subclassifications)
is separated.
There is at most one infomorphism 
over a (type) function $f : Y_{1} \rightarrow Y_{2}$
with a separated source classification.
In particular,
an infomorphism
${\langle{g,f}\rangle}
: {\langle{{\wp}(Y_{1}),Y_{1},\ni_{Y{1}}}\rangle} \rightleftarrows {\langle{{\wp}(Y_{2}),Y_{2},\ni_{Y_{2}}}\rangle}$
between type powerset classifications,
must have the unique instance function $g = f^{-1}$,
since $g(X_{2}) = \{ y \in Y \mid X_{2} \ni_{Y{2}} f(y) \} = f^{-1}(X_{2})$ for every $X_{2} \in {\wp}(Y_{2})$.
Thus,
for any type function $f : Y_{1} \rightarrow Y_{2}$,
there is precisely one infomorphism
$g : {\langle{{\wp}(Y_{1}),\ni_{Y{1}}}\rangle} \rightarrow 
{\langle{X_{2},Y_{1},\models_{f}}\rangle}=\mathrmbfit{cls}(f)({\wp}(Y_{2}),\ni_{Y_{2}})$
in the fiber category $\mathrmbfit{cls}(Y_{1})$,
namely $g = f^{-1}$.
More generally,
for any type set $Y$,
any subclassification of the type powerset classification ${\langle{{\wp}(Y),Y,\ni_{Y}}\rangle}$ 
is separated,
and hence cohesive, 
in the fiber category $\mathrmbfit{cls}(Y_{1})$.
\end{eg}